\documentclass[11pt,letter]{article}

\usepackage{tikz}
\usetikzlibrary{decorations.pathmorphing}
\usepackage{subcaption}
\usepackage{fullpage}
\usepackage{xurl}
\usepackage{hyperref}  
\usepackage{xspace}
 \usepackage{mathrsfs}
\usepackage{amsthm}
\usepackage[fleqn]{amsmath}
\usepackage{amsfonts}
\usepackage{cleveref} 
\crefname{observation}{observation}{observations}
\usepackage{graphicx,color}

\graphicspath{{Figures/}}

\usepackage{boxedminipage}
\usepackage{framed}
\usepackage{mathtools}
\usepackage{thmtools}
\usepackage{thm-restate}
\usepackage{xspace}
\usepackage{footnote}
\usepackage{nicefrac}
\usepackage{todonotes}
\usepackage{footnote}
 \usepackage{xifthen}
 \usepackage{tabularx}
\usepackage{tikz} 
\usetikzlibrary{calc,math,patterns,shapes,backgrounds}
\usetikzlibrary{positioning,decorations.pathreplacing}
\usetikzlibrary{calc}
\tikzstyle{path} = [color=black,opacity=.30,line cap=round, line join=round, line width=10pt]
\usepackage{tcolorbox}
\usepackage{amsfonts}
\usepackage{multirow}
\usepackage[inline]{enumitem}

\DeclareMathOperator{\Int}{Int}

\newcommand{\C}[1]{\ensuremath{\operatorname{\mathcal {#1} }}} 
\newcommand{\SC}[1]{\ensuremath{\operatorname{\mathscr {#1} }}} 

\newcommand{\ceil}[1]{\ensuremath{\left\lceil {#1} \right\rceil}} 

\usepackage{amsmath,amssymb,amsxtra}
\usepackage{todonotes}

	\DeclareMathOperator{\operatorClassFPT}{{\sf FPT}\xspace}
	\newcommand{\classFPT}{\ensuremath{\operatorClassFPT}\xspace}

	\DeclareMathOperator{\operatorClassXP}{{\sf X}P\xspace}
	\newcommand{\classXP}{\ensuremath{\operatorClassXP}\xspace}

\newcommand{\NP}{\textsf{NP}\xspace}
\newcommand{\NPH}{\textsf{NP}-hard\xspace}

\newcommand{\mcm}[3]{\newcommand{#1}[#2]{{\ensuremath{#3}}}} 
\mcm{\Nbb}{0}{\mathbb{N}}
\mcm{\Zbb}{0}{\mathbb{Z}}
\mcm{\Rbb}{0}{\mathbb{R}}
\mcm{\Cbb}{0}{\mathbb{C}}
\mcm{\Qbb}{0}{\mathbb{Q}}
\mcm{\Acal}{0}{\cal A}
\mcm{\Bcal}{0}{\cal B}
\mcm{\Ccal}{0}{\cal C}
\mcm{\Dcal}{0}{\cal D}
\mcm{\Ecal}{0}{\cal E}
\mcm{\Fcal}{0}{\cal F}
\mcm{\Gcal}{0}{\cal G}
\mcm{\Hcal}{0}{\cal H}
\mcm{\Ical}{0}{\cal I}
\mcm{\Jcal}{0}{\cal J}
\mcm{\Kcal}{0}{\cal K}
\mcm{\Lcal}{0}{\cal L}
\mcm{\Mcal}{0}{\cal M}
\mcm{\Ncal}{0}{\cal N}
\mcm{\Ocal}{0}{{\cal O}}
\mcm{\Pcal}{0}{{\cal P}}
\mcm{\Qcal}{0}{{\cal Q}}
\mcm{\Rcal}{0}{{\cal R}}
\mcm{\Scal}{0}{{\cal S}}
\mcm{\Tcal}{0}{{\cal T}}
\mcm{\Ucal}{0}{{\cal U}}
\mcm{\Vcal}{0}{{\cal V}}
\mcm{\Wcal}{0}{{\cal W}}
\mcm{\Xcal}{0}{{\cal X}}
\mcm{\Ycal}{0}{{\cal Y}}
\mcm{\Zcal}{0}{{\cal Z}}

\newcommand{\yesinstances}{\yesinstances\xspace}

\newcommand{\pd}[1]{solution descendant}
\newcommand{\pds}[1]{solution descendants}
\newcommand{\pdr}{representative}
\newcommand{\pdrs}{representatives}
\newcommand{\maxddisjointcoveragefull}{\textsc{Diverse Shortest Cycle Coverage}\xspace}
\newcommand{\maxddisjointcoverage}{{\sc DSCC\xspace}}

\newcommand{\pmaxddisjointcoverage}{{\sc Planar DSC}}
\newcommand{\dsp}{\textsc{Disjoint Shortest Paths}\xspace}

\newcommand{\F}{{\sf e}}

\newcommand{\N}{{\sf n}}
\newcommand{\tgamma}{\wtilde{\gamma}}
\newcommand{\tbeta}{\wtilde{\beta}}

\newcommand{\tpsi}{\wtilde{\psi}}
\newcommand{\ttau}{\wtilde{\tau}}

\newcommand{\wth}{\wtilde{h}}

\theoremstyle{plain}

\newtheorem{observation}{Observation}[section]
\newtheorem{corollary}{Corollary}[section] 
\newtheorem{lemma}{Lemma}[section] 
 
\newtheorem{proposition}{Proposition}[section] 

\theoremstyle{definition}
\newtheorem{definition}{Definition}[section]

\DeclareMathOperator*{\argmax}{arg\,max}

\newtcolorbox{remarkbox}
{
colback=blue!3,       
  colframe=blue!30,     
  coltext=black
}

\newlength{\RoundedBoxWidth}
\newsavebox{\GrayRoundedBox}
\newenvironment{GrayBox}[1]%
   {\setlength{\RoundedBoxWidth}{.93\textwidth}
    \def\boxheading{#1}
    \begin{lrbox}{\GrayRoundedBox}
       \begin{minipage}{\RoundedBoxWidth}}%
   {   \end{minipage}
    \end{lrbox}
    \begin{center}
    \begin{tikzpicture}%
       \node(Text)[draw=black!20,fill=white,rounded corners,%
             inner sep=2ex,text width=\RoundedBoxWidth]%
             {\usebox{\GrayRoundedBox}};
        \coordinate(x) at (current bounding box.north west);
        \node [draw=white,rectangle,inner sep=3pt,anchor=north west,fill=white] 
        at ($(x)+(6pt,.75em)$) {\boxheading};
    \end{tikzpicture}
    \end{center}}     

\newenvironment{defproblemx}[2][]{\noindent\ignorespaces%
                                \FrameSep=6pt%
                                \parindent=0pt%
                \vspace*{-1.5em}
                \ifthenelse{\isempty{#1}}{%
                  \begin{GrayBox}{\textsc{#2}}%
                }{%
                  \begin{GrayBox}{\textsc{#2}  parameterized by~{#1}}%
                }
                \begin{tabular*}{\textwidth}{@{\hspace{.1em}} >{\itshape} p{1.8cm} p{0.8\textwidth} @{}}%
            }{
                \end{tabular*}%
                \end{GrayBox}%
                \ignorespacesafterend
            }  

\newcommand{\defparproblema}[4]{
  \begin{defproblemx}[#3]{#1}
    Input:  & #2 \\
    Task: & #4
  \end{defproblemx}
}%

\newcommand{\defproblema}[3]{
  \begin{defproblemx}{#1}
    Input:  & #2 \\
    Task: & #3
  \end{defproblemx}
}%

\usepackage{soul}

\newcommand{\what}{\widehat}
\newcommand{\wtilde}{\widetilde}

\newcommand{\WO}{\textsf{W[1]}}
\newcommand{\WOH}{\textsf{W[1]-hard}\xspace}
\newcommand{\FPT}{\textsf{FPT}\xspace}

\newcommand{\mcc}{\textsc{Multicolored Clique}\xspace}

\newcommand{\dsce}{\textsc{Edge-Disjoint Shortest Cycle Packing}\xspace}
\newcommand{\dscv}{\textsc{Vertex-Disjoint Shortest Cycle Packing}\xspace}

\newcommand{\ttheta}{\wtilde{\theta}}

\newcommand{\tcal}[1]{\C{\wtilde{#1}}}
\newcommand{\hcal}[1]{\C{\what{#1}}}

\newcommand{\bbeta}{\check{\beta}}
\newcommand{\tdelta}{\wtilde{\delta}}
\newcommand{\hbeta}{\what{\beta}}
\newcommand{\lsctfull}{laminar shortest cycles tree}
 
\title{Packing Diverse Short Cycles}  

\usepackage{authblk}

\author[1]{Akanksha Agrawal}
\author[2]{Fedor V. Fomin}
\author[2]{Petr A. Golovach}
\author[1]{Vinod Gupta}
\author[2]{Yash Hiren More}
\author[1]{Vidya Sagar Sharma}

\affil[1]{Indian Institute of Technology Madras, Chennai, India}
\affil[2]{University of Bergen, Bergen, Norway}

\begin{document}

\maketitle

\begingroup
\renewcommand{\thefootnote}{}
\footnotetext{\scriptsize
This work was supported by the Research Council of Norway under the BWCA  (grant no.~314528) and Extreme-Algorithms (grant no~355137) projects, the European Research Council (ERC) via grant NewPC, reference 101199930, and the Trond Mohn forskningsstiftelse (grant no. TMS2023TMT01).
Email addresses:
    \texttt{akanksha@cse.iitm.ac.in},
    \texttt{fedor.fomin@uib.no},
    \texttt{petr.golovach@uib.no},
    \texttt{guptavinod439@gmail.com},
    \texttt{yash.h.more@uib.no},
    and \texttt{vidyasagartifr@gmail.com}.
}
\addtocounter{footnote}{-1}
\endgroup

\begin{abstract}
    Bentert, Fomin, Golovach, Korhonen,  Lochet, Panolan,  Ramanujan,  Saurabh, and Simonov (SODA 2025) initiated the parameterized study of \dsce: given a weighted graph $G$ and an integer $k$, decide whether $G$ contains $k$ edge-disjoint cycles of minimum weight. They showed that the problem admits an algorithm running in time $n^{\Ocal(k^6)}$ and asked whether it is fixed-parameter tractable or \WOH parameterized by $k$. We resolve this question by proving that \dsce is \WOH parameterized by $k$, even on unweighted subcubic graphs. The same lower bound also applies to the vertex-disjoint variant.
    
    For planar graphs,  the same paper provides a construction of a kernel with $\Ocal(k^2)$ vertices and an algorithm running in time $k^{\Ocal(k)} \cdot n^{\Ocal(1)}$, and explicitly asked whether the problem admits a single-exponential algorithm of running time $2^{\Ocal(k)} \cdot n^{\Ocal(1)}$. Rather than addressing this question in isolation, we introduce a more general framework, \maxddisjointcoveragefull, which asks for $k$ shortest cycles that may overlap in a controlled way while maximizing the total weight of covered edges. This framework simultaneously captures edge-disjoint shortest cycle packing, the problem of finding diverse shortest cycles, and the problem of maximizing edge coverage by shortest cycles. Our main algorithmic result shows that \maxddisjointcoveragefull{} can be solved on planar graphs in time $2^{\Ocal(k)} \cdot n^{\Ocal(1)}$, thereby establishing a single-exponential algorithm for \dsce as a special case.
    
    The key ingredient behind our algorithm is a structural analysis of the \emph{Laminar Shortest Cycles Tree}, a tree-like decomposition that reveals a laminar interaction pattern among shortest cycles in planar graphs and enables an efficient dynamic programming algorithm.
\end{abstract}

\newpage
\tableofcontents \newpage
\section{Introduction}
Cycle packing is a fundamental optimization problem in which one seeks to find the maximum number of edge- (or vertex-) disjoint cycles in a graph. 
The authors in~\cite{psc,DBLP:journals/talg/BentertFGKLPRSS26} initiated the study of \dsce and \dscv, the problems where one seeks a packing of $k$ shortest (minimum-weight) 
edge- (or vertex-) disjoint cycles, respectively.
They showed that both packing problems are solvable in time $n^{\Ocal(k^6)}$, while ruling out a kernel of size polynomial in $k$ under standard complexity assumptions.
Our first theorem resolves the FPT vs.\ \WOH question left open in~\cite{DBLP:journals/talg/BentertFGKLPRSS26}.
%

\begin{restatable}{theorem}{Whardness}
\label{thm:hard_generalgraphs}
\dsce  is \WOH   parameterized by $k$ on unweighted subcubic graphs.  
\end{restatable}

Since two cycles in a subcubic graph are edge-disjoint if and only if they are vertex-disjoint, we immediately obtain the hardness for the vertex-disjoint variant of the problem. 

\begin{corollary}\label{corr:hard_generalgraphsvertexd}
\dscv  is \WOH   parameterized by $k$ on unweighted subcubic graphs.   
\end{corollary}







\Cref{thm:hard_generalgraphs}  rules out fixed-parameter tractability under standard assumptions 
for general graphs.  For planar graphs, the authors in~\cite{psc,DBLP:journals/talg/BentertFGKLPRSS26}
showed that \dsce admits a kernel with $\Ocal(k^2)$ vertices and an algorithm running
in time $k^{\Ocal(k)} \cdot n^{\Ocal(1)}$, and asked explicitly whether a
\emph{single-exponential} algorithm  with running time $2^{\Ocal(k)} \cdot
n^{\Ocal(1)}$  exists.
 
Our next theorem resolves a more general question than what was asked in ~\cite{DBLP:journals/talg/BentertFGKLPRSS26}.
It works for a natural  relaxation of edge-disjointness: for an integer $d \geq 0$, two cycles $C_1$ and $C_2$ are \emph{$d$-diverse} if $w(E(C_1) \cap E(C_2)) \leq d$, that is, the total weight of shared edges is at most~$d$. 
Edge-disjointness is the special case $d = 0$. 
Building on this, we introduce the following problem, which simultaneously captures edge-disjoint shortest cycle packing, the search for {diverse} shortest cycles, and the problem of maximizing edge coverage by shortest cycles.

\defproblema{\maxddisjointcoveragefull{} (\maxddisjointcoverage)}
{A weighted graph $G$ with weight function $w : E(G) \rightarrow \mathbb{Z}_{> 0}$, and integers $k,d$, and $\ell$.}
{Decide whether there exists a set $\Ccal = \{C_1, C_2, \ldots, C_k\}$ of exactly $k$ many pairwise $d$-diverse shortest cycles of $G$ such that the weight of the set of edges covered by these cycles is at least $\ell$, that is, $w\!\left(\bigcup_{i=1}^{k} E(C_i)\right) \geq \ell$.}

The framework of \maxddisjointcoveragefull{} captures several problems of independent
interest.

\begin{itemize}

\item \textbf{Edge-disjoint Packing.}
  Taking $d = 0$ (or equivalently $\ell = k \cdot g(G)$) forces the $k$ cycles to be
  pairwise edge-disjoint, recovering \dsce exactly.
  %

\item \textbf{Diverse Shortest Cycles.}
  Taking $\ell = 0$ removes the coverage requirement, leaving the problem of finding
  $k$ shortest cycles that are pairwise $d$-diverse.
  This fits within the growing literature on \emph{diversity} in combinatorial
  optimization~\cite{baste2020diversity,fomin2024diverse,hanaka2021finding}: rather
  than a single optimal solution, one seeks a collection of near-optimal solutions
  that are spread out in a meaningful sense.
  On planar graphs, $d$-diversity among shortest cycles translates directly to
  $d$-diversity among global minimum cuts in the dual graph, connecting our result to
  the work of de~Berg, Mart{\'\i}nez, and Spieksma~\cite{deBerg2023diverseCuts} on diverse minimum
  $s$-$t$ cuts.

\item \textbf{Maximum Edge Coverage by Shortest Cycles.}
  Taking $d = w(E(G))$ removes the diversity requirement, leaving the problem of finding $k$ shortest cycles whose union covers edges of total weight as large as possible.
  This can be seen as a variant of the classical
\textsc{Maximum Coverage} problem (find $k$ sets from a family
$\mathcal{S}$ covering the most elements) with the important
distinction that $\mathcal{S}$ is not given explicitly: the family
of all shortest cycles of $G$ can be exponentially large and given implicitly rather than enumerated.
%
  %

\end{itemize}

Our main algorithmic result not only resolves the open question from~\cite{DBLP:journals/talg/BentertFGKLPRSS26} asking
for a single-exponential algorithm for \dsce on planar graphs, but
also establishes the first \classFPT algorithms for \textsc{Diverse
Shortest Cycles} and \textsc{Maximum Edge Coverage by Shortest Cycles} on planar graphs.

\begin{restatable}{theorem}{thmmaxddisjointcoverage}
    \label{thm:maxddisjointcoverage}
    \maxddisjointcoveragefull{} on planar graphs admits an algorithm with running time $2^{\Ocal(k)}\cdot n^{\Ocal(1)}$.
\end{restatable}

The running time in~\Cref{thm:maxddisjointcoverage} is stated assuming the word RAM model. We note that the result holds for real-valued edge weights assuming the standard real RAM model of computation, in which basic arithmetic operations on real numbers can be performed in constant time.

\subsection{Related Work}

Both \dsce and \dscv are \NP-complete even on planar graphs, since they generalize the triangle packing  roblem~\cite{psc,DBLP:journals/talg/BentertFGKLPRSS26}.
Rautenbach and Regen~\cite{RautenbachR09} studied these problems on graphs of girth $g \in \{4,5\}$. 
In this case, both problems are solvable in polynomial time on graphs of maximum degree~$3$, but are \textsf{APX}-hard for maximum degree~$4$.
For $g \geq 6$, the problems are \textsf{APX}-hard already on cubic graphs.
Approximation algorithms of ratio $(3+\varepsilon)$ for \dsce and \dscv on planar graphs follow from the framework of Schlomberg, Thiele, and 
Vygen~\cite{DBLP:journals/siamcomp/SchlombergTV25} for uncrossing families of cycles.

The systematic parameterized study of \dscv and \dsce was initiated in~\cite{psc,DBLP:journals/talg/BentertFGKLPRSS26}.
For general graphs, they gave an \classXP algorithm with running time $n^{\Ocal(k^6)}$, which constituted the first algorithm solving these problems for any fixed~$k$, and ruled out a polynomial kernel under standard complexity assumptions. 
The \classFPT vs.\ \WOH status of \dsce and \dscv on general graphs was left explicitly open. 
For planar graphs, they obtained a polynomial kernel with $\Ocal(k^2)$ vertices and an \classFPT algorithm with running time
$k^{\Ocal(k)} \cdot n^{\Ocal(1)}$, and explicitly asked whether a single-exponential algorithm exists. 
Our \Cref{thm:hard_generalgraphs} resolves the first open question by proving \WO-hardness on subcubic graphs, and \Cref{thm:maxddisjointcoverage} resolves the second by giving a $2^{\Ocal(k)} \cdot n^{\Ocal(1)}$ algorithm for the more general \maxddisjointcoveragefull{} on planar graphs.

A close relative of shortest cycle packing is \dsp. 
In this problem, given $k$ pairs of terminal vertices, the task is to decide whether they can be connected by pairwise vertex- (or edge-) disjoint shortest paths. 
The problem is \NPH when~$k$ is part of the input~\cite{Eil98} and \WOH parameterized by~$k$~\cite{BNRZ23}. Lochet~\cite{Loc21} gave a polynomial-time algorithm for any fixed~$k$ in undirected unweighted graphs, with an improved running time obtained in~\cite{BNRZ23}. 
The recent breakthrough of Pilipczuk, Stamoulis, and W\l{}odarczyk~\cite{PilipczukSW26} establishes that the problem is \classFPT on planar graphs.

In our paper, we are using laminar shortest cycles trees introduced in~\cite{psc,DBLP:journals/talg/BentertFGKLPRSS26}; similar laminar structures of cycles in planar graphs have been use before~\cite{DBLP:journals/siamdm/HartvigsenM94,DBLP:journals/combinatorica/GoemansW98, DBLP:journals/siamcomp/SchlombergTV25,DBLP:conf/icalp/Schlomberg24}.

There has been growing interest in finding \emph{diverse} collections of solutions to combinatorial optimization problems from the perspective of parameterized complexity. 
The paper ~\cite{baste2020diversity} initiated this line of research, and~\cite{fomin2024diverse} developed a general framework for diverse collections in matroids and graphs.
Eiben, Koana, and Wahlstr\"om~\cite{eiben2025determinantal} developed algebraic techniques based on determinantal sieving that apply to various problems in the diverse-solutions settings, although their methods are not applicable to our problem.
Other general techniques were recently proposed by Kumabe~\cite{Kumabe25} and 
Ashok et al.~\cite{ashok2026solutionmanyoraclebasedfpt}. However, our problem does not fit into this framework as in our case, the diversity measures and the size of the diverse objects (that is, the girth) are unbounded. 
In the context of cuts, the papers ~\cite{deBerg2023diverseCuts,deBerg2025diverse} studied diverse minimum $s$-$t$ cuts and diverse solutions in problems with a distributive lattice structure. 
More recently \cite{dehaleesan2026fpt} gave an FPT algorithm for finding diverse minimum $s$-$t$ cuts.
Our problem \maxddisjointcoveragefull{} contributes to this line of work, where the objects of interest are shortest cycles and diversity is measured by the weight of pairwise edge intersections; on planar graphs, this is equivalent to diversity among global minimum cuts in the dual graph.%

\section{Overview of our methods}\label{subsec:intro}
 
\subsection{Proof of \Cref{thm:hard_generalgraphs}}

The reduction in~\cite{psc} 
establishes \WO-hardness of deciding, given a graph $G$ and integers $k$ and $L$, whether $G$ contains $k$ vertex- (or edge-) disjoint cycles of length at most $L$ each.
In this reduction, $L$ may be larger than the girth of $G$, and this flexibility is crucial. In the case of shortest cycles, however, all
cycles must have length exactly equal to the girth of $G$. 
This tight length constraint requires substantially more intricate arguments, even though we reuse some ideas from~\cite{psc}.

\Cref{thm:hard_generalgraphs} is proved by reducing from the \WOH \mcc problem.
This problem asks, given a graph $G$ whose vertices are partitioned into $h$ sets $V_1,\ldots,V_h$, called \emph{color classes}, whether $G$ has a \emph{multicolored} clique of size $h$, that is, a clique containing exactly one vertex from each color class. 
The parameter of \mcc is the size of the solution size $h$.
The key ingredients of our reduction are the constructions of \emph{vertex} and \emph{edge} selection gadgets which are used to encode the selection of exactly one vertex in each color class and exactly one edge between every pair of color classes of the input instance of \mcc. 
Then we combine these gadgets to encode the instance of \mcc.
To construct vertex and edge selection gadgets, we select a sufficiently large (with respect to the size of $G$) integer $L=2\ell$ which will be the girth of the constructed graph, and then build gadgets as unions of overlapping packings of cycles of length $L$. 

\begin{figure}[ht]
 \begin{center}
 \scalebox{0.6}{
 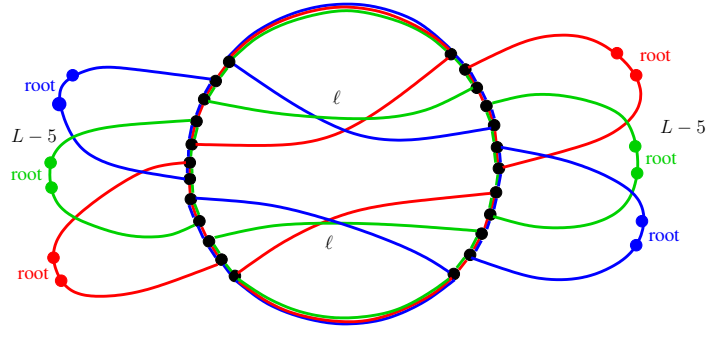}
 \end{center}
 \caption{Construction of the edge selection gadget $F$ for $m=3$. The paths in the left and in the right parts are of length $L-5$, and the paths in the middle part are of length $\ell=L/2$. The middle vertices of the paths in the left and in the right parts are \emph{roots} that are used to connect the edge selection gadget with the corresponding vertex selection gadgets. For $F$ constructed for the pair of color classes $V_i,V_j$ with $i<j$, the roots from the left part are used to connect the edge selection gadget with the vertex selection gadget constructed for $V_i$, and the roots from the right part are used for the connection with the vertex selection gadget constructed for $V_j$.
 }\label{fig:esg_intro}
 \end{figure}

 The crucial difference between the \WO-hardness proof in~\cite{psc} and our result is that the edge selection gadget allows us to work with shortest cycles, while \cite{psc} has to relax the length constraint. 
For the edge selection gadget $F$, we assume that the number of edges between each pair of color classes is the same integer $m$. Then $F$ contains $m$ $4$-tuples of $L$-length cycles. The idea is demonstrated in~\Cref{fig:esg_intro} for $m=3$, where distinct $4$-tuples are colored in red, green, and blue, respectively. The crucial property of the gadget is that if we wish to take $4$ edge-disjoint shortest cycles from $F$, then because of the overlap of the shortest cycles, the only option is to choose cycles of the same color. This choice corresponds to the selection of an edge in the instance of \mcc between two color classes.

\begin{figure}[ht]
 \begin{center}
 \scalebox{0.6}{
 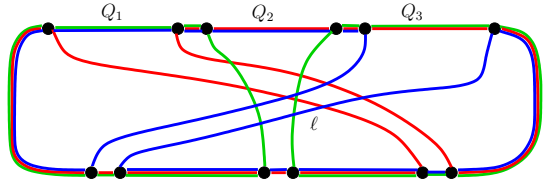}
 \end{center}
 \caption{Construction of the vertex selection gadget for $n=3$. The paths  in the middle part are of length $\ell=L/2$. The paths $Q_1,Q_2,Q_3$ correspond to the vertices of a color class and are used to connect the vertex selection gadget with the edge selection gadgets.
 }\label{fig:vsg_intro}
 \end{figure}

The construction of the vertex selection gadget $H$ exploits the same idea. We assume that each color class contains $n$ vertices. Then $H$ is composed of $n$ pairs of cycles of length $L$. We use~\Cref{fig:vsg_intro} to explain the construction for $n=3$; the $3$ pairs of cycles are shown in red, green, and blue. For the vertex selection gadget, the crucial property is that if we are required to take $2$ shortest cycles from $H$, then the only possibility is to select cycles of the same color. The choice naturally corresponds to a selection of a vertex in the color class of the instance of \mcc.
 
We use the paths $Q_s$ in the vertex selection gadget (see~\Cref{fig:vsg_intro}) constructed for a color class $V_i$ and roots of the edge selection gadgets (see~\Cref{fig:esg_intro}) constructed for the pairs $V_i,V_j$ for $j\neq i$. Namely, we identify edges of $Q_s$ corresponding to $v\in V_i$ without common endpoints with the roots corresponding to the edges incident to $v$ in $G$. To see the idea, notice that the selection of, say, the red cycles in the vertex selection gadget in \Cref{fig:vsg_intro} leaves the edges of $Q_1$ free to use in other cycles in the attached edge selection gadgets. We apply this idea to argue that $G$ has a multicolored clique if and only if the constructed graph has a packing of $2h+4\binom{h}{2}$ edge-disjoint shortest cycles.

Let us highlight the aspects of the reduction that we consider the most
delicate. First, since every cycle of the packing must have length exactly
$g(G')=L$, all gadgets are built as unions of overlapping cycles of length
exactly $L$, and their correctness rests on an overlap argument: the
$4$-tuples of shortest cycles in an edge selection gadget (the pairs of 
shortest cycles in an vertex selection gadget) share edges in a
cyclically shifted pattern, which forces a monochromatic choice. 
Second, the
gadgets are glued together by identifying root edges of edge selection
gadgets with specific edges of vertex selection gadgets. Each identification
creates new cycles passing through several gadgets, and we must certify that
none of them has length below $L$. Finally,
the equivalence rests on an exact counting argument: the budget
$k=2h+4\binom{h}{2}$ equals the maximum number of edge-disjoint shortest
cycles that the gadgets can jointly supply, so every gadget must be
saturated, and saturation is possible only along a selection encoding a
multicolored clique.

\subsection{Proof of \Cref{thm:maxddisjointcoverage}}
We give an intuitive idea behind our algorithm that proves \Cref{thm:maxddisjointcoverage}. 
At a high level, our algorithm has two sources of novelty compared to the
$k^{\Ocal(k)}\cdot n^{\Ocal(1)}$-time algorithm of~\cite{psc}: a finer
structural analysis of the laminar shortest cycles tree, which reduces all
guessing from quadratically many objects to linearly many, and a new
two-layer dynamic programming.

Consider an instance $(G,w,k,d,\ell)$ of \maxddisjointcoveragefull{}, where $G$ is a planar graph. We assume that every edge of $G$ is in some shortest cycle, and $G$ is embedded on the plane in such a way that the frontier of the external (outer) face is a shortest cycle. 
Whenever we say that $G$ is a plane graph, we mean that $G$ is embedded this way.
 The work ~\cite{psc} introduces a tree-like structure that captures the shortest cycles in a weighted plane graph, called the {\em laminar shortest cycles tree (LSCT)}; similar structures have also been previously used (for instance, see~\cite{DBLP:journals/siamdm/HartvigsenM94,DBLP:journals/siamcomp/SchlombergTV25,DBLP:journals/combinatorica/GoemansW98,DBLP:conf/icalp/Schlomberg24}).
Intuitively speaking, a family of shortest cycles $\C{C}$ in a given plane graph $G$ is {\em laminar} if for distinct $C,C'\in \C{C}$, one of the following holds: i) $C$ is drawn inside $C'$, ii) $C'$ is drawn inside $C$, or 
iii) the regions (not containing the outer face) bounded by the closed
    curves of $C$ and $C'$ are disjoint.
    
     We begin by computing the laminar shortest cycles tree $\Tcal$ for the
    given plane graph $G$ (see \Cref{def:laminar-tree} for the formal
    definition). The nodes of this tree are shortest cycles in $G$ (but not all of them). 
 \Cref{fig:three_intersecting} illustrates the kind of shortest cycles that are missing
    from $\Tcal$; an intuitive explanation follows below.
The only shortest cycles in $G$ that are not nodes of $\Tcal$ are those
    consisting of two ``non-consecutive'' shortest paths, each of length
    $g(G)/2$, between a pair of vertices $s,t\in V(G)$, where $g(G)$ is the
    weighted girth of $G$; moreover, such a pair excludes the two
    ``outermost'' shortest $s$-$t$ paths.
The shortest cycles that are nodes of $\C{T}$ are {\em tree cycles} and the ones that are not its nodes are {\em non-tree cycles}. 

We briefly describe the construction of this laminar shortest cycles tree, since its properties will be algorithmically exploited. 
Intuitively, the tree is rooted at this outer face shortest cycle, and we (recursively) add shortest cycles to this tree by going ``inside'' the region (not containing the outer face) bounded by the current shortest cycle node $\alpha$ as follows. 
    (See \Cref{fig:lsct-example} for an illustrative example; the concrete details in the figure relate to the formal definition presented later and may be ignored at this point.)
Find a maximal set $\SC{C}$ of shortest cycles where i) each of them is fully contained inside $\alpha$ in the given planar drawing, ii) no
    shortest cycle in $\Ccal$ fully contains another shortest cycle from
    $\Ccal$ in the planar drawing, and iii) there is no cycle (different from
    $\alpha$) contained inside $\alpha$ that fully contains a cycle from
    $\Ccal$. 
We make each $\beta\in \SC{C}$ a child of $\alpha$, and recursively proceed until we reach facial cycles. 
Roughly speaking, the above gives a clean ``laminar'' structure of shortest cycles in the given plane graph, and one may think of it as encountering cycles ``layer-by-layer'' starting from the outermost cycle. 

\begin{figure}[t]
    \centering
    \begin{tikzpicture}

\node (s) at (0, 2) [circle, draw] {$s$};
\node (t) at (0, -2) [circle, draw] {$t$};

\draw[ black] (s)     .. controls (-3.5, 1) and (-3.5, -1) .. (t) node[midway, left] {$P_1$};
\draw[ black] (s)       .. controls (-2.5, 1) and (-2.5, -1) .. (t) node[midway, left] {$P_2$};
\draw[ black] (s)      .. controls (-1.4, 1) and (-1.4, -1) .. (t) node[midway, left] {$P_3$};  
\draw[ black] (s)       .. controls (1.4, 1) and (1.4, -1) .. (t) node[midway, right] {$P_4$};  
\draw[ black] (s)      .. controls (2.5, 1) and (2.5, -1) .. (t) node[midway, right] {$P_5$}; 
\draw[ black] (s)     .. controls (3.5, 1) and (3.5, -1) .. (t) node[midway, right] {$P_6$}; 

\end{tikzpicture}

\caption{The figure illustrates embedding of two distinct vertices $s,t$ in a plane graph $G$, connected by exactly six internally pairwise vertex disjoint $s-t$ paths of lengths exactly $g(G)/2$, where $g(G)$ is the weighted girth of given graph $G$. In the \lsctfull, the cycle $P_1\cup P_6$ will be a node. For each $i\in [5]$, the cycle $P_i\cup P_{i+1}$ will be a child of $P_1\cup P_6$. However, the shortest cycles $P_{i}\cup P_{j}$, where $i,j \in [6]$, $i<j$ and $(i,j) \notin \{(1,6)\} \cup \{(i,i+1) \mid i\in [5]\}$, will not be represented in the \lsctfull. It can be shown that every shortest cycle, other than the ones obtained only using the paths $P_i$, $i\in [6]$, is embedded: i) either ``inside'' (including the boundary) one of the cycles $P_i\cup P_{i+1}$, $i\in [5]$, or ii) embedded ``outside'' (including the boundary) the closed curve defined by $P_1\cup P_6$.}  
    \label{fig:three_intersecting}
\end{figure}
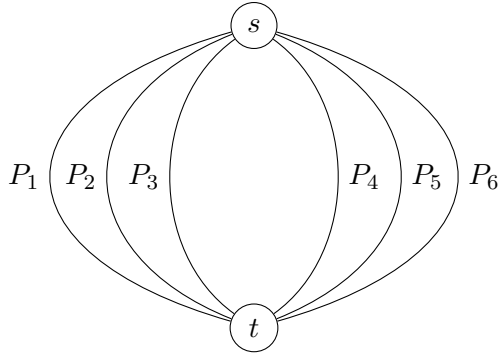

\begin{figure}[t]
\centering
\begin{subfigure}[t]{0.56\textwidth}
\centering
\begin{tikzpicture}[scale=0.78, every node/.style={font=\small}]

  \coordinate (s) at (0,0);
  \coordinate (t) at (8,0);


  \fill[red!6]
    (s) .. controls (2,5.0) and (6,5.0) .. (t)
        .. controls (6,-2.8) and (2,-2.8) .. cycle;

  \fill[blue!10]
    (s) .. controls (2,5.0) and (6,5.0) .. (t)
        .. controls (6,1.4) and (2,1.4) .. cycle;

  \fill[blue!20]
    (s) .. controls (2,5.0) and (6,5.0) .. (t)
        .. controls (6,3.2) and (2,3.2) .. cycle;

  \fill[orange!10]
    (s) .. controls (2,1.4) and (6,1.4) .. (t)
        .. controls (6,-2.8) and (2,-2.8) .. cycle;


  \draw[very thick, red!70!black]       
    (s) .. controls (2,5.0) and (6,5.0) .. (t);
  \draw[thick, blue!70!black, dashed]   
    (s) .. controls (2,3.2) and (6,3.2) .. (t);
  \draw[thick, black!55, dashed]        
    (s) .. controls (2,1.4) and (6,1.4) .. (t);
  \draw[very thick, red!70!black]       
    (s) .. controls (2,-2.8) and (6,-2.8) .. (t);

  \node[red!70!black]  at (4.0,  4.15) {$P_0$};
  \node[blue!70!black] at (4,  2.75) {$P_1$};
  \node[black!55]      at (4, 0.75) {$P_2$};
  \node[red!70!black]  at (4, -2.45) {$P_3$};

  \node[blue!80!black] at (5.5, 2.8)  {$C_1$};
  \node[blue!60!black] at (5.5, 1.5) {$C_2$};
  \node[orange!80!black, font=\small\itshape] at (1.0,-0.2) {$C_3$};
  \node[red!70!black]  at (2.5, -2.9) {$C_r$};

  \fill[white] (3.0,-0.70) ellipse (1.1cm and 0.65cm);
  \draw[thick, teal!70!black] (3.0,-0.70) ellipse (1.1cm and 0.65cm);
  \node[teal!70!black] at (3.0,-0.70) {$C_4$};

  \fill[white] (5.5,-0.55) ellipse (1.1cm and 0.65cm);
  \draw[thick, teal!70!black] (5.5,-0.55) ellipse (1.1cm and 0.65cm);
  \node[teal!70!black] at (5.5,-0.55) {$C_5$};

  \filldraw (s) circle (2.5pt) node[left=3pt]  {$s$};
  \filldraw (t) circle (2.5pt) node[right=3pt] {$t$};

\end{tikzpicture}
\caption{Plane graph $G$. $C_r{=}P_0{\cup}P_3$ is a splittable-node with
  $\SC{C}_{{\sf s}}(C_r){=}\{C_1,C_2,C_3\}$, where $C_1{=}P_0{\cup}P_1$,
  $C_2{=}P_1{\cup}P_2$, and $C_3{=}P_2{\cup}P_3$.
  $C_1$ and $C_2$ are leaf cycles;
  $C_3$ is an unsplittable-node with $\SC{C}_{{u}}(C_3){=}\{C_4,C_5\}$. The shortest cycles formed by paths $P_0$ and $P_2$, and by paths $P_1$ and $P_3$ are non-tree cycles.}
\label{fig:lsct-graph}
\end{subfigure}
\hfill
\begin{subfigure}[t]{0.41\textwidth}
\centering
\begin{tikzpicture}[
  level distance    = 1.5cm,
  level 1/.style    = {sibling distance=2.4cm},
  level 2/.style    = {sibling distance=1.4cm},
  every node/.style = {circle, draw, minimum size=9mm,
                        inner sep=1pt, font=\small},
  snode/.style = {fill=blue!15,   draw=blue!55!black,   thick},
  unode/.style = {fill=orange!20, draw=orange!65!black, thick},
  lnode/.style = {fill=gray!15,   draw=gray!50,         thick},
]
\node[snode] (Cr) {$C_r$}
  child { node[lnode] (C1) {$C_1$} }
  child { node[lnode] (C2) {$C_2$} }
  child { node[unode] (C3) {$C_3$}
    child { node[lnode] (C4) {$C_4$} }
    child { node[lnode] (C5) {$C_5$} }
  };

\node[draw=none, rectangle, right=-0.15cm of Cr,
      font=\scriptsize, blue!70!black]   {(S)};
\node[draw=none, rectangle, right=-0.15cm of C3,
      font=\scriptsize, orange!70!black] {(U)};
\end{tikzpicture}

\smallskip
{\small\centering
  \textcolor{blue!70!black}{$\bullet$ splittable-node (S)}\\[1pt]
  \textcolor{orange!70!black}{$\bullet$ unsplittable-node (U)}\enspace
  \textcolor{gray!60!black}{$\bullet$ leaf}
}
\caption{LSCT $\Tcal(G)$ rooted at $C_r$.}
\label{fig:lsct-tree}
\end{subfigure}

\caption{Example of \lsctfull\ construction
  (\Cref{def:laminar-tree}).}
\label{fig:lsct-example}
\end{figure}

\subparagraph{A simpler case of edge-disjoint cycles.}  The work \cite{psc} uses the \lsctfull\ to obtain an FPT algorithm for {\sc Edge-Disjoint Shortest Cycle Packing}, running in time $k^{O(k)}\cdot n^{O(1)}$, for the given plane graph $G$ as follows: i) if the number of leaves in the \lsctfull\ $\C{T}$ is at least $4k$, using the Four Color Theorem, one can directly conclude that the instance is a yes-instance. ii) Else, the number of leaves in $\C{T}$ is less than $4k$. 
In this case, they show that there is a set of $O(k^2)$ appropriately selected ``deepest'' (in the layers) shortest cycles which necessarily contains a solution (if one exists). 
  We remark that the above marking scheme relies on the ``strict
    laminarity'' property of some optimal solution; we discuss strict
    laminarity in more detail shortly. 
With these marked cycles, now one can easily obtain an FPT algorithm running in time $k^{O(k)}\cdot n^{O(1)}$. 

Before describing the challenges we face in obtaining the $2^{O(k)}\cdot n^{O(1)}$-time algorithm for \maxddisjointcoveragefull, and the way we overcome them, we discuss how using our approach we can obtain a simpler single-exponential FPT algorithm for {\sc Edge-Disjoint Shortest Cycle Packing} (which is the special case of our problem for $d=0$). 
Again, the non-trivial case for us is when the number of leaves in the \lsctfull\ $\C{T}$ for the given plane graph $G$ is less than $4k$. 
In this case, we can obtain a family $\C{Q}$ of shortest paths that together cover all the non-tree cycles in $G$, where: i) $|\C{Q}| \in O(k)$, and ii) for each non-tree cycle $C$ in $G$, there are two distinct paths $P_1,P_2\in \C{Q}$ such that $C = P_1 \cup P_2$. 
With the above family of paths $\C{Q}$ at hand, we can simply ``guess'' the exact subset of these paths forming the non-tree cycles in the solution we are looking for. 
  We remark that although the number of non-tree cycles can be as large as
    $\Omega(k^2)$, guessing the paths requires iterating over only
    $2^{\Ocal(k)}$ subsets of $\Qcal$, as opposed to the $2^{\Omega(k^2)}$
    choices needed to guess the exact set of non-tree cycles in the solution.

It remains to handle the tree cycles, which are exactly the nodes of the
    laminar shortest cycles tree $\Tcal$. We proceed by a bottom-up dynamic
    programming over $\Tcal$, starting from its leaves, as follows. At a node
    $\alpha$ of $\Tcal$, we seek a solution that contains $\alpha$ and some
    descendants of $\alpha$, and we assume that for every strict descendant
    of $\alpha$ the corresponding dynamic programming entries have already
    been computed.
We note that the number of nodes in $\Tcal$ may not be bounded by a
    function of $k$, and thus we do not know which descendants to query
    directly.
However, as the number of leaves in $\C{T}$ is bounded, so is the number of nodes of degree at least $3$, as well as the number of inclusion-maximal paths in $\C{T}$ containing only nodes of degree $2$.
So intuitively, we can guess whether the ``closest'' descendant of $\alpha$ in the solution is a leaf, a node of degree at least three, or a node lying on a maximal degree-$2$ path (and how many solution cycles appear below it).

 Once we have the above information (the number of possible guesses is
    $2^{\Ocal(k)}$), we know precisely the descendants that are leaves or
    nodes of degree at least three, and the unknown ones lie on maximal
    degree-$2$ paths.
 The crucial property, which holds only for packing edge-disjoint cycles,
    is the following: if a solution exists, then there also exists a solution
    that satisfies ``strict laminarity'' in the following sense:  for any two distinct shortest cycles $C,C'$ in the solution, one of the following holds: i) $C$ is embedded strictly inside (modulo the vertices) the region bounded by $C'$, ii) $C'$ is embedded strictly inside (modulo the vertices) the region bounded by $C$, or iii) the regions bounded by the curves $C$ and $C'$ are disjoint (modulo the vertices). 
In the above, (roughly speaking) the region bounded by a cycle $C$ is always the one not containing the outer face. 
Now to identify the unknown descendant that lies in a given maximal degree $2$ path, using ``strict laminarity'', we can show that  greedily choosing the ``closest''  non-intersecting shortest cycle already works. 
The above gives us the $2^{O(k)}\cdot n^{O(1)}$-time algorithm for {\sc Edge-Disjoint Shortest Cycle Packing}.

\input{Figures/fig-4c}

\subparagraph{Handling the general case of \maxddisjointcoveragefull.} 
The algorithmic steps we discussed for {\sc Edge-Disjoint Shortest Cycle Packing}
up to the greedy choice
    handling descendants on maximal degree-$2$ paths of $\C{T}$, remain the
    same. However, as the strict laminarity property completely breaks down when cycles can share edges, the greedy strategy does not work. 
    More precisely, two cycles in different maximal degree-$2$ paths may intersect and the greedy strategy falls apart.
   Instead, our approach is to handle the choices of descendants across
    different maximal degree-$2$ paths simultaneously, in one go. To achieve
    this, we need a better understanding of the structure of the solution we
    are looking for, as we explain next.
    
Let  $\alpha$ be a node of $\C{T}$ and let the ``closest'' tree cycles in an optimal solution that appear ``below'' $\alpha$ in $\C{T}$ be the unknown cycles $\beta_1,\beta_2,\cdots, \beta_b$. In the above, by ``closest'' we mean that there is no tree cycle in the solution that is strictly between $\alpha$ and $\beta_i$, $i\in [b]$. Recall that we cannot precisely guess the $\beta_i$s that belong to maximal degree $2$ paths. Let us take a closer look at the ``unknown'' cycle $\beta_i$, which belongs to a maximal degree $2$ path $Q$ in the subtree $\C{T}_\alpha$ of $\C{T}$ rooted at $\alpha$. Furthermore, let $\gamma_i$ be the endpoint of $Q$ that is closest to $\alpha$. Using planarity, we can guarantee that the edges that $\beta_i$ may share with some $\beta_j$, where $j\in [b]\setminus \{i\}$ must also belong to $\gamma_i$ (see \Cref{fig:4c-1}) and such intersections are highly structured, as illustrated in the figure. 

Also, by the structure of intersections of shortest cycles, we can guarantee that $P_i = \gamma_i \cap \beta_i$ is a sub-path of $\gamma_i$. These two properties let
    us completely identify the   ``non-trivial'' part of the intersection of $\beta_i$ with other cycles once we know, roughly speaking, the ``start'' and ``end'' cycle intersections lying on $P_i$. Thus, we will use the above information to derive the solution we are looking for. At the above step, we will again exploit planarity to ensure that we can achieve the above in time $2^{O(k)}\cdot n^{O(1)}$. We remark that only ``start'' and ``end'' cycles are relevant for  $\beta_i$. We encode these guessed ``end-cycles'' as a graph, in which we can guarantee that each connected component is a path or a cycle. 
  These highly structured components
    then allow us to identify the correct solution descendants by a second,
    separate dynamic programming.  
 
Let us highlight the aspects of the algorithm that we consider the most
novel. First, the marking scheme of~\cite{psc} inherently costs
$k^{\Ocal(k)}$: it identifies $\Ocal(k^2)$ candidate cycles and guesses the
solution among them. To obtain a single-exponential running time, every
guess in our algorithm must range over objects whose number is linear in
$k$ (the longitudinal paths covering the non-tree cycles, the marked
representatives, and the intersection types) and it is planarity that keeps
the product of all these choices single-exponential: a sparsity argument in
the dual graph shows that the total size of the ``relevant sets'' of all
representatives is linear in their number, so that the intersection types
can be enumerated in $2^{\Ocal(k)}$ time overall. Second, for $d>0$ the
strict laminarity property, on which the approach of~\cite{psc} (as well as
our simpler algorithm for the edge-disjoint case) crucially relies,
completely breaks down: an optimal solution may be forced to use cycles that
pairwise share edges, and the choices of solution descendants on different
maximal degree-$2$ paths of $\Tcal$ interact through these shared edges.
Our replacement for the greedy step is a second layer of dynamic
programming: the guessed ``start'' and ``end'' intersections define an
auxiliary graph whose connected components are paths and cycles, and the
interactions within each component are resolved by a dynamic program along
it. Finally, the coverage objective introduces a difficulty absent
from pure packing: the algorithm must account exactly for the weight of
edges covered by two solution cycles, including cycles it never explicitly
identifies. The structured intersections in planar graphs (no edge belongs
to three solution cycles, and the pairwise intersections are subpaths of the
known representatives) are precisely what makes this bookkeeping possible.

\section{Preliminaries}
\label{sec:preliminaries}
\noindent
{\bf Numbers and Sets.} We denote the set of natural numbers by $\mathbb{N} = \{0,1,2,\ldots\}$. For $p\in \mathbb{N}$, we use $[p]$ to denote the set $\{1,\ldots,p\}$.
The set of integers is denoted by $\mathbb{Z}$, and $\mathbb{Z}_{>0} = \{a\in \mathbb{Z} \mid a>0\}$. 
We will use the following folklore result, which can be obtained using the well-known AM-GM inequality. 

\begin{proposition}\label{prop:num-sum-prod}
For $z, z_1,z_2,\cdots, z_\ell \in \mathbb{N}$, where $z = \sum_{i\in [\ell]}z_i$, we have $\prod_{i\in [\ell]}z_i \in 2^{O(z)}$. 
\end{proposition}

\smallskip
\noindent
{\bf Graphs \& weighted graphs.} Throughout this paper, we consider finite undirected graphs and refer to the textbook of Diestel~\cite{Diestel} for any undefined graph-theoretic notions. Unless stated otherwise, all graphs are simple.

For a graph $G$, $V(G)$ and $E(G)$ denote the vertex set and the edge set of $G$, respectively. We denote the number of vertices and edges of $G$ by $n$ and $m$, respectively. A graph together with an edge-weight function $w: E(G) \rightarrow \mathbb{Z}_{>0}$ is a {\em weighted graph}. Unless stated otherwise, whenever we write weighted graph, we will mean an edge weighted graph. For $U\subseteq V(G)$, we write $G[U]$ for the subgraph of $G$ induced by $U$, i.e., $V(G[U])=U$ and $E(G[U]) = \{\{u,v\}\in E(G)\mid u,v\in U\}$. For two graphs $G_1$ and $G_2$, their \emph{intersection} is the graph $G_1\cap G_2$ with $V(G_1\cap G_2)=V(G_1)\cap V(G_2)$ and $E(G_1\cap G_2)=E(G_1)\cap E(G_2)$, while their \emph{union} $G_1\cup G_2$ is the graph with $V(G_1\cup G_2)=V(G_1)\cup V(G_2)$ and $E(G_1\cup G_2)=E(G_1)\cup E(G_2)$.

Consider a weighted graph $G$ with $w: E(G) \rightarrow \mathbb{Z}_{>0}$. For $X \subseteq E(G)$, $w(X) = \sum_{e\in X}w(e)$. For a subgraph $H$ of $G$, $w(H) = w(E(H))$. An ordered sequence of distinct vertices $P=(u_1=u,u_2,\ldots,u_\ell=v)$ is a \emph{path} in $G$ between vertices $u$ and $v$ if $\{u_i,u_{i+1}\}\in E(G)$ for all $1\le i\le \ell-1$. We use the notation $P(u,v)$ to denote a path between two vertices $u,v$; alternatively we say that $P$ is a $u-v$ path. The vertices $u_1$ and $u_\ell$ are the \emph{endpoints} of $P$, and the remaining vertices are its \emph{internal} vertices. 
A \emph{subpath} of $P$ is any path of the form $(u_i,u_{i+1},\ldots,u_j)$ for $1\le i\leq j \le \ell$.
The length of a path is defined as the sum of the weights of edges of the path, i.e., $w(P) = \sum_{i \in [l-1]}{w(\{u_i, u_{i+1}}\})$. 
For two vertices $u$ and $v$ of $G$, the distance between them is the length of the shortest path between them. A \emph{cycle} is a path along with an additional edge between the two endpoints of the path.
The {\em length} of a cycle is the sum of the weights of its edges.
A {\em shortest cycle} of a graph is the smallest length cycle.
The \emph{girth}, $g(G)$, is the minimum-weight  of a cycle in $G$, i.e. $g(G) = \min\{w(C)\mid C\mbox{ is a cycle in } G\}$. We say that the graph $G$ is \emph{clean} if every edge and vertex belongs to some shortest cycle. 

\smallskip
\noindent
{\bf Planar graphs.} A graph is \emph{planar} if we can draw it on a plane without any pair of edges crossing (except at the shared endpoint, if any), and such a drawing is called a \emph{planar embedding} of the graph. 
We can find a planar embedding of a planar graph in linear time \cite{hopcroft_tarjan_planarity}. A planar graph $G$ with a fixed embedding on the plane is a {\em plane} graph. The \emph{faces} of a plane graph are the open regions of the plane bounded by edges and containing no vertices or edges in their interior. The outer region is the \emph{external} face, and the other faces are \emph{internal}. The vertices and edges appearing on the closed walk bounding a face form its \emph{frontier}. A cycle $C$ of $G$ is called a \emph{facial cycle} if it is the frontier of some face.

For a shortest cycle $C$ of a plane graph $G$, we denote by $G_{C}$ the subgraph of $G$ induced by the vertices lying on $C$ or in the interior $\Int(C)$ of the region bounded by the curve of $C$ in the plane. (In the above by interior region, we mean the region not containing the outer face.) 
For the two shortest cycles $C_1$ and $C_2$ in the plane graph $G$, we say $C_1 \leq C_2$ if every vertex and edge of $C_1$ is either a vertex or edge of $C_2$ or is embedded in the curve of $C_2$, i.e., $\Int(C_2)$.

The \emph{dual} graph of a plane graph $G$, denoted $G^*=(F(G),E^*)$, has one vertex for each face of $G$. For every edge $e\in E(G)$, there is a corresponding \emph{dual} edge $e^*\in E^*$ whose endpoints are the two faces whose frontiers contain $e$; if $e$ lies on the frontier of exactly one face (i.e., $e$ is a bridge), then $e^*$ is a self-loop at that face. Note that $G^*$ need not be simple even if $G$ is simple. 
    
\smallskip
\noindent
\textbf{Parameterized Complexity.}
We refer to the textbooks by Cygan et al.~\cite{cygan2015parameterized} and by Downey and Fellows~\cite{DowneyF99} for a comprehensive introduction to parameterized complexity. Here, we briefly recall the main definitions.

A \emph{parameterized problem} is a language $L\subseteq \Sigma^*\times\mathbb{N}$, where $\Sigma^*$ denotes the set of finite strings over a finite alphabet $\Sigma$. An instance of $L$ is a pair $(x,k)$, where $x\in\Sigma^*$ is the input and $k\in\mathbb{N}$ is the \emph{parameter}. A parameterized problem is \emph{fixed-parameter tractable} (\classFPT, for short) if it can be solved in time $f(k)\cdot |x|^{\mathcal{O}(1)}$ for some computable function $f$. The class \classFPT consists of all fixed-parameter tractable parameterized problems. 
%

For two parameterized problems $A,B \subseteq \Sigma^{*} \times\mathbb{N}$, a \emph{parameterized reduction} from $A$ to $B$ is an algorithm which given an instance $(x,k)$ of $A$ produces an equivalent instance $(x',k')$ of $B$ in time $f(k)\cdot|x|^{\Ocal(1)}$ and the parameter $k' \leq g(k)$ for some computable functions $f$ and $g$.

The standard approach to ruling out the existence of an \FPT-algorithm for a parameterized problem, under the standard assumptions in parameterized complexity that $\FPT \neq W[1]$, is to prove that the problem is \emph{$W[1]$-hard}. 
In particular, if there exists a \emph{parameterized reduction} from a $W[1]$-hard parameterized problem $A$ to a parameterized problem $B$, then $B$ is also $W[1]$-hard.

\section{Hardness of \dsce}\label{sec:hardnes}
    
In this section, we prove~\Cref{thm:hard_generalgraphs}. 
We restate it here for the reader's convenience. 

\Whardness*

We reduce from a well-known \WOH\ problem called \mcc~\cite{misFellowsHRV09,DBLP:journals/jcss/Pietrzak03}, which is formally defined below.

\defparproblema{\mcc}{A graph $G$ with a partition $V_1,V_2, \ldots, V_h$ of $V(G)$.}
{h}
{Decide whether there is a clique $X \subseteq V(G)$ such that for each $i \in [h]$, $|X \cap V_i|=1$.}



Let $(G, V_1,V_2, \dots, V_h)$ be an instance of \mcc. 
We assume that $h \geq 3$, as otherwise the instance can be resolved in polynomial time. 
Furthermore, we assume that for all distinct $i,j\in [h]$, the number of edges in $G$ with one endpoint in $V_i$ and the other in $V_j$ is exactly the same integer $m\geq 1$.
If this is not the case, then we can add edges with their endpoints in $V_i$ and $V_j$ forming connected components of the graph. Since $h\geq 3$, this does not change the answer. Finally, we assume that 
for all $i \in [h]$, the size of $V_i$ is the same integer  $n\geq 1$ (otherwise, we can add isolated vertices without changing the answer). 

We set $\ell=(hm+1)(2m+1)n$ and define $L=2\ell$. 
We construct the graph $G'$ with $g(G')=L$. In the construction, we use two blocks, called \emph{edge} and \emph{vertex selection gadget}, respectively.  
Copies of the edge selection gadgets are constructed for all pairs of indices $i,j\in[h]$ such that $i<j$. 
The crucial property of the gadget is that it allows us to choose exactly $m$ sets of four edge-disjoint shortest cycles corresponding to the selection of one edge in $G$ with its endpoints in $V_i$ and $V_j$. 
We construct a copy of the vertex selection gadget for each $i\in[h]$. 
This gadget has $n$ possibilities to choose sets of pairs of edge-disjoint shortest cycles, and each choice corresponds to the selection of a vertex of $V_i$. 
In the final stage of our reduction, we construct $G'$ using edge and vertex selection gadgets to ensure that $G'$ has $k=2h+4\binom{h}{2}$ edge-disjoint shortest cycles if and only if $G$ has a multicolored clique of size $h$.

\begin{figure}[ht]
 \begin{center}
 \scalebox{0.7}{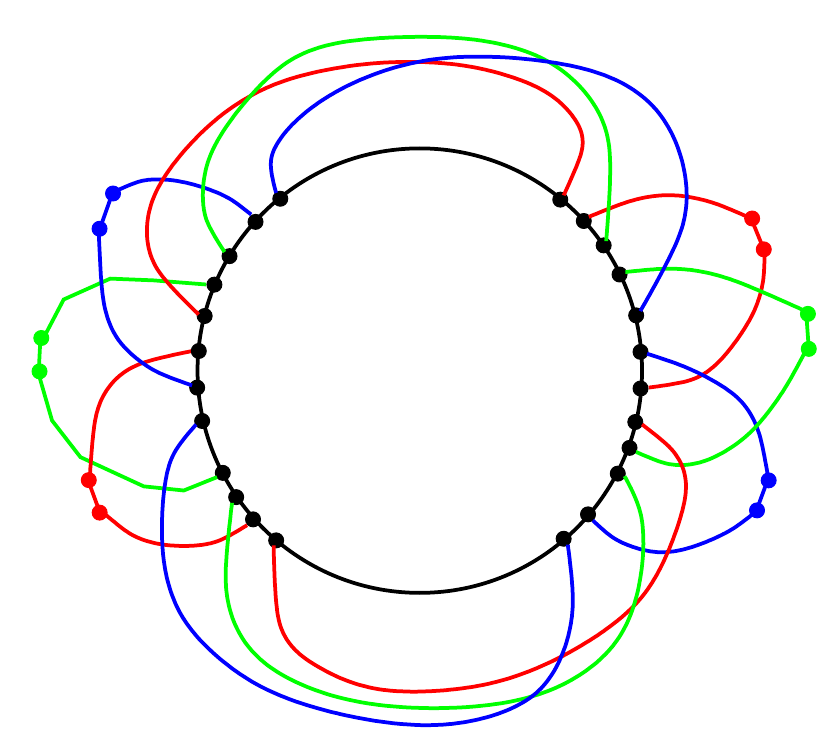}
 \end{center}
 \caption{Construction of the edge selection gadget. Notice that because $2m-1<\ell$, 
 $P_m$ has one endpoint on the shortest path in $C$ between the endpoints of $Q_1$ and, symmetrically, 
 $P_m'$ has one endpoint on the shortest path in $C$ between the endpoints of $Q_1'$.
 }\label{fig:esg}
 \end{figure}

\subparagraph{Edge selection gadget.} We construct the edge selection gadget $F$ as follows (see~\Cref{fig:esg}).
\begin{itemize}
\item[(i)] Construct a cycle $C=v_0\dots v_{p-1}v_0$ of length $p=L+4m+2$; we assume that for $i\geq p$, $v_i=v_{i\mod p}$. 
\item[(ii)] For each $i\in[m]$, construct 
\begin{itemize}
\item a $(v_{2i-1},v_{2i+2m-2})$-path $P_i$ of length $L-2m+1$,
\item a $(v_{2i+2m-1},v_{2i+2m+\ell-1})$-path $Q_i$ of length $\ell$,
\item a $(v_{2i+2m+\ell},v_{2i+4m+\ell-1})$-path $P_i'$ of length  $L-2m+1$,
\item a $(v_{2i+4m+\ell},v_{2i+4m+2\ell})$-path $Q_i'$ of length $\ell$.
\end{itemize}
\end{itemize}
Notice that for each $i\in[m]$, the length of $P_i$ and $P_i'$ is odd. 
We select the middle edges of these paths and denote them by $e_i$ and $e_i'$, respectively. We say that the $m$-tuples $(e_1,\ldots,e_m)$ and $(e_1',\ldots,e_m')$ are \emph{roots} of $F$. 
For $i\in[m]$, we say that $e_i$ and $e_i'$ are opposite.
For every $i\in [m]$, we denote by $C_i,C_i',S_i,S_i'$ the cycles obtained from the paths $P_i,P_i',Q_i,Q_i'$, respectively, by connecting their endpoints by the shortest paths in the cycle $C$.

The properties of the edge selection gadget are summarized in the following lemma.

\begin{lemma}\label{lem:esg}
The edge selection gadget $F$ has the following properties:
\begin{itemize}
\item[(i)] $\Delta(F)\leq 3$ and the endpoints of the root edges are of degree two,
\item[(ii)] for every two endpoints $x$ and $y$ of distinct root edges, the distance between $x$ and $y$ is at least $L-2m+1$,
\item[(iii)] $g(F)=L$, and for each $i\in[m]$, $C_i,C_i',S_i,S_i'$ are shortest cycles,
\item[(iv)] $F$ contains at most four edge-disjoint shortest cycles, and $\{C_i,C_i',S_i,S_i'\}$ for $i\in[m]$ are the only packings of four edge-disjoint shortest cycles,
\item[(v)] for each $i\in[m]$, there is a packing of four edge-disjoint shortest cycles such that 
any root edge except $e_i$ and $e_i'$ is not contained in the cycles of the packing, and
for any packing of four edge-disjoint shortest cycles, there is a unique $i\in[m]$ such that 
the opposite roots $e_i,e_i'$ are in the cycles of the packing and all other roots are not included in the cycles. 
\end{itemize}
\end{lemma}

\begin{proof}
Property (i) immediately follows from the construction of $F$. We have (ii) because the root edges are the middle edges of paths $P_i,P_i'$ for $i\in[m]$. Since these paths are of length $L-2m+1$ and the endpoints of distinct paths are distinct, the distance between $x$ and $y$ is at least $L-2m+1$ for any  two endpoints $x$ and $y$ of distinct root edges. 

To see (iii), note that for all $i\in[m]$, the length of the cycles $C_i,C_i',S_i,S_i'$ is $L$ by the definition of the cycles. To show that $g(F)=L$, observe that the length of $C$ is $p=L+4m+2>L$. Therefore, for any shortest cycle $R$, there is $i\in[m]$ such that $R$ contains an edge of one of the paths $P_i,P_i',Q_i,Q_i'$. Suppose that $R$ contains an edge of $P_i$. Notice that the distance between the endpoints of $P_i$ is $2m-1$ by the construction of $F$. Moreover, the shortest path between the endpoints of $P_i$ is unique, and it is the shortest path between these vertices in $C$. Thus, $R=C_i$.  
Similarly, if $R$ contains an edge of $Q_i$ then $R=S_i$. The cases when $R$ contains an edge of $P_i'$ or $Q_i'$ are symmetric. This proves that $g(F)=L$. Moreover, we proved that $\bigcup_{i=1}^m\{C_i,C_i',S_i,S_i'\}$ is the set of all shortest cycles in $F$. 

To prove (iv), observe that for each $i\in[m]$, $\{C_i,C_i',S_i,S_i'\}$ is a packing of four edge-disjoint cycles. Notice that the cycles $C_i$ and $C_j$ for $i,j\in[m]$ have common edges by the construction of $F$, and the same holds for $C_i'$ and $C_j'$, $S_i$ and $S_j$, and $S_i'$ and $S_j'$. 
Thus, every packing of four edge-disjoint shortest cycles is $\{C_{i_1},C_{i_2}',S_{i_3},S_{i_4}'\}$ for some $i_1,i_2,i_3,i_4\in[m]$, and $F$ has no packing of size five.
If $i_1>i_3$, then $C_{i_1}$ and $S_{i_3}$ share an edge by the construction of $F$. Therefore, $i_1\leq i_3$. Similarly, we have that $i_3\leq i_2$ and $i_2\leq i_4$, that is, $i_1\leq i_3\leq i_2\leq i_4$. If $i_4>i_1$, then $C_{i_1}$ and $S_{i_4}'$ have a common edge.  
We conclude that $i_1=i_2=i_3=i_4$. This proves that $\{C_i,C_i',S_i,S_i'\}$ for $i\in[m]$ are the only packings of four edge-disjoint shortest cycles.

Property (iv) immediately implies (v) 
and this concludes the proof.  
\end{proof}

\begin{figure}[ht]
 \begin{center}
 \scalebox{0.7}{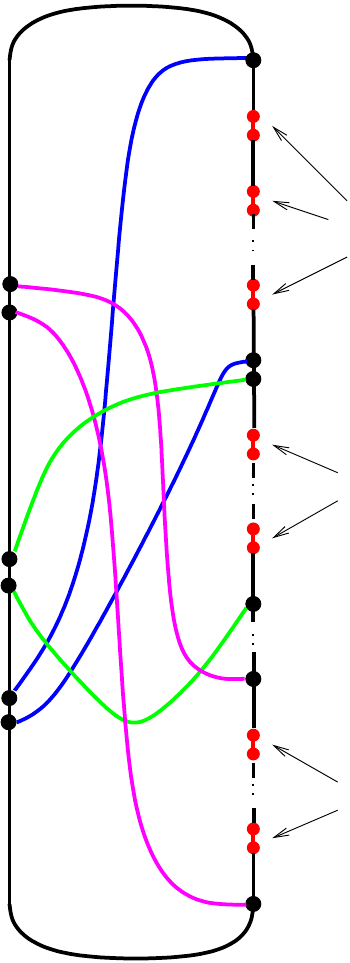}
 \end{center}
 \caption{Construction of the vertex selection gadget. The slot edges are shown in red. Notice that because $(hm+1)(2m+1)(n-1)<\ell$, 
 $P_n$ has one endpoint on the shortest path in $C$ between the endpoints of $P_1'$.
 }\label{fig:vsg}
 \end{figure}

\subparagraph{Vertex selection gadget.} Next we construct the vertex selection gadget $H$ (see~\Cref{fig:vsg}). Its construction follows the ideas of the hardness proof in~\cite[Theorem~1.2]{psc}.
\begin{itemize}
    \item[(i)] Construct a cycle $C=u_0\dots u_{q-1}u_0$ of length $q=(hm+1)(2m+1)+L$; we assume that for $i\geq q$, $u_i=u_{i\mod q}$,
    \item[(ii)] set $r=(hm+1)(2m+1)$, and for each $i\in[n]$, 
    \begin{itemize}
        \item construct a $(u_{ri-1},u_{ri-1+\ell})$-path $P_i$ of length $\ell$, 
        \item construct a $(u_{ri+\ell},u_{ri+L})$-path $P_i'$ of length $\ell$.
    \end{itemize}
\end{itemize}
For each $i\in[n]$, denote by $Q_i$ the shortest $(u_{r(i-1)},u_{ri-1})$-path in $C$. Each $Q_i$ is of length $r-1=(hm+1)(2m+1)-1$. For every $i\in[n]$, we denote by $S_i$ the set of $hm$ edges of $Q_i$ at distance $2m$ from each other and the endpoints of $Q_i$. We say that the edges of $S_i$ are \emph{slot} edges. For each $i\in[n]$, we denote by $C_i$ ($C_i'$, respectively) the cycle obtained from $P_i$ ($P_i'$, respectively) by connecting its endpoints by the shortest path between them in $C$.
Given a packing of two edge-disjoint shortest cycles, we say that a slot edge $e$ is \emph{free} if it is not included in the cycles of the packing.

We summarize the properties of the vertex selection gadget in the next lemma. 


\begin{lemma}\label{lem:vsg}
The vertex selection gadget $H$ has the following properties:
\begin{itemize}
\item[(i)] $\Delta(H)\leq 3$ and the endpoints of the slot edges are of degree two,
\item[(ii)] for every two endpoints $x$ and $y$ of distinct slot edges, the distance between $x$ and $y$ is at least $2m$,
\item[(iii)] $g(H)=L$, and for each $i\in[n]$, $C_i,C_i'$ are shortest cycles,
\item[(iv)] $H$ contains at most two edge-disjoint shortest cycles, and $\{C_i,C_i'\}$ for $i\in[n]$ are the only packings of two edge-disjoint shortest cycles,
\item[(v)] for every $i\in[n]$, there is a packing of two shortest cycles in $H$ such that the slot edges in $S_i$ are free, and for each packing of two shortest cycles in $H$, there is a unique $i\in[n]$ such that $S_i$ contains free slots.
\end{itemize}
\end{lemma}

\begin{proof}
Properties (i) and (ii) follow immediately from the construction of $H$.

To show (iii), observe that for each $i\in[n]$, the length of $C_i$ and $C_i'$ is $L$. Consider an arbitrary shortest cycle $R$ in $H$. Since the length of $C$ is $q=(hm+1)(2m+1)+L>L$, $R$ contains an edge $e$ outside $C$. Then there is $i\in [n]$ such that either $e\in E(P_i)$ or $e\in E(P_i')$. Suppose that $e\in E(P_i)$ (the case $e\in E(P_i')$ is symmetric). Then $R$ contains $P_i$. Notice that the shortest path between the endpoints of $P_i$ in $C$ is the unique shortest path between these vertices in $H$. Thus, $R=C_i$. By symmetry, $R=C_i'$ if $e\in E(P_i')$. This proves (iii). Also, we obtain that 
$\bigcup_{i=1}^n\{C_i,C_i'\}$ is the set of shortest cycles of $H$.

To prove (iv), notice that $C_i$ and $C_i'$ are edge-disjoint shortest cycles for $i\in[n]$. Then each $\{C_i,C_i'\}$ is a packing of two edge-disjoint shortest cycles. To show that we have no other packings of size at least two, observe
that for distinct $i,j\in[n]$, $C_i$ and $C_j$ have common edges by the construction of $H$, and the same holds for $C_i'$ and $C_j'$. Then any packing of edge-disjoint shortest cycles is of the form $\{C_i,C_j'\}$ for some $i,j\in [n]$. However, if $i\neq j$, $C_i$ and $C_j'$ have common edges by the construction of $H$. This proves (iv).

Finally, (iv) together with the construction of $H$ immediately imply (v). This concludes the proof.
\end{proof}

Now we are ready to construct $G'$ and complete the proof of \Cref{thm:hard_generalgraphs}.

\begin{figure}[ht]
 \begin{center}
 \scalebox{0.7}{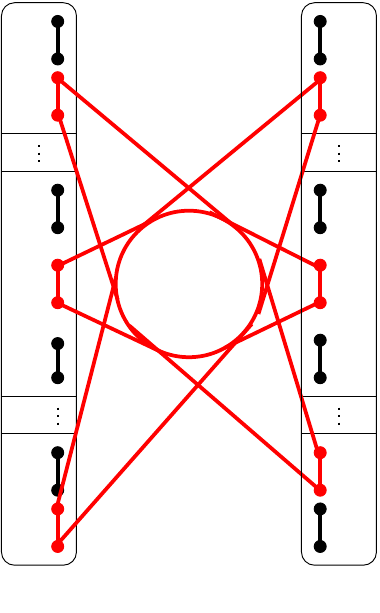}
 \end{center}
 \caption{Construction of $G'$.
 }\label{fig:equiv}
 \end{figure}

\subparagraph{Construction of $G'$.} Consider the instance $(G,V_1,\ldots,V_h)$ of \mcc. 
Let $V_i=\{x_1^i,\ldots,x_n^i\}$. For $i,j\in [h]$ such that $i<j$, let $f_1^{(i,j)},\ldots,f_m^{(i,j)}$ be the edges with their endpoints in $V_i$ and $V_j$.
We construct $G'$ as follows.
\begin{itemize}
    \item[(i)] For each $i\in [h]$, construct a copy $H_i$ of the vertex selection gadget $H$. Denote by $S_1^i,\ldots,S_n^i$ the sets of slot edges in $H_i$. 
    \item[(ii)] For each pair of indices $i,j\in [h]$ such that $i<j$, construct a copy $F_{i,j}$ of the edge selection gadget $F$. 
    \item[(iii)] For each pair of indices $i,j\in [h]$ such that $i<j$ and each $t\in[m]$, consider 
    $f_t^{(i,j)}$ and assume that $f_t^{(i,j)}=x_a^ix_b^j$ for some $a,b\in[n]$. Then we identify the root edge $e_t$ of $F_{i,j}$ and an arbitrary slot edge in $S_a^i$ that was not identified with any root edge in the previous steps, and identify the opposite root edge $e_t'$ of $F_{i,j}$ and an arbitrary slot edge in $S_b^j$ that was not identified with any root edge before.
\end{itemize}
The construction of $G'$ is shown in~\Cref{fig:equiv}
Notice that because $\Delta(G)\leq hm$ and $|S_a^i|=|S_b^j|=hm$, we always can find unused slot edges in Step~(iii) to perform the identifications. 

We summarize the properties of $G'$ in the next lemma.

\begin{lemma}\label{lem:Gprime}
The constructed graph $G'$ has the following properties:
\begin{itemize}
\item[(i)] $\Delta(G')\leq 3$,
\item[(ii)] $g(G')=L$, and for each shortest cycle $R$, there is either $i\in[h]$ such that $R$ is a shortest cycle in $H_i$ or there are $i,j\in [h]$ with $i<j$ such that $R$ is a shortest cycle in $F_{i,j}$. 
\end{itemize}
\end{lemma}

\begin{proof}
Property (i) immediately follows from Claims~(i) of \Cref{lem:esg} and \Cref{lem:vsg}. To show (ii), note that $G'$ has cycles of length $L$ that are cycles of the vertex and edge selection gadgets. Moreover, by Claims~(iii) of \Cref{lem:esg} and \Cref{lem:vsg}, these cycles are shortest cycles in the gadgets. Consider a cycle $R$ of $G'$ that is a shortest chordless cycle for which there is no vertex of edge selection gadget containing $R$. We claim that the length of $R$ is at least $L+1$.

Suppose that $R$ contains vertices of at least two copies of the vertex selection gadget. Then, by the construction of $G'$, $R$ contains at least two edge-disjoint paths in some edge selection gadgets 
$F_{i,j}$ and $F_{i,j'}$ whose endpoints are endpoints of distinct root edges of the gadgets. By Claim~(ii) of \Cref{lem:esg}, the length of each of these paths is at least $L-2m+1$. Then the length of $R$ is at least $2(L-2m+1)>L$ by the  choice of $L$. 
Assume that $R$ contains vertices of one vertex selection gadget $H_i$ for some $i\in[h]$. 
Then $R$ contains vertices of one edge selection gadget $F_{i,j}$ for some $j \in [h]$. As $E(R)\setminus E(H_i)\neq\emptyset$ and $E(R)\setminus E(F_{i,j})\neq\emptyset$, $R$ is the union of edge-disjoint paths in $H_i$ and $F_{i,j}$, and the endpoints of these paths in $F_{i,j}$ are the endpoints of the root edges identified with slot edges of $H_i$. 
Since $R$ is a chordless cycle, $R$ does not contain a path in $F_{i,j}$ whose endpoints are the endpoints of the same root edge. If $R$ has two edge-disjoint paths whose endpoints are endpoints of distinct root edges of $F_{i,j}$, the length of $R$ is at least  $2(L-2m+1)>L$ by Claim~(ii) of \Cref{lem:esg}. Suppose that $R$ contains one path in $F_{i,j}$ whose endpoints are endpoints of distinct root edges. Then $R$ is the union of this paths and a path $H_i$ with the same endpoints. By Claims~(ii) of \Cref{lem:esg} and \Cref{lem:vsg}, the length of $R$ is at least $(L-2m+1)+2m>L$. This finishes the case analysis and proves that the length of $R$ is at least $L+1$.

Since any shortest cycle in $G'$ has no chord, we obtain that $g(G')=L$, and each shortest cycle is a shortest cycle of either a vertex or an edge selection gadget. 
This concludes the proof.
\end{proof}

Now we are ready to complete the proof of \Cref{thm:hard_generalgraphs} by showing that the 
instance $(G,V_1,\ldots,V_h)$ of \mcc is equivalent to the constructed instance $(G',k)$ of \dsce.

\begin{lemma}\label{lem:equiv}
The graph $G$ has a multicolored clique if and only if $G'$ has a packing of $k$ edge-disjoint shortest cycles.
\end{lemma}

\begin{proof}
Suppose that $\{x_{i_1}^1,\ldots,x_{i_h}^h\}$ is a clique of $G$. We construct the packing $\mathcal{P}$ of edge-disjoint shortest cycles of $G'$ as follows. Recall that $g(G')=L$ by \Cref{lem:Gprime}. We construct $\mathcal{P}$ by using packings of cycles of length $L$ in the vertex and edge selection gadgets.
By Claim~(v) of \Cref{lem:vsg}, for each $j\in[h]$, there is a packing of two edge-disjoint cycles in $H_j$ such that the slots of $S^j_{i_j}$ are free. We include in $\mathcal{P}$ these cycles.
Consider arbitrary $j,r\in [h]$ such that $j<r$. 
Because $x^j_{i_j}$ and $x^r_{i_r}$ are adjacent in $G$, by the construction of $G'$, there is a pair of two opposite root edges $e_t$ and $e_t'$ of $F_{j,r}$ such that $e_t$ is identified with an edge of $S^j_{i_j}$ and $e_t'$ is identified with an edge of $S^r_{i_r}$. 
By Claim~(v) of \Cref{lem:esg}, $F_{j,r}$ has a packing of four edge-disjoint shortest cycles that have common edges only with the slot edges of $S^j_{i_j}$ and $S^r_{i_r}$, and do not contain other slot edges of $H_j$ and $H_r$, then we take these four cycles in $\mathcal{P}$. By doing this for all pairs 
$j,r\in[h]$ such that $j<r$, we obtain $\mathcal{P}$ that contains $2h+4\binom{h}{2}=k$ edge-disjoint shortest cycles.

For the opposite direction, suppose that $G'$ has a packing of $k=2h+4\binom{h}{2}$ edge-disjoint shortest cycles. 
By the second claim of \Cref{lem:Gprime}, each shortest cycle of $G'$ is a shortest cycle of either a vertex or an edge selection gadget. 
By Claims~(iv) of \Cref{lem:esg} and \Cref{lem:vsg}, each vertex selection gadget contains at most two edge-disjoint shortest cycles, and each edge selection gadget contains at most four edge-disjoint shortest cycles. 
Since $\mathcal{P}$ is of size $2h+4\binom{h}{2}$, $\mathcal{P}$ contains exactly two cycles from each vertex selection gadget and exactly four cycles from each edge selection gadget. 
By Claim~(v) of \Cref{lem:vsg}, for each $j\in[h]$, there is a unique $i_j\in[n]$ such that $S^{j}_{i_j}$ contains free slots with respect to $\mathcal{P}$. 
    We claim that $\{x_{i_1}^1,\ldots,x_{i_h}^h\}$ is a clique of $G$.

To show the claim, consider arbitrary $j,r\in[h]$ such that $j<r$. 
We prove that $x_jx_r\in E(G)$. To see this, recall that $F_{j,r}$ contains four edge-disjoint shortest cycles from $\mathcal{P}$. By Claim~(v) of \Cref{lem:esg}, there is a unique $t\in[m]$ such that 
the opposite roots $e_t,e_t'$ of $F_{j,r}$ are in the cycles of the packing $\mathcal{P}$, and all other roots are not included in the cycles. This implies that $e_t$ and $e_t'$ are identified with the free slots of $H_j$ and $H_r$. Since only $S_{i_j}^j$ and $S_{i_r}^r$ contain free slots, we have that $e_t$ is identified with an edge of $S_{i_j}^j$, and $e_t'$ is identified with an edge of $S_{i_r}^r$. Since $e_t$ and $e_t'$ are opposite roots of $F_{j,r}$, this means that $G$ contains the edge $x_{i_j}^jx_{i_r}^r$. Because $j,r\in[h]$ are arbitrary, we have that the vertices of  
$\{x_{i_1}^1,\ldots,x_{i_h}^h\}$ are pairwise adjacent, that is, they form a multicolored clique. This concludes the proof.
\end{proof}

Taking into account that $G'$ can be constructed in polynomial time and $k=2h+4\binom{h}{2}$, \Cref{lem:equiv} proves \Cref{thm:hard_generalgraphs}.

\section{FPT algorithm for Planar \maxddisjointcoveragefull} \label{sec:maxddisjointcoverage}

In this section we prove~\Cref{thm:maxddisjointcoverage}, by obtaining an algorithm with the desired running time bound via a dynamic programming
over a structure from \cite{psc} called the {\em laminar shortest cycles tree}.
The proof is technical, and \Cref{fig:thm-maxd-dependency} summarizes the dependencies between the main steps of the proof. 
We remark that our problem is not directly amenable to a dynamic programming over this structure alone: the recurrence relies on a second, auxiliary dynamic programming over a graph of shortest cycles, layered on top of the tree recursion. 
We will begin by recalling the notion of {\em laminar shortest cycles tree} for a given plane graph and some useful results regarding them from~\cite{psc}.

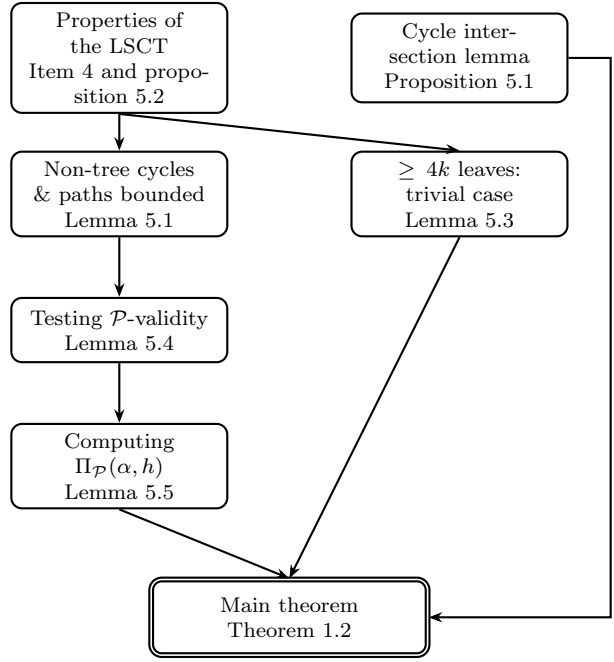
\begin{figure}[t]
\centering
\usetikzlibrary{arrows.meta}
\begin{tikzpicture}[
    >={Stealth[length=1.8mm]},
    thick,
    every node/.style={draw, rounded corners, align=center,
                      font=\scriptsize, text width=2.6cm, minimum height=0.85cm}
]

\node (lsct)  at (1.5,0)
  {Properties of the LSCT\\\Cref{prop:def:laminar-tree,prop:lsct-use-property}};
\node (inter) at (6,0)
  {Cycle intersection lemma\\\Cref{claim:intersection-between-two-cycles}};

\node (bound) at (1.5,-1.8)
  {Non-tree cycles \& paths bounded\\\Cref{lem:non-tree-cycles-are-bounded}};
\node (leaves) at (6,-1.8)
  {$\geq 4k$ leaves: trivial case\\\Cref{lem:4k-leaves}};

\node (test) at (1.5,-3.6)
  {Testing $\Pcal$-validity\\\Cref{obs:test-h-validity}};

\node (compute) at (1.5,-5.4)
  {Computing $\Pi_{\Pcal}(\alpha,h)$\\\Cref{lem:compute-pi}};

\node (thm) at (3.75,-7.4) [double, text width=3.4cm, minimum height=1cm]
  {Main theorem\\\Cref{thm:maxddisjointcoverage}};

\draw[->] (lsct.south)  -- (bound.north);
\draw[->] (lsct.south)  -- (leaves.north);
\draw[->] (bound.south) -- (test.north);
\draw[->] (test.south)  -- (compute.north);
\draw[->] (leaves.south)   -- (thm.north);
\draw[->] (compute.south)  -- (thm.north);
\draw[->] (inter.east) -- (8,0) -- (8,-7.4) -- (thm.east);

\end{tikzpicture}
\caption{Dependency structure of the proof of \Cref{thm:maxddisjointcoverage}.
The top row, together with \Cref{claim:intersection-between-two-cycles}, is
imported from~\cite{psc}. \Cref{lem:compute-pi} is itself established in
\Cref{subsec:compute-pi} via a further dynamic program (see the companion
figure).}
\label{fig:thm-maxd-dependency}
\end{figure}

\subsection{Laminar Shortest Cycles Tree}\label{sec:prelim-lsct}
We start by  recalling the terminology and basic properties of  {\em laminar shortest cycles tree} from~\cite{psc}; we slightly rephrase some of them for convenience. We begin by recalling how two shortest cycles may interact in a planar graph. 

\begin{proposition}[Lemma 5.1~\cite{psc}] \label{claim:intersection-between-two-cycles}
Consider a weighted planar graph $G$ and two shortest cycles $C,C'$ in $G$, where $|V(C)\cap V(C')| \geq 2$. Then exactly one of the following holds (see~\Cref{fig:touching-crossing}): i) $C\cap C'$ is a subpath of $C$ (resp. $C'$), or ii) $E(C)\cap E(C') =\emptyset$ and $V(C)\cap V(C')$ has exactly two vertices, say $s,t$, where $s$ and $t$ are at distance exactly $g(G)/2$ in $G$ (and also in $C$ and $C'$). 
\end{proposition}
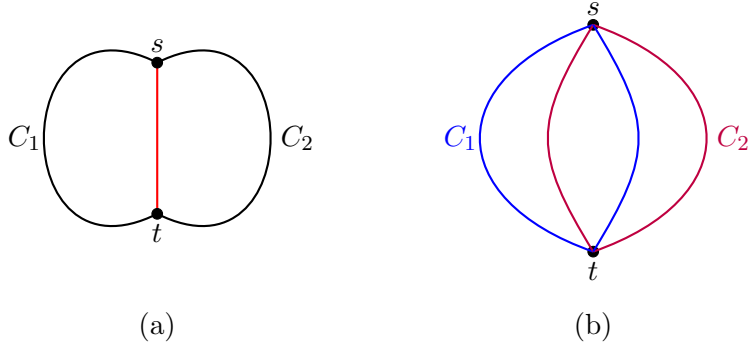
\begin{figure}[t]
\centering

\begin{tikzpicture}[scale=1]
\draw[red,thick] (0,1) -- (0,-1);

\filldraw (0,1) circle (2pt);
\filldraw (0,-1) circle (2pt);

\node[above] at (0,1) {$s$};
\node[below] at (0,-1) {$t$};

\draw[thick]
(0,1)
.. controls (-2,2) and (-2,-2) ..
(0,-1);

\draw[thick]
(0,1)
.. controls (2,2) and (2,-2) ..
(0,-1);

\node[left]  at (-1.4,0) {$C_1$};
\node[right] at (1.5,0) {$C_2$};

\node at (0,-2.5) {(a)};
\end{tikzpicture}
\hspace{1.2cm}
\begin{tikzpicture}[scale=1]

\filldraw (0,1.5) circle (2pt);
\filldraw (0,-1.5) circle (2pt);

\node[above] at (0,1.5) {$s$};
\node[below] at (0,-1.5) {$t$};

\draw[blue,thick]
(0,1.5)
.. controls (-2,0.8) and (-2,-0.8) ..
(0,-1.5); 

\draw[purple,thick]
(0,-1.5)
.. controls (2,-0.8) and (2,0.8) ..
(0,1.5);

\draw[purple,thick]
(0,1.5)
.. controls (-0.8,0.2) and (-0.8,-0.2) ..
(0,-1.5);

\draw[blue,thick]
(0,-1.5)
.. controls (0.8,-0.2) and (0.8,0.2) ..
(0,1.5); 

\node[blue,left] at (-1.4,0) {$C_1$};
\node[purple,right] at (1.5,0) {$C_2$};

\node at (0,-2.5) {(b)};
\end{tikzpicture}

\caption{Interactions between shortest cycles in a planar graph:
(a) two cycles that share a common path (shown in red color);
(b) two cycles intersecting exactly at two vertices $s$ and $t$.} 
\label{fig:touching-crossing}
\end{figure}

Recall a plane graph is a planar graph whose embedding is fixed. Consider a plane graph $G$ and two distinct cycles $C, C'$ in $G$. Also, each edge $e$ of $G$ is drawn as a curve in the plane, and each cycle defines a closed Jordan curve, where $\Int(C)$ is the open bounded region enclosed by the Jordan curve of $C$ (i.e., the region not containing the outer face of $G$).

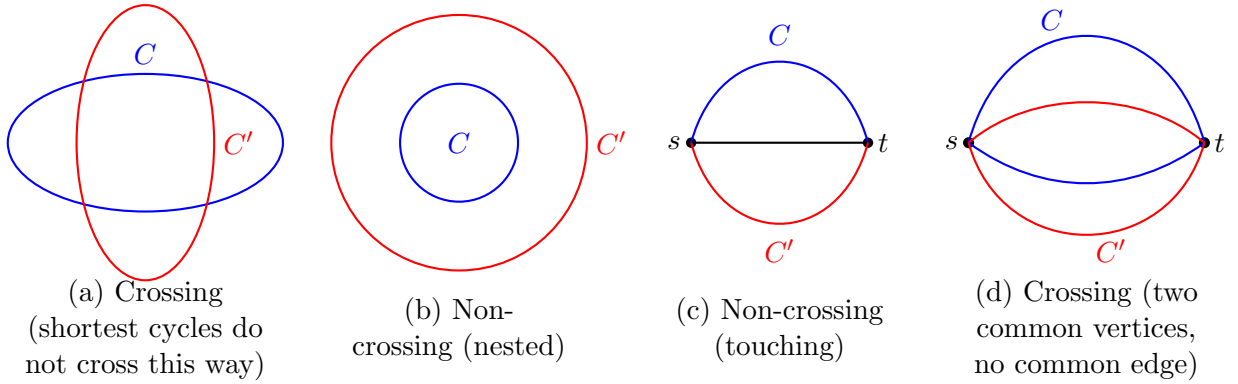
\begin{figure}[t]
\centering
\begin{tikzpicture}[scale=1.3]

\begin{scope}[xshift=-3.5cm]
  \draw[thick, blue] (0,0) ellipse (1.4cm and 0.7cm);
  \node[blue] at (0, 0.9) {$C$};
  \draw[thick, red] (0,0) ellipse (0.7cm and 1.4cm);
  \node[red] at (0.95, 0) {$C'$};
 \node[text width=3.4cm, align=center] at (0,-1.9){(a) Crossing (shortest cycles do not cross this way)};
\end{scope}

\hspace{-4mm}

\begin{scope}[xshift=0cm]
  \draw[thick, red] (0,0) ellipse (1.3cm and 1.3cm);
  \node[red] at (1.55, 0) {$C'$};
  \draw[thick, blue] (0,0) ellipse (0.6cm and 0.6cm);
  \node[blue] at (0, 0) {$C$};
  \node at (0, -1.9)[text width=3.4cm, align=center] {(b) Non-crossing (nested)};
\end{scope}

\hspace{-7mm}

\begin{scope}[xshift=3.8cm]
  \fill (-0.9, 0) circle (1.5pt) node[left] {$s$};
  \fill ( 0.9, 0) circle (1.5pt) node[right] {$t$};

  \draw[thick, black] (-0.9, 0) -- (0.9, 0);

  \draw[thick, blue]
      (-0.9, 0) .. controls (-0.6, 1.1) and (0.6, 1.1) .. (0.9, 0);
  \node[blue] at (0, 1.1) {$C$};

  \draw[thick, red]
      (-0.9, 0) .. controls (-0.6, -1.1) and (0.6, -1.1) .. (0.9, 0);
  \node[red] at (0, -1.1) {$C'$};

  \node at (0, -1.9)[text width=3.4cm, align=center] {(c) Non-crossing (touching)};
\end{scope}

\hspace{12mm}

\begin{scope}[xshift=6.0cm]
  \fill (-1.2,0) circle (1.6pt) node[left] {$s$};
  \fill ( 1.2,0) circle (1.6pt) node[right] {$t$};

  \draw[thick, blue]
    (-1.2,0) .. controls (-0.8,1.45) and (0.8,1.45) .. (1.2,0)
    node[midway, above] {};

  \draw[thick, blue]
    (-1.2,0) .. controls (-0.45,-0.55) and (0.45,-0.55) .. (1.2,0)
    node[midway, below] {};

  \draw[thick, red]
    (-1.2,0) .. controls (-0.55,0.55) and (0.55,0.55) .. (1.2,0)
    node[midway, above] {};

  \draw[thick, red]
    (-1.2,0) .. controls (-0.85,-1.25) and (0.85,-1.25) .. (1.2,0)
    node[midway, below] {};

  \node[blue] at (-0.35,1.25) {$C$};
  \node[red] at (0.25,-1.1) {$C'$};

  \node[text width=3.4cm, align=center] at (0,-1.9)
    {(d) Crossing (two common vertices, no common edge)};
\end{scope}

\end{tikzpicture}
\caption{Illustration of (some) crossing and non-crossing cycles in a plane graph. Two shortest cycles can cross exactly one way as illustrated in (d). 
}
\label{fig:crossing-noncrossing}
\end{figure}

\begin{definition}[Crossing cycles]\label{def:crossing}
Let $G$ be a plane graph and let $C, C'$ be two distinct cycles in $G$. The cycles $C$ and $C'$ \emph{cross} if both of the following conditions hold (see \Cref{fig:crossing-noncrossing}):
    \begin{enumerate}[label=(\roman*)]
        \item there is $e \in E(C)$ such that every interior point of the curve of $e$ lies in $\Int(C')$, and
        \item there is $e' \in E(C')$ such that every interior point of the curve of $e'$ lies in $\Int(C)$.
    \end{enumerate}
    If $C$ and $C'$ do not cross, we say they are \emph{non-crossing}.
\end{definition}

\begin{definition}[Laminar family of cycles]\label{def:laminar-family}
    Let $G$ be a plane graph. A family $\mathcal{C}$ of cycles in $G$ is 
    \emph{laminar} if no two distinct cycles in $\mathcal{C}$ cross. 
    Equivalently, for any two distinct cycles $C_1, C_2 \in \mathcal{C}$, 
    exactly one of the following holds: $\Int(C_1) \subseteq \Int(C_2)$, 
    $\Int(C_2) \subseteq \Int(C_1)$, or $\Int(C_1) \cap \Int(C_2) = \emptyset$.
\end{definition} 

We will next see how the shortest cycles in a weighted plane graph can be categorized into two types, \emph{splittable} and \emph{unsplittable}---we state their definitions as implied by Lemma~5.4 of~\cite{psc}, as opposed to the original definitions given in~\cite{psc}. 
We rephrase some of these definitions in terms of notions we call \emph{longitudinal} paths and cycles. 

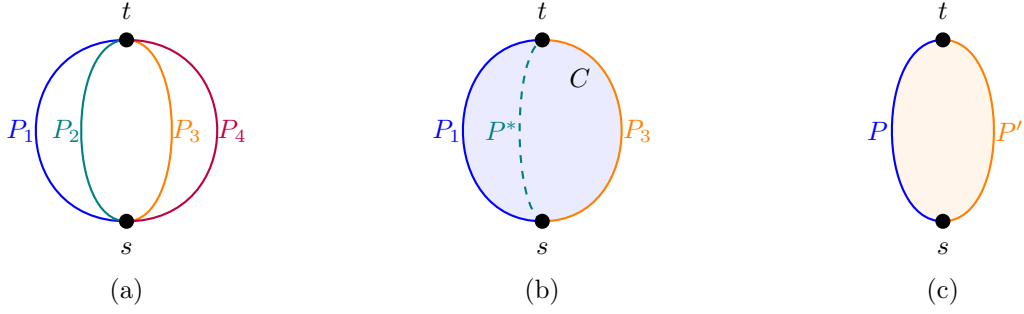
\begin{figure}[t]
\centering
\begin{tikzpicture}[scale=1.0,
  every node/.style={font=\small},
  pole/.style={circle, fill=black, inner sep=2pt},
  lpath/.style={thick, draw},
  internal/.style={thick, dashed, draw}
]

\begin{scope}[xshift=0cm]
  \coordinate (s) at (0,-1.2);
  \coordinate (t) at (0, 1.2);

  \draw[lpath, blue]   (s) .. controls (-1.6,-1.2) and (-1.6,1.2) .. (t);
  \draw[lpath, teal]   (s) .. controls (-0.8,-1.2) and (-0.8,1.2) .. (t);
  \draw[lpath, orange] (s) .. controls ( 0.8,-1.2) and ( 0.8,1.2) .. (t);
  \draw[lpath, purple] (s) .. controls ( 1.6,-1.2) and ( 1.6,1.2) .. (t);

  \node[pole] at (s) {};
  \node[pole] at (t) {};
  \node[below=4pt] at (s) {$s$};
  \node[above=4pt] at (t) {$t$};

  \node[left=2pt]  at (-1.0, 0) {\textcolor{blue}{$P_1$}};
  \node[left=2pt]  at (-0.4, 0) {\textcolor{teal}{$P_2$}};
  \node[right=2pt] at ( 0.4, 0) {\textcolor{orange}{$P_3$}};
  \node[right=2pt] at ( 1.0, 0) {\textcolor{purple}{$P_4$}};

  \node[below=6pt] at (0,-1.6) {(a)};
\end{scope}

\begin{scope}[xshift=5.5cm]
  \coordinate (s) at (0,-1.2);
  \coordinate (t) at (0, 1.2);

  \fill[blue!8]
    (s) .. controls (-1.4,-1.2) and (-1.4,1.2) .. (t)
        .. controls ( 1.4,1.2) and ( 1.4,-1.2) .. cycle;

  \draw[lpath, blue]   (s) .. controls (-1.4,-1.2) and (-1.4,1.2) .. (t)
      node[midway, left=-3.5pt] {$P_1$};
  \draw[lpath, orange] (s) .. controls ( 1.4,-1.2) and ( 1.4,1.2) .. (t)
      node[midway, right=-3.5pt] {$P_3$};

  \draw[internal, teal]
    (s) .. controls (-0.4,-1.0) and (-0.4,1.0) .. (t)
    node[midway, left=-3.5pt] {$P^*$};

  \node[pole] at (s) {};
  \node[pole] at (t) {};
  \node[below=4pt] at (s) {$s$};
  \node[above=4pt] at (t) {$t$};

  \node at (0.5,0.7) {$C$};

  \node[below=6pt] at (0,-1.6) {(b)};
\end{scope}

\begin{scope}[xshift=10.8cm]
  \coordinate (s) at (0,-1.2);
  \coordinate (t) at (0, 1.2);

  \fill[orange!8]
    (s) .. controls (-0.9,-1.2) and (-0.9,1.2) .. (t)
        .. controls ( 0.9,1.2) and ( 0.9,-1.2) .. cycle;

  \draw[lpath, blue]   (s) .. controls (-0.9,-1.2) and (-0.9,1.2) .. (t)
      node[midway, left=-2.5pt] {$P$};
  \draw[lpath, orange] (s) .. controls ( 0.9,-1.2) and ( 0.9,1.2) .. (t)
      node[midway, right=-3.5pt] {$P'$};


  \node[pole] at (s) {};
  \node[pole] at (t) {};
  \node[below=4pt] at (s) {$s$};
  \node[above=4pt] at (t) {$t$};

  \node[below=6pt] at (0,-1.6) {(c)};
\end{scope}

\end{tikzpicture}
\caption{Illustration of \Cref{def:longitudinal}. 
(a) A pole $\{s,t\}$ with four longitudinal paths $P_1, P_2, P_3, P_4$.
(b) A splittable cycle $C = P_1 \cup P_3$ (blue and orange): the dashed green 
path $P^*$ lies strictly inside $C$, witnessing splittability.
(c) Two consecutive longitudinal paths $P$ and $P'$: the bounded region of 
$P \cup P'$ contains no other longitudinal path.}
\label{fig:longitudinal}
\end{figure}

\begin{definition}
\label{def:longitudinal}
Consider a weighted plane graph $G$ (see \Cref{fig:longitudinal} for illustration).

\begin{itemize}

    \item \textbf{Pole pair.} A pair of distinct vertices $s,t\in V(G)$ is a \emph{pole pair}(or simply poles) if there exist at least three pairwise internally vertex-disjoint shortest $s$-$t$ paths in $G$, each of length exactly $g(G)/2$.
    
    \item \textbf{$s$-$t$ longitudinal path.} Let $s,t$ be a pole pair. 
    A shortest $s$-$t$ path in $G$ is an \emph{$s$-$t$ longitudinal path} (or simply a \emph{longitudinal path} when the pole pair is clear from context).
    
    \item \textbf{$s$-$t$ longitudinal cycle.}
    Let $s,t$ be a pole pair. A shortest cycle $C=P_1\cup P_2$, where $P_1$ and $P_2$ are $s$-$t$ longitudinal paths satisfying $V(P_1)\cap V(P_2)=\{s,t\}$, is an \emph{$s$-$t$ longitudinal cycle} (or simply a \emph{longitudinal cycle}). We call $P_1$ and $P_2$ the \emph{longitudinal paths of $C$}.
    
    \item \textbf{$s$-$t$ splittable cycle.}
    An $s$-$t$ longitudinal cycle $C$ is \emph{$s$-$t$ splittable} (or simply \emph{splittable}, when the context is clear) if the bounded region enclosed by $C$ contains a shortest $s$-$t$ path that is internally vertex-disjoint from $V(C)$.
    
    \item \textbf{Unsplittable cycle.}
    A shortest cycle that is not splittable is \emph{unsplittable}.
    
    \item \textbf{Consecutive $s$-$t$ longitudinal paths.}
    Let $P$ and $P'$ be $s$-$t$ longitudinal paths, where $V(P)\cap V(P')=\{s,t\}$. The paths $P$ and $P'$ are \emph{consecutive} if the bounded region enclosed by the cycle $P \cup P'$ contains no other $s$-$t$ longitudinal path. 
\end{itemize}
\end{definition}


Let $G$ be a weighted plane graph. 
We introduce the following partial order on the family of all shortest cycles of $G$. For two shortest cycles $C_1$ and $C_2$, we write $C_1\leq C_2$ if every vertex and edge of~$C_1$ is a vertex or an edge of~$C_2$ or is embedded in the internal face $\Int(C_2)$ of $C_2$.

Consider a pole $\{s,t\}$ in $G$. Following~\cite{psc}, we introduce specific families of paths and cycles as follows.
Let $C = P \cup P'$ be an $s$-$t$ splittable cycle where $P, P'$ are $s$-$t$ longitudinal paths with $V(P) \cap V(P') = \{s,t\}$. Recall that for a shortest cycle $C$ of a plane graph $G$, we denote by $G_C$ the subgraph of $G$ induced by the vertices lying on $C$ or in $\Int(C)$. Let $\Pcal_s(C) = \{P_0, P_1, \dots, P_q\}$ be the inclusion-wise maximal family of distinct shortest $s$-$t$ paths in $G_C$, ordered so that $P_0 \cup P_q = C$, and for each $i \in [q]$ the paths $P_{i-1}$ and $P_{i}$ are consecutive, i.e., no path of $\Pcal_{s}(C)$ is embedded in the interior of the cycle $P_{i-1} \cup P_{i}$ (see~\Cref{fig:Ps-family}). By~\cite[Lemma 5.4]{psc} the paths in $\Pcal_{s}(C)$ are pairwise internally vertex-disjoint  and  this ordering is unique. For paths in $\Pcal_{s}(C)$, we will always follow the ordering described here. Let $\SC{C}_s(C)$ be the set of cycles obtained by taking consecutive pairs of paths in $\Pcal_{s}(C)$. By ~\cite[Lemma 5.5]{psc}, any shortest cycle $C'$ contained in $G_C$ is either the union of two
paths from $\Pcal_{s}(C)$, or is embedded inside one of the
cycles in $\SC{C}_{s}(C)$, i.e., $C' \leq \what{C}$ for some $\what{C} \in \SC{C}_{s}(C)$. 

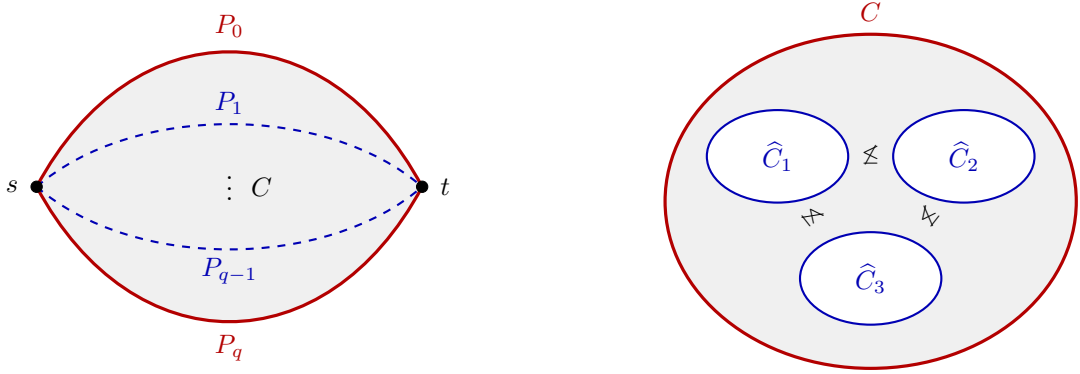
\begin{figure}[t]
\centering
\begin{subfigure}[t]{0.54\textwidth}
\centering
\begin{tikzpicture}[scale=0.85, every node/.style={font=\small}]

  \coordinate (s) at (0,0);
  \coordinate (t) at (6,0);

  \fill[gray!12]
    (s) .. controls (1.5, 2.8) and (4.5, 2.8) .. (t)
        .. controls (4.5,-2.8) and (1.5,-2.8) .. cycle;

  \draw[very thick, red!70!black]
    (s) .. controls (1.5, 2.8) and (4.5, 2.8) .. (t);
  \draw[very thick, red!70!black]
    (s) .. controls (1.5,-2.8) and (4.5,-2.8) .. (t);

  \draw[thick, blue!70!black, dashed]
    (s) .. controls (1.5, 1.3) and (4.5, 1.3) .. (t);

  \draw[thick, blue!70!black, dashed]
    (s) .. controls (1.5,-1.3) and (4.5,-1.3) .. (t);

  \filldraw (s) circle (2.5pt) node[left=3pt] {$s$};
  \filldraw (t) circle (2.5pt) node[right=3pt] {$t$};

  \node[red!70!black]  at (3,  2.5) {$P_0$};
  \node[blue!70!black] at (3,  1.3) {$P_1$};
  \node                at (3,  0.1) {$\vdots$};
  \node[blue!70!black] at (3, -1.3) {$P_{q-1}$};
  \node[red!70!black]  at (3, -2.5) {$P_q$};

  \node at (3.5, 0) {$C$};

\end{tikzpicture}
\caption{Splittable cycle $C{=}P_0{\cup}P_q$ with family
         $\Pcal_{s}(C){=}\{P_0,P_1,\dots,P_q\}$.
         Consecutive pairs $P_i{\cup}P_{i+1}$, $i\in [q-1]$ form $\SC{C}_{s}(C)$.}
\label{fig:Ps-family}
\end{subfigure}
\hfill
\begin{subfigure}[t]{0.43\textwidth}
\centering
\begin{tikzpicture}[scale=0.85, every node/.style={font=\small}]

  \fill[gray!12] (0,0) ellipse (3.2cm and 2.6cm);
  \draw[very thick, red!70!black] (0,0) ellipse (3.2cm and 2.6cm);
  \node[red!70!black] at (0, 2.95) {$C$};

  \fill[white] (-1.45, 0.7) ellipse (1.1cm and 0.72cm);
  \draw[thick, blue!70!black] (-1.45, 0.7) ellipse (1.1cm and 0.72cm);
  \node[blue!70!black] at (-1.45, 0.7) {$\what{C}_1$};

  \fill[white] (1.45, 0.7) ellipse (1.1cm and 0.72cm);
  \draw[thick, blue!70!black] (1.45, 0.7) ellipse (1.1cm and 0.72cm);
  \node[blue!70!black] at (1.45, 0.7) {$\what{C}_2$};

  \fill[white] (0, -1.2) ellipse (1.1cm and 0.72cm);
  \draw[thick, blue!70!black] (0, -1.2) ellipse (1.1cm and 0.72cm);
  \node[blue!70!black] at (0, -1.2) {$\what{C}_3$};

  \node[font=\scriptsize] at (0, 0.72)   {$\not\leq$};
  \node[font=\scriptsize, rotate=-52] at (-0.9, -0.22) {$\not\leq$};
  \node[font=\scriptsize, rotate= 52] at ( 0.9, -0.22) {$\not\leq$};

\end{tikzpicture}
\caption{Unsplittable cycle $C$ with family
         $\SC{C}_{u}(C){=}\{\what{C}_1,\what{C}_2,\what{C}_3\}$.
         Elements are pairwise incomparable and maximal inside $G_C$.}
\label{fig:Cu-family}
\end{subfigure}

\caption{Illustration of the families $\Pcal_{s}(C)$ and $\SC{C}_{s}(C)$
         for a splittable cycle (\emph{left}), and $\SC{C}_{u}(C)$ for an
         unsplittable cycle (\emph{right}).}
\label{fig:two-families}
\end{figure}

For an unsplittable cycle $C$ in a weighted plane graph $G$, let $\SC{C}_{u}(C)$ be the set of all unsplittable shortest cycles
of~$G_C$ distinct from~$C$ that are maximal with respect to~$\leq$;
see \Cref{fig:Cu-family}.
In other words,
\begin{itemize}
    \item[(i)] for every distinct $C_1, C_2 \in \SC{C}_{u}(C)$, neither $C_1 \le C_2$ nor $C_2 \le C_1$, and
    \item[(ii)] for any $C_1 \in \SC{C}_{u}(C)$, there is no unsplittable shortest cycle $C_2$ in $G_C$ other than $C$ such that $C_1 \le C_2$.
\end{itemize}

As is shown in~\cite[Lemma 5.6]{psc}, the family $\SC{C}_{ u}(C)$ is unique, laminar, and every shortest cycle in $G_C$ distinct from $C$ is either in $\SC{C}_{ u}(C)$
or embedded inside some cycle of $\SC{C}_{ u}(C)$.

We are ready to describe the construction of \emph{laminar shortest cycles tree (LSCT)}. 
Let $G$ be a \emph{clean} plane weighted graph. As is shown in ~\cite{psc}, one can find in polynomial time an embedding of $G$ such that the outer face of $G$ is a shortest cycle $C$ of $G$. 
We do not insist that $G$ is connected, but we will assume that all the vertices and edges of $G$ are drawn inside the region bounded by one of its fixed shortest cycle.
    
\begin{definition}[\cite{psc}]
\label{def:laminar-tree}
Let $G$ be a clean weighted plane graph  embedded in such a way that its outer face is a shortest cycle $C_r$.  The \emph{\lsctfull} (LSCT) of $G$, denoted by $\Tcal(G)$, is a tree rooted at $C_r$ whose nodes are (not necessarily all) shortest cycles of $G$, constructed recursively as follows. Initialize $V(\Tcal(G)) = \{C_r\}$, $E(\Tcal(G)) = \emptyset$,
    $\Ccal = \{C_r\}$, and process each $C \in \Ccal$ as follows, repeating
    until $\Ccal = \emptyset$ (see \Cref{fig:lsct-example}):
\begin{tcolorbox}[
    colback=blue!4,
    colframe=blue!35,
    arc=4pt,
    boxrule=0.6pt,
    left=6pt, right=6pt, top=4pt, bottom=4pt
]
    \begin{enumerate}
        \item\label{leaf-def:laminar-tree} If $C$ is the boundary of an internal face in the planar
        embedding, then $C$ is a \emph{leaf} of $\Tcal(G)$; remove $C$
        from $\Ccal$.
        \item\label{split-def:laminar-tree} If $C$ is a splittable cycle, then for each
        $C' \in \SC{C}_{{  s}}(C)$, add $C'$ as a child of $C$ in
        $\Tcal(G)$ and add $C'$ to $\Ccal$; then remove $C$ from $\Ccal$.
        Here $C$ is a \emph{splittable-node} of $\Tcal(G)$.
        \item\label{unsplit-def:laminar-tree} If $C$ is an unsplittable cycle, then for each
        $C' \in \SC{C}_{{  u}}(C)$, add $C'$ as a child of $C$ in
        $\Tcal(G)$ and add $C'$ to $\Ccal$; then remove $C$ from $\Ccal$.
        Here $C$ is an \emph{unsplittable-node} of $\Tcal(G)$.
    \end{enumerate}
\end{tcolorbox}
Whenever the context is clear, we will simply write $\C{T}$, instead of $\C{T}(G)$. A shortest cycle $C$ of $G$ is a \emph{tree cycle} if $C \in V(\Tcal)$, and otherwise it is a \emph{non-tree cycle}. For clarity, we will refer to vertices of $\C{T}$ as nodes. Also, to clearly highlight we are considering a node of $\C{T}$ (which must be a shortest cycle in $G$), we will use letters like $\alpha, \beta, \gamma,\iota$, etc.  
\end{definition}



We now state some useful properties of LSCT obtained in~\cite[Lemma 5.7]{psc}.

\begin{proposition}\label{prop:lsct-use-property}
    Consider a weighted clean plane graph $G$ with $w: E(G) \rightarrow \mathbb{Z}_{>0}$ where a shortest cycle of $G$ is the outer face of the embedding, and its $LSCT$, ${\cal T}$. Then, the following holds:
    \begin{enumerate}
        \item\label{prop:desc-lsct} For tree cycles $C_1, C_2$ in $G$, $C_1 \leq C_2$ if and only if $C_1$ is a descendant of $C_2$ in ${\cal T}$.\footnote{We will follow the convention that a node is a descendant of itself in a rooted tree.}
        
        
        \item\label{prop:cycle-type-lsct} For any shortest cycle $C$ in $G$ exactly one of the following holds: i) $C$ is a tree cycle, or ii) there exist a unique splittable-node $C^*$ of $\Tcal(G)$ and two non-consecutive paths $P_i,P_j \in \Pcal_{s}(C^*) = \{P_0,P_{1},P_2,\dots, P_q\}$ such that $C = P_i \cup P_j$ and $(P_i,P_j)\notin \{(P_0, P_q), (P_q, P_0)\}$.
        
        \item\label{prop:s-node-lsct} For any splittable-node $C$ of $\Tcal(G)$ the paths in $\Pcal_{s}(C)$ are pairwise internally vertex-disjoint. Moreover, for distinct splittable nodes $C,C'$ of $\C{T}(G)$, we have $\C{P}_{\sf s}(C)\cap \C{P}_{\sf s}(C') =\emptyset$.

        \item\label{prop:def:laminar-tree} The laminar shortest cycles tree of $G$ can be computed in polynomial time.
    \end{enumerate}
    
\end{proposition}

\begin{lemma}[Implicit in \cite{psc}]
    \label{lem:non-tree-cycles-are-bounded} 
   For a weighted clean plane graph $G$ with LSCT $\Tcal(G)$, where $t$ is the number of leaves in $\C{T}$, the following holds.
    \begin{enumerate} [label = (\roman*)]
        \item 
        \label{item-1-of-lem:non-tree-cycles-are-bounded}
        The number of non-tree cycles in $G$ is at most $\Ocal(t^2)$.
        \item \label{item-lem:loni-paths-bounded}
        The number of longitudinal paths in $G$ are at most $\Ocal(t)$. Furthermore, we can find these paths in time bounded by $n^{O(1)}$. 
    \end{enumerate}
\end{lemma}
%
%

The next result follows from the definition of laminar shortest cycles tree. 

\begin{lemma}[Implicit in \cite{psc}]
\label{lem:intersection-with-ancestor-is-more}
    Consider a (weighted) clean plane graph $G$ with LSCT $\Tcal$.
    For $\alpha,\beta,\gamma \in V(\C{T})$,
    \footnote{
    We use Greek letters (e.g., $\alpha,\beta,\gamma$) to denote nodes of $\Tcal$(which are also shortest cycles in $G$). 
    Symbols such as $C$ and $C'$ denote arbitrary cycles (either tree or non-tree cycles). 
    } 
    if $\beta$ is a descendant of $\alpha$ in $\Tcal$ and $\gamma$ is not a descendant of $\alpha$, then $E(\beta)\cap E(\gamma)\subseteq E(\alpha)$.
\end{lemma}

Note that by the Four Color Theorem, any planar graph on $4n'$ vertices has an independent set of size at least $n'$. Also, the dual of a planar graph is planar. Recall that the leaves of a laminar shortest cycles trees are the facial cycles of the given plane graph. 
The above gives the next result.   

\begin{lemma}[Implicit in \cite{psc}]
    \label{lem:4k-leaves}
    If LSCT $\C{T}$ of a weighted plane graph $G$ contains at least $4k$ leaves, then there is a set of $k$ pairwise edge-disjoint cycles in $G$, and these $k$ shortest cycles are leaves of $\C{T}$.
\end{lemma}

\subsection{Proof of \Cref{thm:maxddisjointcoverage}}
We start the proof of ~\Cref{thm:maxddisjointcoverage}, which we restate here. 

\thmmaxddisjointcoverage*

Consider an instance $(G, w : E(G) \rightarrow \mathbb{Z}_{> 0}, k,d,\ell)$ of  \maxddisjointcoveragefull\ (\pmaxddisjointcoverage), where $G$ is a planar graph with $V(G) = [n]$. Recall our goal is to check if there is a set $\Ccal = \{C_1, C_2, \ldots, C_k\}$ of exactly $k$ many $d$-diverse shortest cycles in $G$ such that $w(\cup_{i\in k} E(C_i)) \geq \ell$. We will make the following assumptions regarding the instance:
\begin{remarkbox}
\begin{enumerate}
\item $G$ is a clean graph, i.e., each vertex and edge of $G$ is part of some shortest cycle in the graph (as otherwise such vertices/edges can be safely removed from $G$), 
\item $k \geq 1$, as otherwise, we can return the empty set as the solution, and 
\item $0 \leq d \leq g(G)/2$, where $g(G)$ is the (weighted) girth of $G$. The above  is valid due to~\Cref{claim:intersection-between-two-cycles}.
\item $0 \leq \ell\leq k \cdot g(G)$, as otherwise, the instance is trivially a no-instance of the problem.  
\end{enumerate}
\end{remarkbox} 

Throughout this section (and its subsections) we will be working with the instance $(G, w : E(G) \rightarrow \mathbb{Z}_{> 0}, k,d,\ell)$. We presented an overview of our algorithm in \Cref{subsec:intro}, and we will formalize those steps here.

\subparagraph{Laminar shortest cycles tree setup.} 
For notational simplicity, we add a cycle $C_r$ to the graph $G$ where $w(C_r) = g(G)$, increment $k$ and $\ell$ by $1$ and $g(G)$, respectively. 
Furthermore, we fix a planar embedding of $G$ such that all vertices and edges of $G$ (other than those in the newly added cycle $C_r$) lie inside the curve bounded by the cycle $C_r$. For convenience (and due to the equivalence of the original and the modified instances), we will denote the modified instance by $(G,w,k,d,\ell)$, and we will assume that if the instance is a yes-instance, then the considered  solution always contains the cycle $C_r$.

We use \Cref{prop:def:laminar-tree} (also see \Cref{def:laminar-tree}) to obtain the laminar shortest cycles tree $\Tcal$ of the weighted plane graph $G$. By \Cref{def:laminar-tree}, the root of $\Tcal$ is $C_r$, and each node of $\Tcal$ is a shortest cycle of $G$.
If the number of leaves in $\Tcal$ is at least $4k$, then by \Cref{lem:4k-leaves}, $G$ contains at least $k-1$ pairwise edge-disjoint shortest cycles that are leaves of $\Tcal$. These cycles together with $C_r$ yield a set of $k$ pairwise edge-disjoint cycles of $G$ that covers edges of total weight $k \cdot g(G)$, which is the maximum possible value for a set of $k$ shortest cycles. 
Thus, when the number of leaves in $\C{T}$ is at least $4k$, we can trivially resolve the instance. From now on we assume that the number of leaves in $\Tcal$ is less than $4k$. 

\subparagraph{Handling non-tree cycles.} Recall the laminar shortest cycles tree $\mathcal{T}$ partitions the shortest cycles of $G$ into two types (see~\Cref{sec:prelim-lsct}): tree cycles (shortest cycles of $G$ that are nodes of
$\mathcal{T}(G)$) and non-tree cycles (the remaining shortest cycles of
$G$). From \Cref{lem:non-tree-cycles-are-bounded} (\Cref{item-1-of-lem:non-tree-cycles-are-bounded}) we know that the number of non-tree cycles is bounded by $O(k^2)$; and this bound is tight. Due to this (and the time bound needed), we cannot keep track of the exact subset of non-tree cycles in the optimal solution we are looking for. 

Recall from \Cref{prop:lsct-use-property} (\ref{prop:cycle-type-lsct}), each non-tree cycle must consist of longitudinal paths, and from \Cref{lem:non-tree-cycles-are-bounded} (\ref{item-lem:loni-paths-bounded}) the number of longitudinal paths is $\Ocal(k)$. So instead of keeping track of the non-tree cycles, we will ``guess'' (or, more precisely, iterate over choices of the exact subset of) longitudinal paths participating in non-tree cycles. This leads us to the following definitions (and some results regarding them).

Let $U$ be the set of all splittable nodes of $\C{T}$ (see \Cref{def:laminar-tree} (\ref{split-def:laminar-tree})), and for $\alpha \in U$, let $\alpha$ be an $s_\alpha, t_\alpha$-splittable cycle (see \Cref{def:longitudinal}).\footnote{Here we are following the notations from \Cref{sec:prelim-lsct} that appears after \Cref{def:longitudinal}. Also, the notations we define here (and later sub-sections) will also be used in the subsequent sub-sections.} 
We let $\C{Q} = \cup_{\alpha\in U} \C{P}_{\sf s}(\alpha)$. Note $\C{Q}$ is the set of all longitudinal paths, and we can find it in polynomial time (see \Cref{lem:non-tree-cycles-are-bounded} \ref{item-lem:loni-paths-bounded}). For $\C{P} \subseteq \C{Q}$, $\C{P}$ is {\em $d$-diverse} if for every distinct $P,P' \in \C{P}$, $w(P \cap P') \leq d$. (Recall our notation that $w(P \cap P') = w(E(P)\cap E(P'))$.)
We now define the validity of coverage of paths in a $d$-diverse $\C{P}\subseteq \C{Q}$ using $q$ many non-tree cycles.

\begin{definition}
    \label{def:valid-set-of-paths}
    Consider $\C{P}\subseteq \C{Q}$ and an integer $q$, where $\C{P}$ is $d$-diverse and $0\leq q \leq k$. We say that $q$ is $\C{P}$-\emph{valid} if there is a set of $d$-diverse shortest non-tree
    cycles $\C{C}'=\{C_1,C_2, \ldots, C_q\}$, where for $i\in [q]$, $C_i = P_i\cup Q_i$, for $s_i-t_i$ longitudinal paths $P_i,Q_i$ for some pole $s_i,t_i$ with $V(P_i)\cap V(Q_i) = \{s_i,t_i\}$,\footnote{Here the poles for some of the cycles may be the same.} such that $\Pcal =\{P_i,Q_i \mid i \in [q]\}$. In the above we say that $\C{C}'$ {\em exactly covers} $\C{P}$.  
\end{definition}

Recall a solution for $(G,w,k,d,\ell)$ must have {\em exactly} $k$ shortest cycles (satisfying other conditions). 
So in our states, we need to keep track of different values of $q$ that are $\cal P$-valid, as for the same set $\C{P}\subseteq \C{Q}$, we may be able to construct different numbers of $d$-diverse shortest cycles exactly covering $\C{P}$. 
We next give an observation for testing $\cal P$-validity. 

\begin{lemma}\label{obs:test-h-validity}
Given $\C{P} \subseteq \C{Q}$ and an integer $0\leq q \leq k$, in time polynomial in $n$ we can do the following: i) test if $\C{P}$ is $d$-diverse, and ii) if $\C{P}$ is $d$-diverse, test if $q$ is $\C{P}$-valid. 
\end{lemma}
\begin{proof}
The proof of the first item follows from \Cref{lem:non-tree-cycles-are-bounded} \ref{item-lem:loni-paths-bounded}. We now focus on the proof of the second item, and assume that $\C{P}$ is $d$-diverse. 

Note that if $\Pcal = \emptyset$, then $q$ is $\Pcal$-valid if and only if $ q = 0$. We now assume that $\C{P}\neq \emptyset$. Recall that $U$ is the set of all splittable nodes in $\C{T}$. Note that for distinct $\alpha,\alpha' \in U$, $\C{P}_{\sf s}(\alpha) \cap \C{P}_{\sf s}(\alpha') =\emptyset$ (see \Cref{prop:lsct-use-property} (\ref{prop:s-node-lsct})). Let $U' = \{\alpha_1,\alpha_2,\ldots, \alpha_x\} \subseteq U$ be the set such that for each $\alpha\in U'$, $\C{P} \cap \C{P}_{\sf s}(\alpha) \neq \emptyset$. Furthermore, for each $i \in [x]$, let $\C{P}_i = \C{P}\cap \C{P}_{\sf s}(\alpha_i)$. Note the above gives us a partition of $\C{P}$, and it can be computed in polynomial time. For $i\in [x]$, let $s_i,t_i$ be the pole of the splittable cycle $\alpha_i$, and note that each path in $\C{P}_i$ is an $s_i-t_i$ shortest path in $G$ of length exactly $g(G)/2$. We next do some sanity checks where we can directly rule out that $q$ is $\C{P}$-valid.

If for any $i \in [x]$, $|\C{P}_{\sf s}(\alpha_i)| \leq 3$ (in which case it must be exactly $3$), then note all three cycles formed by a pair of paths from $\C{P}_{\sf s}(\alpha_i)$ must be tree cycles of $\C{T}(G)$ (one of these will be the parent, and the other two its children in $\C{T}(G)$). As \Cref{def:valid-set-of-paths} requires only non-tree cycles, if for any $i\in [x]$, $|\C{P}_{\sf s}(\alpha_i)| \leq 3$, then there cannot exist any $\C{P}$-valid integer, so we conclude that $q$ is not $\C{P}$-valid. If for any $i \in [x]$, $|\Pcal_{i}| =1$, then there is no $\C{P}$-valid integer, so we can declare that $q$ is not $\C{P}$-valid. Hereafter, we assume that for each $i \in [x]$, $|\Pcal_{i}| \geq 2$ and $|\C{P}_{\sf s}(\alpha_i)| \geq 4$.

For each $i \in [x]$, let $\Pcal_{\sf s}(\alpha_i)=\{P_{i,0}, P_{i,1}, P_{i,2}, \ldots, P_{i,b_i}\}$, where $\alpha_i = P_{i,0} \cup P_{i,b_i}$, and the paths in $\C{P}_{\sf s}(\alpha_i)$ appear in the stated order inside the curve of $\alpha_i$ in $G$'s embedding. Note we can construct these orderings in polynomial time, and paths in $\Pcal_{\sf s}(\alpha_i)$ are internally pairwise vertex-disjoint (see \Cref{prop:lsct-use-property} (\ref{prop:s-node-lsct})). We consider the following cases based on the value of $d$.

\begin{itemize}
\item Consider the case when $d <g(G)/2$. Notice that for any $i \in [x]$, any two distinct non-tree $d$-diverse cycles that are each union of two paths from $\C{P}_i$ must actually be edge-disjoint (also see \Cref{prop:lsct-use-property} (\ref{prop:s-node-lsct})). 
Thus, we can declare that $q$ is not $\C{P}$-valid if one of the following holds: i) $|\C{P}_i|$ is odd, for some $i\in [x]$, or ii) $\sum_{i=1}^x |\C{P}_i|/2 \neq q$. 
Otherwise, consider the following. 
For $i\in [x]$, let $\C{P}_i = \{P_{i,1}, P_{i,2},\ldots, P_{i,2 c_i}\}$, where $c_i \geq 1$. 
Consider the case when for some $i\in [x]$, $c_i =1$. 
In the above case, if $P_{i,1}, P_{i,2}$ are consecutive in the ordering of paths in $\C{P}_{\sf s}(\alpha_i)$, or $P_{i,1}\cup P_{i,2} = \alpha_i$ then $\C{P}_i\subseteq \C{P}$ cannot be covered with non-tree cycles as required by \Cref{def:valid-set-of-paths}, so we can declare that $q$ is not $\C{P}$-valid. 
Else, notice that the non-tree cycle $P_{i,1}\cup P_{i,2}$ covers the paths in $\C{P}_i$. Now consider any $i \in [x]$, where $c_i > 1$. 
Notice that the cycles $C_j = P_{i,j} \cup P_{i,c_i +j}$, for $j\in [c_i]$, exactly cover $\C{P}_i$ and they are pairwise edge-disjoint (and thus $d$-diverse). 
We will now show that unless we have declared that $q$ is not $\C{P}$-valid in the previous discussions, then $q$ is $\C{P}$-valid. Note the above discussed procedure can be executed in polynomial time, and we next argue the correctness. 

So we consider the case that has not been already covered before, and note that in that case, for each $i\in [x]$, we constructed a set of non-tree cycles, $\C{C}_i$, which has $d$-diverse shortest cycles that exactly covers $\C{P}_i$. Let $\C{C} = \cup_{i=1}^{x} \C{C}_i$, and note that $|\C{C}| = q$ by the construction we described. 
We now argue that $\C{C}$ is $d$-diverse, thus implying that $q$ is $\C{P}$-valid. 
To prove this it is enough to argue that for any distinct $i,i' \in [x]$, and cycles $C\in \C{C}_i$ and $C' \in \C{C}_{i'}$, $w(C\cap C') \leq d$. 
Recall that $\C{P}$ is $d$-diverse and the non-tree cycles $C$ and $C'$ must be associated with distinct splittable nodes (otherwise, they are already edge-disjoint), say $\alpha$ and $\alpha'$ of $\C{T}$, respectively, where $\alpha = X \cup Y, \alpha' = X' \cup Y'$, where $X,Y\in \C{P}_{\sf s}(\alpha)$ and $X',Y'\in \C{P}_{\sf s}(\alpha')$. Note that $C$ (resp. $C'$) can contain edges of at most one of the paths $X,Y$ (resp. $X', Y'$). The above discussions together with the assumption that $\C{P}$ is $d$-diverse and $C,C'$ are non-tree cycles, implies that $w(C\cap C') \leq d$. 

\item Now consider the case when $d=g(G)/2$. 
Recall that for each $i\in [x]$, we have $|\Pcal_i|\geq 2$ and $|\C{P}_{\sf s}(\alpha_i)|\geq 4$.
Similar to the previous case, notice that if for some $i\in [x]$, $|\C{P}_i| =2$, where either the two paths in $\C{P}_i$ are consecutive in the ordering of paths in $\C{P}_{\sf s}(\alpha_i)$, or their union exactly gives the tree cycle $\alpha_i$, then $q$ is not $\C{P}$-valid. 
Similarly, if for $i\in [x]$, $|\C{P}_i| =3$ and all the three paths in $\C{P}_i$ are consecutive in $\C{P}_{\sf s}(\alpha_i)$, then $q$ is not $\C{P}$-valid. 
Now we assume that we are not in the above case and construct numbers $y_i,z_i$, for $i\in [x]$ which denote the minimum number of non-tree cycles needed to exactly cover $\C{P}_i$, and the maximum number of non-tree cycles that exactly cover $\C{P}_i$. Consider any $i\in [x]$. 
Note that to cover the paths in $\C{P}_i$ we need $\ceil{\C{P}_i/2}$ non-tree cycles, and such a set of non-tree cycles exists; so we set $y_i = \ceil{\C{P}_i/2}$. 
Let $\C{H}_i$ be the set of (unordered) pairs of distinct paths from $\C{P}_i$ that are either consecutive in ordering of paths in $\C{P}_{\sf s}(\alpha_i)$, or it is the (unordered) pair that has the first and the last path of $\C{P}_{\sf s}(\alpha_i)$. 
Let ${\C{P}_i \choose 2}$ be the set of pairs of distinct paths from $\C{P}_i$, and $\C{Z}_i = {\C{P}_i \choose 2} \setminus \C{H}_i$. 
Notice that $\C{Z}_i$ is $d$-diverse, i.e., for any distinct unordered pairs $\{\iota,\iota'\}, \{\delta,\delta'\} \in \C{Z}_i$, we have $w((\iota\cup\iota')\cap (\delta\cup\delta')) \leq d$. Furthermore, $z_i = |\C{Z}_i|$, and they can be constructed in polynomial time. Similar to the previous case, note that for distinct $i,i' \in [x]$, and $C = \iota\cup \iota'$, for $\{\iota, \iota'\}\in \C{Z}_i$ and $C'=\delta\cup \delta'$, for $\{\delta, \delta'\}\in \C{Z}_{i'}$, we can obtain that $w(C \cap C') \leq d$. 
Thus (unless already stated otherwise by the previous arguments), $q$ is $\C{P}$-valid if and only if $\sum_{i\in [x]}y_i \leq q \leq \sum_{i\in [x]}z_i$, which can be tested in polynomial time. 
\end{itemize}
This concludes the proof.
\end{proof}

\subparagraph{Starting with the dynamic programming states.} We now define the states of our dynamic programming table, for a fixed $d$-diverse $\C{P}\subseteq \C{Q}$ set of longitudinal paths.  

\begin{definition}[States]
    \label{def:pi-p-alpha-h}
    Consider a set of longitudinal paths $\Pcal\subseteq \C{Q}$, where $\C{P}$ is $d$-diverse. 
 Recall that $w(\C{C} \cup \C{P}) = \sum_{e \in E_{\C{C} \cup \C{P}}}w(e)$, where $E_{\C{C} \cup \C{P}} = \bigcup_{X \in \C{C} \cup \C{P}} E(X)$.
    For an integer $h$, where $1\leq h \leq k$, and a node $\alpha$ of $\C{T}$, we define $\Pi_{\Pcal}(\alpha, h)$ to be the maximum over $w(\C{C} \cup \C{P})$, where $\C{C}$ is a set of $d$-diverse tree cycles of $G$ satisfying the following conditions: 
    \begin{enumerate}
        \item
        \label{item-1-of-def:pi-p-alpha-h}
        $|\C{C}| = h$ and $\alpha \in \C{C}$.

        \item
        \label{item-2-of-def:pi-p-alpha-h} For each $\beta\in \C{C}$, $\beta$ is a descendant of $\alpha$ in $\C{T}$.\footnote{For a rooted tree, we use the convention that each node is also a descendant of itself.}

        \item
        \label{item-3-of-def:pi-p-alpha-h}
        Every tree cycle in $\C{C}$ is {\em $d$-diverse} from every path in $\Pcal$, i.e., for each $C\in \C{C}$ and $P\in \C{P}$, we have $w(C\cap P) \leq d$. 
    \end{enumerate}
    In the above, any $\C{C}$ satisfying the above three  conditions is a {\em solution} for $\Pi_{\Pcal}(\alpha, h)$. Moreover, it is an {\em optimal solution} if we have $w(\C{C} \cup \C{P}) = \Pi_{\Pcal}(\alpha, h)$.    
\end{definition}

In \Cref{subsec:compute-pi} we prove the following lemma, and for the remainder of this section we show how we can compute the optimal solution for $(G,w,k,d,\ell)$ using this lemma.

\begin{lemma}[\Cref{subsec:compute-pi}]\label{lem:compute-pi}
Given $(G,w,k,d,\ell)$, $\C{T}$ and a $d$-diverse $\C{P}\subseteq \C{Q}$, in (total) time bounded by $2^{\Ocal(k)}\cdot n^{\Ocal(1)}$, for each $\alpha \in V(\C{T})$ and integer $h \in [k]$, we can compute the entry $\Pi_{\Pcal}(\alpha, h)$. 
\end{lemma} 

We are ready to prove the main result of this section.  

\begin{proof}[Proof of \Cref{thm:maxddisjointcoverage}] 
Consider an instance $(G,w,k,d,\ell)$ of \pmaxddisjointcoverage. As discussed earlier, we assume that $(G,w,k,d,\ell)$ satisfies the conditions stated in the beginning of the section, and we have computed $\C{T}$ (using \Cref{prop:def:laminar-tree}), rooted at $C_r$. If the number of leaves in $\C{T}$ is more than $4k$, we can resolve the instance in polynomial time as discussed earlier, so we assume that $\C{T}$ has at most $4k$ leaves.  

Now for each $d$-diverse $\C{P}\subseteq \C{Q}$, we compute $\Pi_{\Pcal}(\alpha, h)$, for each node $\alpha$ of $\C{T}$ and an integer $h \in [k]$. As the number of leaves in $\C{T}$ is bounded by $4k$, from \Cref{lem:compute-pi,obs:test-h-validity} and \Cref{lem:non-tree-cycles-are-bounded} (\Cref{item-lem:loni-paths-bounded}), it follows that we can execute the above stated step in (total) time bounded by $2^{\Ocal(k)} \cdot n^{\Ocal(1)}$. Furthermore, within the same running time we can compute the following number:\footnote{We use the convention that empty maximum is $-\infty$, which is a number smaller than any fixed integer.} 
 $$
        z= \max_{\substack{
          d{\scriptstyle \text{-diverse }} \Pcal \subseteq \C{Q}\\ h\in [k]:~(k-h) {\scriptstyle\text{ is }} \C{P}{\scriptstyle \text{-valid}}
        }}
        \Pi_{\Pcal}(C_r, h)
    $$
    Finally, we return that $(G,w,k,d,\ell)$ is a yes-instance if and only if $ z\geq \ell$. 
    
  The desired runtime bound follows from above discussions. Thus, to conclude the proof it is enough to argue that the weight of an optimal solution for $(G,w,k,d,\ell)$ is exactly $z$. (In the above if an optimal solution does not exist, by the weight of its optimal solution we mean $-\infty$.) To this end, consider an optimal solution $\Ccal$ for $(G,w,k,d,\ell)$, where $C_r\in \C{C}$. 
  
We will first show that $z\geq w(\C{C})$. We assume that right hand side is not $-\infty$, else the inequality is trivially satisfied. Let $\Ccal_T$ and $\Ccal_N$ denote the sets of tree cycles and
    non-tree cycles in $\Ccal$, respectively. Let $h^* = |\Ccal_T|$, and let
    $\Pcal^*$ be the set of longitudinal paths covered by the $k-h^*$ cycles in
    $\Ccal_N$. Then the edges covered by $\Ccal_N$ are exactly the edges covered by
    $\Pcal^*$, and by \Cref{def:valid-set-of-paths}, $(k-h^*)$ is $\Pcal^*$-valid. Note that $z \geq \Pi_{\Pcal^*}(C_r, h^*)$, so to conclude that $z\geq w(\C{C})$, it is enough to argue that $\C{C}_T$ satisfies all the conditions stated in \Cref{def:pi-p-alpha-h}, for $\C{P}^*$, $h^*$, and $C_r$. By the construction of $h^*$ and as $C_r\in \C{C}$, \Cref{def:pi-p-alpha-h} (\ref{item-1-of-def:pi-p-alpha-h}) is satisfied by $\C{C}_T$. As $C_r$ is the root of $\C{T}$, trivially \Cref{def:pi-p-alpha-h} (\ref{item-2-of-def:pi-p-alpha-h}) is satisfied by $\C{C}_T$. Now we argue that for each $C\in \C{C}_T$ and $P\in \C{P}^*$, we have $w(C\cap P) \leq d$. Note that $P$ is a longitudinal path, say an $s-t$ longitudinal path, for pole $s,t$. As $P\in \C{P}^*$, there exists $C'\in \C{C}_N$, such that $C' = P \cup P'$, where both $P,P'$ are $s-t$ longitudinal paths with $V(P) \cap V(P') = \{s,t\}$. Moreover, as $C'$ is a non-tree cycle, \Cref{prop:lsct-use-property} (\ref{prop:cycle-type-lsct}) implies that there is a unique splittable node $\gamma$ in $\C{T}$, such that $P,P'$ are non-consecutive paths of $\C{P}_{\sf s}(\gamma)$. The above together with the fact that $\C{C}$ is $d$-diverse and $C$ is a tree cycle (and \Cref{claim:intersection-between-two-cycles}), implies that the edges (if they exist) in $E(C)\cap E(C')$ must lie in at most one of the paths $P, P'$. Thus we can conclude that $w(C\cap P)\leq d$. This concludes the proof that $z\geq w(\C{C})$. 

    We now show that $z\leq w(\C{C})$ (and again assume that $z\neq -\infty$). 
    Consider $\C{P}'\subseteq \C{Q}$, an integer $h'\in [k]$ where $k-h'$ is $\C{P}'$-valid, such that $z = \Pi_{\C{P}'}(C_r,h')$. 
    Furthermore, let $\Ccal_1$ be the set of tree cycles satisfying the conditions stated in \Cref{def:pi-p-alpha-h} for $\C{P}'$, $h'$, and $C_r$, such that $\Pi_{\C{P}'}(C_r,h') = w(\Ccal_1 \cup \C{P}')$. As $k-h'$ is $\C{P}'$-valid, there exists a set of $d$-diverse non-tree cycles $\C{C}_2=\{C_1,C_2, \ldots, C_{k-h'}\}$, such that $\C{C}_2$ exactly covers $\C{P}'$ (see \Cref{def:valid-set-of-paths}). 
    Let $\C{C}' = \C{C}_1 \cup \C{C}_2$, and note that $w(\C{C}') = z \leq w(\C{C})$, where $\C{C}$ is an optimal solution for $(G,w,k,d,\ell)$. 
    Thus, it is enough to argue that $\C{C}'$ is a solution for $(G,w,k,d,\ell)$. 
    Note by construction, $|\C{C}'| = k$, and to establish the statement above, it is enough to argue that for  $C_1\in \C{C}_1$ and $C_2 \in \C{C}_2$, we have $w(C_1\cap C_2) \leq d$. 
    As $C_2 = P \cup P'$ is a non-tree cycle (see \Cref{prop:lsct-use-property} (\ref{prop:cycle-type-lsct})), it must be some $s-t$ splittable cycle where $P,P' \in \C{P}'$ are $s-t$ longitudinal paths. 
    Similar to the arguments used while proving the other inequality, note we can obtain that $E(C_1)$ can intersect with edges in at most one of the paths $P,P'$. 
    This implies that $w(C_1 \cap C_2) \leq d$ (see \Cref{def:pi-p-alpha-h} (\ref{item-3-of-def:pi-p-alpha-h})). 
    
    This concludes the proof of \Cref{thm:maxddisjointcoverage} (subject to the proof of \Cref{lem:compute-pi}). 
\end{proof}

    \subsection{Proof of \Cref{lem:compute-pi}: Computing Table Entries}\label{subsec:compute-pi}

In this section we prove~\Cref{lem:compute-pi}. The proof proceeds in several stages. One of the crucial steps is the technical \Cref{lem:compute-pi-ext}, whose proof occupies \Cref{sec:conputation-of-pi-alpha-h-using-theta}.

Consider an instance $(G,w,k,d,\ell)$ of \pmaxddisjointcoverage, where $G$ is a clean plane graph with $V(G)=[n]$ and the outer face is the cycle $C_r$, where vertices of $C_r$ have no edges to any vertex in $V(G)\setminus V(C_r)$, and for any cycle $C$ in $G$, we have $C \leq C_r$, $w: E(G) \rightarrow \mathbb{Z}_{>0}$, $k\geq 1$, $0 \leq d \leq g(G)/2$. 
Also, we have the LSCT $\C{T}$ of $G$ is rooted at $C_r$, which has at most $4k$ leaves. 
We consider $\C{P} \subseteq \C{Q}$, where $\C{P}$ is $d$-diverse, an integer $h\in  [k]$, and a node $\alpha$ of $\C{T}$ (see~\Cref{def:pi-p-alpha-h}), and describe a (recursive) formula to compute the entry $\Pi_{\Pcal}(\alpha, h)$. We start with a simple sanity check.  

\smallskip
\noindent{\bf Sanity Check.} Note that if $\{\alpha\}$ is not $d$-diverse from $\C{P}$ (see \Cref{def:pi-p-alpha-h} (\ref{item-3-of-def:pi-p-alpha-h})), then $\Pi_{\C{P}}(\alpha,h)= -\infty$. 
(Recall our convention that empty maximum is $-\infty$). 
Note we can test the above condition in polynomial time, and accordingly set the value if needed. 

Hereafter, we assume that $\{\alpha\}$ is $d$-diverse from $\C{P}$, and next describe our base case. 

\smallskip
\noindent{\bf Base Case.} When $\alpha$ is a leaf of $\C{T}$, or $h=1$. In each of the above cases $\alpha$ must be in any solution for $\Pi_{\C{P}}(\alpha,h)$. From \Cref{{def:pi-p-alpha-h}} (\ref{item-1-of-def:pi-p-alpha-h}), if $h=1$, we can correctly set $\Pi_{\C{P}}(\alpha,h) = w(\{\alpha\} \cup \C{P})$. Moreover, if $\alpha$ is a leaf of $\C{T}$ and $h>1$, then we can correctly set $\Pi_{\C{P}}(\alpha,h)$ to $-\infty$. 

\medskip\noindent{\bf Towards the recursive formula.} We will now consider the case when $\alpha$ is not a leaf of $\C{T}$ and $h>1$ (and $\{\alpha\}$ is $d$-diverse from $\Pcal$). Furthermore, we will assume that we have correctly computed the entries $\Pi_{\Pcal}(\beta, h')$, for each descendant $\beta$ of $\alpha$ in $\C{T}$, where $\beta\neq \alpha$ and $h' \in [k]$. To compute $\Pi_{\C{P}}(\alpha,h)$, we would have ideally liked to query entries for descendant of $\alpha$ that belongs to a (fixed unknown) optimal solution, which we formalize in the next definition (also see \Cref{fig:pd-and-pdr}). However, as the number of nodes in $\C{T}$ may not be bounded by a function of $k$, to achieve the running time we desire, the above type of queries will not be possible for us, due to which we will need a slightly more non-trivial notion of ``relevant'' descendants.  

\begin{figure}
\centering
\begin{tikzpicture}[
scale=1,
every node/.style={circle,draw,minimum size=15pt,inner sep=1pt},
font=\scriptsize
]

\node(a) at (0,0)[draw=red,line width=1pt,line width=1pt]{\(\alpha\)};

\node(l0) at(-1.5,-0.5){};
\node(l1) at (-3,-1){};
\node(m1) at (0,-1){};
\node(r1) at (3,-1){$\gamma_5$};
\draw[blue, line width=1pt] (r1) circle (10pt);

\node(l2) at (-4,-2)[draw=red,line width=1pt]{$\beta_1$};
\node() at (-3.25,-2.05)[draw=none,line width=1pt]{$= \gamma_1$};
\draw[blue, line width=1pt] (l2) circle (10pt);
\node(l3) at (-5,-3){};
\node(l4) at (-6,-4){};

\node(lb1) at (-4,-3){};
\node(lb2) at (-4,-4){};

\node(m2) at (0,-2){$\delta_1$};

\node(new1) at (-1.5,-2.0){};

\node(m3l) at (-1.5,-3){$\gamma_2$};
\draw[blue, line width=1pt] (m3l) circle (10pt);
\node(m4l) at (-1.5,-4)[draw=red,line width=1pt]{$\beta_2$};
\node(m5l) at (-1.5,-5)[draw=red,line width=1pt]{};
\node(m6l) at (-1.5,-6){};

\node(m3c) at (0,-3){$\gamma_3$};
\node(m3c2) at (0,-4)[draw=red,line width=1pt]{$\beta_3$};
\draw[blue, line width=1pt] (m3c) circle (10pt);

\node(m3r) at (1.5,-3){};
\node(m4r) at (1.5,-4)[draw=red,line width=1pt]{$\beta_4$};
\node() at (2.25,-4.05)[draw=none,line width=1pt]{$= \gamma_4$};
\draw[blue, line width=1pt] (m4r) circle (10pt);
\node(m5r1) at (1,-5){};
\node(m5r2) at (2,-5){};

\node(r2) at (4,-2){$\gamma_5^{\sf s}$};
\node(r3) at (5,-3){};
\node(r4) at (5.5,-4)[draw=red,line width=1pt]{$\beta_5$};
\node(r5) at (6,-5){$\gamma_5^{\sf t}$};
\node(r6) at (6.6,-6){};


\draw (a)--(l0);
\draw (l1)--(l0);
\draw (a)--(m1);
\draw (a)--(r1);

\draw (l1)--(l2);
\draw (l2)--(l3);
\draw (l3)--(l4);

\draw (l2)--(lb1);
\draw (lb1)--(lb2);

\draw (m1)--(m2);
\draw (m1)--(new1);

\draw (m2)--(m3l);
\draw (m3l)--(m4l);
\draw (m4l)--(m5l);
\draw (m6l)--(m5l);

\draw (m2)--(m3c);
\draw (m3c2)--(m3c);

\draw (m2)--(m3r);
\draw (m3r)--(m4r);
\draw (m4r)--(m5r1);
\draw (m4r)--(m5r2);

\draw (r1)--(r2);
\draw (r2)--(r3);
\draw (r3)--(r4);
\draw (r4)--(r5);
\draw (r5)--(r6);

\end{tikzpicture}

\caption{The picture illustrates, for a node $\alpha$ of $\C{T}$, the subtree $\C{T}_\alpha$, and a solution $\C{C}$ for $\Pi_{\C{P}}[\alpha,7]$ using red border nodes. The nodes labelled $\beta_1,\beta_2, \beta_3, \beta_4$, and $\beta_5$ are solution descendants of $\alpha$ in $\C{C}$, and the nodes $\gamma_1,\gamma_2, \gamma_3, \gamma_4$, and $\gamma_5$ are their \pdrs, respectively. We also illustrate a maximal degree $2$ path below $\gamma_5$ from $\gamma^{\sf s}_5$ to $\gamma^{\sf t}_5$ (relevant for \Cref{sec:conputation-of-pi-alpha-h-using-theta}).}
\label{fig:pd-and-pdr}
\end{figure}
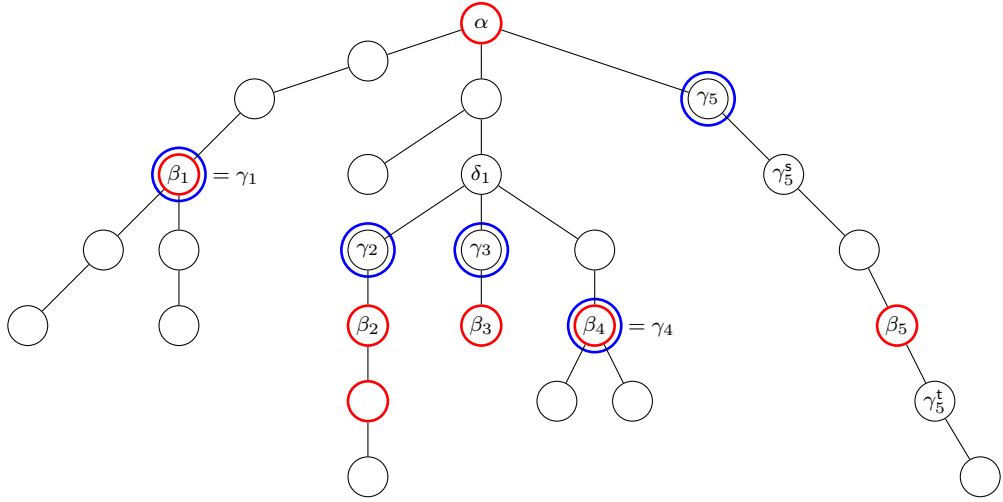

\begin{definition}
\label{def:productive-descendant}
Let $\Ccal$ be a solution for $\Pi_{\Pcal}(\alpha, h)$, where $\alpha$ is a non-leaf node of $\Tcal$ and $h > 1$. A node $\beta \in \Ccal \setminus \{\alpha\}$ is a \emph{solution descendant}
of $\alpha$ in $\Ccal$ if there is no $\beta' \in \Ccal \setminus \{\alpha, \beta\}$ that lies on the unique path between $\alpha$ and $\beta$ in $\Tcal$.
\end{definition}

For an illustration of the above definition, see \Cref{fig:pd-and-pdr}. The following observation (which will be useful later) relates $\Pi_{\C{P}}(\alpha,h)$ to the entries of the (unknown) solution descendants of $\alpha$ for some fixed (unknown) optimal solution, and it easily follows from \Cref{def:pi-p-alpha-h}. 
 
 {Recall that for a node $\gamma$ of $\C{T}$, $\C{T}_\gamma(G)$ is the subtree of $\C{T}$ rooted at $\gamma$.}
 \begin{observation}
    \label{obs:partitioning-of-a-solution}
    Let $\Ccal$ be a solution for $\Pi_{\Pcal}(\alpha, h)$, where $\alpha$ is a non-leaf node of $\Tcal$ and $h > 1$. Furthermore, let $\{\beta_1, \beta_2, \ldots, \beta_{b}\}$ be the set of \pds{} of $\alpha$ in $\Ccal$. For each $i \in [b]$, let $\Ccal_i = \Ccal \cap V(\C{T}_{\beta_i})$, and $h_i = |\Ccal_i|$. Then the following holds. 
    \begin{enumerate}
        \item 
        \label{item-1-of-obs:partitioning-of-a-solution}
        $\{\{\alpha\}, \Ccal_1, \Ccal_2, \ldots, \Ccal_{b}\}$ is a partition of $\Ccal$. 
        
        \item 
        \label{item-2-of-obs:partitioning-of-a-solution}
        For each $i \in [b]$, $h_i \geq 1$, and $\sum_{i = 1}^{b} h_i = h - 1$.
        
        \item 
        \label{item-3-of-obs:partitioning-of-a-solution} 
 For any $i\in [b]$, replacing $\Ccal_i$ by any optimal solution $\Ccal_i^*$ for $\Pi_{\Pcal}(\beta_i,h_i)$ yields a solution
            $\Ccal^*=(\Ccal\setminus \Ccal_i)\cup \Ccal_i^*$
        for $\Pi_{\Pcal}(\alpha,h)$, where $w(\Ccal^*\cup \C{P}) \geq w(\Ccal \cup \C{P})$.
        
        \item\label{item-1-of-obs:partitioning-of-a-solution-opt-sol} For each $i\in [b]$, $\C{C}_i$ is a solution for $\Pi_{\Pcal}(\beta_i, h_i)$. Moreover, if $\C{C}$ is an optimal solution for $\Pi_{\Pcal}(\alpha, h)$, then $\C{C}_i$ is an optimal solution for $\Pi_{\Pcal}(\beta_i, h_i)$.
    \end{enumerate}
\end{observation}

\begin{proof}
    The first statement follows  from \Cref{def:productive-descendant}, since every cycle in $\Ccal \setminus \{\alpha\}$ belongs to the subtree of exactly one solution descendant of $\alpha$. Moreover, the above implies the second statement. The fourth statement is implied by the third one, so we focus on the proof of third item. 
     
As $\beta_i \in \Ccal_i \cap \Ccal_i^{*}$, for any cycle $\delta \in \Ccal^*\setminus \Ccal^*_i$ and $\delta'\in \Ccal^*_i$, $E(\delta\cap \delta') \subseteq E(\beta_i \cap \delta')$ (see \Cref{lem:intersection-with-ancestor-is-more}), and thus we can conclude that $\C{C}^*$ is a solution for $\Pi_{\Pcal}(\alpha, h)$. Moreover, $w(\Ccal^*\cup \C{P}) \geq w(\Ccal \cup \C{P})$ follows from $w(\Ccal^*_i\cup \C{P}) \geq w(\Ccal_i\cup \C{P})$.
\end{proof} 

We will introduce another definition that helps us compute $\Pi_{\Pcal}(\alpha, h)$ using at most $2^{O(k)}$ queries, and to this end we construct some sets of marked nodes from $\C{T}$ as follows. Recall the number of leaves in $\C{T}$ is bounded by $4k$, and thus the number of degree three nodes in $\C{T}$ is at most $4k$. 
For $\gamma\in V(\C{T})$, let $\C{M}_\gamma$ be the nodes from $\C{T}_\gamma(G)$ containing: i) each leaf, ii) each node of degree at least three, iii) for each maximal path in $\C{T}_\gamma(G)$ where all nodes are of degree exactly $2$, the endpoints of these paths that are closest to $\gamma$, and iv)  all the children of $\gamma$. We have the following simple observation, which immediately follows from the above definition. 

\begin{observation}\label{obs:size-of-M}
For any node $\gamma$ of $\C{T}$, we have $|\C{M}_\gamma|\leq 16k$.
\end{observation}

Intuitively, to compute $\Pi_{\Pcal}(\alpha, h)$, instead of querying over the solution descendants of $\alpha$ in a (fixed unknown) optimal solution, we will manage with the ``closest'' descendants from the set $\C{M}_\alpha$. We formalize this in the next definition (see \Cref{fig:pd-and-pdr}). 

\begin{definition}
\label{def:productive-descendant-representative}
Let $\Ccal$ be a solution for $\Pi_{\Pcal}(\alpha, h)$, where $\alpha$ is a non-leaf node of $\Tcal$ and $h > 1$. Consider a \pd{} $\beta$ of $\alpha$ in $\C{C}$. The {\em $(\alpha, \C{C})$-\pdr{}} (or simply, {\em \pdr}, when the context is clear), $\gamma$ for $\beta$ is the vertex from $\C{M}_\alpha$ that is closest to $\beta$ in the unique path between $\beta$ and $\alpha$ in $\C{T}$. 
\end{definition}

\begin{figure}[t]

\centering

\begin{tikzpicture}[scale=1]

\draw[thick] (0,0) circle (3);
\node at (0,3.4) {$\delta$};

\draw[thick]
(240:3)
.. controls (-1,-1) and (1,0) ..
coordinate[pos=0.3] (p1)
coordinate[pos=0.75] (p2)
(330:3);

\draw[thick]
(175:3)
.. controls (-0.8,0.25) and (-0.8,0) ..
coordinate[pos=0.4] (p3)
(p1);

\draw[thick]
(p2)
.. controls (0.8,0.8) and (1.5,0.8) ..
coordinate[pos=0.4] (p4)
(30:3);

\draw[thick]
(p3)
.. controls (0,0.8) and (0.4,0.8) ..
coordinate[pos=0.5] (p5)
(p4);

\draw[thick]
(p5)
.. controls  (-0.4,1) and (-1.5,2) ..
(150:3);

\draw[thick]
(p5)
.. controls (0,1) and (0.2,2) ..
(70:3);


\fill[black] (p1) circle (1.5pt);
\fill[black] (p2) circle (1.5pt);

\fill[black] (p3) circle (1.5pt);
\fill[black] (p4) circle (1.5pt);
\fill[black] (p5) circle (1.5pt);

\fill[black] (240:3) circle (1.5pt);
\fill[black] (330:3) circle (1.5pt);
\fill[black] (175:3) circle (1.5pt);
\fill[black] (30:3) circle (1.5pt);
\fill[black] (150:3) circle (1.5pt);
\fill[black] (70:3) circle (1.5pt);


\node at (0.5,-1.8) {$\gamma$};      

\node at (-1.8,-1.1) {$\gamma_1$};      
\node at (0,-0.2) {$\gamma_2$};            
\node at ( 2.0,0.3) {$\gamma_3$};      
\node at (-2.0,0.6) {$\gamma_4$};      
\node at (1.5,1.8) {$\gamma_5$};           

\end{tikzpicture}

\caption{Illustration of the {\em relevant-subset.}
Here we have a {\em $\delta$-lineage-free} set $R = \{\gamma,\gamma_1,\gamma_2,\gamma_3,\gamma_4,\gamma_5\}$.
For the cycle $\gamma$ we have $\Gamma_{\delta,R,\gamma}= \{\delta,\gamma_1,\gamma_2,\gamma_3\}$
}
\label{fig:relevant_subset}
\end{figure}
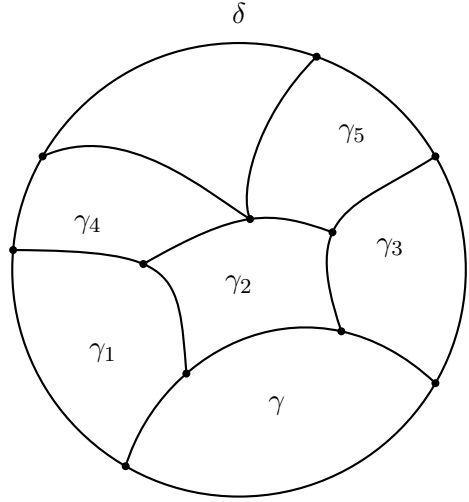

For $\delta\in V(\C{T})$, a set $R \subseteq \C{M}_\delta$ ($\subseteq V(\C{T}_\delta)$) is {\em $\delta$-lineage-free} (or, {\em lineage-free}, if the context is clear) if for distinct $\gamma, \gamma' \in R$, neither $\gamma$ is a descendant of $\gamma'$, nor $\gamma'$ is a descendant of $\gamma$, in $\C{T}_\delta$. 
Consider a $\delta$-lineage-free set $R$ and $\gamma\in R$.
We let $\Gamma_{\delta, R,\gamma}$ be the set of cycles in $R \cup \{\delta\}$ (other than $\gamma$), that share an edge with $\gamma$, i.e., $\Gamma_{\delta, R,\gamma} = \{\gamma' \mid \gamma' \in (R\cup \{\delta\}) \setminus \{\gamma\}, E(\gamma')\cap E(\gamma) \neq \emptyset\}$. We say that $\Gamma_{\delta, R,\gamma}$ is the {\em $(\delta,R,\gamma)$-relevant set} (see \Cref{fig:relevant_subset}). Next we give some simple observations, the first of which immediately follows from \Cref{def:productive-descendant,def:productive-descendant-representative}.

\begin{observation}
    \label{obs:relation-between-pd-and-pdrs}
    Consider a solution $\Ccal$ for $\Pi_{\Pcal}(\alpha, h)$, where $\alpha$ is a non-leaf node of $\Tcal$ and $h > 1$. Let $\{\beta_i \mid i\in [b]\}$ be the set of \pds{} of $\alpha$ in $\Ccal$, where for $i\in [b]$, $\gamma_i$ is the \pdr{} of $\beta_i$.
    \begin{enumerate}
        \item 
        \label{item-2-of-obs:relation-between-pd-and-pdrs}
        For distinct $i,j \in [b]$, $\beta_i \neq \beta_j$ and $\gamma_i \neq \gamma_j$.
     
     \item The set $\{\gamma_i\mid i\in [b]\}$ is lineage-free. 
    \end{enumerate}
\end{observation} 



The next observation shows that no edge can belong to three cycles in a $\delta$-lineage-free set. 
Consequently, solution descendant descendant $\beta$ of some cycle $\gamma \in R$ cannot share the same edge with two distinct cycles in the $(\delta,R,\gamma)$-relevant set. We will crucially rely on this property in our dynamic programming. 

\begin{observation}
    \label{obs:split-relation-between-pd-and-pdrs}
    Consider $\delta \in V(\C{T})$, a $\delta$-lineage-free set $R$, and $\gamma^* \in R$. 
    Then for any descendant $\beta^*$ of $\gamma^*$ in $\C{T}_\delta$ and distinct $\gamma,\gamma' \in \Gamma_{\delta,R,\gamma^*}$, we have $E(\beta^*) \cap E(\gamma) \cap E(\gamma') = \emptyset$.
\end{observation}

\begin{proof}     
Towards a contradiction, suppose there is a descendant $\beta^*$ of $\gamma^*$ in $\C{T}_\delta$ and distinct $\gamma,\gamma' \in \Gamma_{\delta,R,\gamma^*}$ such that $ E(\beta^*) \cap E(\gamma) \cap E(\gamma') \neq \emptyset$. As $\gamma\neq \gamma'$, at least one of $\gamma$ or $\gamma'$ is different from $\delta$, say $\gamma \neq \delta$. As $E(\beta^* \cap \gamma) \neq \emptyset$, from~\Cref{claim:intersection-between-two-cycles} we can obtain that  $\beta^*\cap \gamma$ is a path in $G$. Note that $\beta^*$ and $\gamma$ are not in ancestor-descendant relationship, and thus, from \Cref{prop:lsct-use-property} (\ref{prop:desc-lsct}), neither $\beta^* \leq \gamma$, nor $\gamma \leq \beta^*$ in the (given) planar embedding of $G$. Now, by the planarity of $G$, if a cycle $C$ contains an edge in $E(\beta^*) \cap E(\gamma)$, then exactly one of the following holds: i) $C \leq \beta^*$, ii) $C\leq \gamma$, iii) $\beta^* \leq C$, or iv) $\gamma \leq C$. Notice that $\gamma'$ satisfies none of the above conditions, as argued next. Indeed, neither $\gamma' \leq \beta^*$ nor $\gamma' \leq \gamma$ (recall $\gamma \neq \delta$). Also exactly one of $\beta^* \leq \gamma'$ or   $\gamma \leq \gamma'$ cannot hold, whether or not $\gamma' = \delta$ (see \Cref{prop:lsct-use-property}~(\ref{prop:desc-lsct})). Therefore, $\gamma'$ cannot contain an edge from $E(\beta^*) \cap E(\gamma)$, which contradicts our assumption. 
\end{proof}

The next observation exploits the planarity of $G$ to bound the sum of the sizes of the relevant sets of all representative descendants, which will be useful in obtaining the desired running-time.

\begin{observation}\label{obs:planar-gamma}
    For $\delta \in V(\C{T})$ and a $\delta$-lineage-free set $R = \{\gamma_i\mid i\in [b]\}$, $\sum_{i\in [b]} |\Gamma_{\delta, R,\gamma_i}| \leq 7 b$.
\end{observation}
We remark that the constant  $7 b$ is likely not tight, but it is sufficient for our purposes.
\begin{proof}
    Let $G_R$ be the plane graph obtained from $G$ by deleting each vertex and edge that are not present in any of the cycles in $R$. 
    Furthermore, let $\wtilde{G}_R$ be the graph with $V(\wtilde{G}_R) = R$ and for distinct $\gamma,\gamma' \in R$, $\{\gamma,\gamma'\}$ is an edge in $\wtilde{G}_R$ if $E(\gamma) \cap E(\gamma')\neq \emptyset$. 
    Note the following:
    \begin{enumerate}
        \item For distinct $\gamma,\gamma' \in R$, $\{\gamma,\gamma'\} \in E(\wtilde{G}_R)$ if and only if $\gamma \in \Gamma_{\delta,R,\gamma'}\setminus \{\delta\}$ and $\gamma' \in \Gamma_{\delta,R,\gamma}\setminus \{\delta\}$. In other words, $N_{\wtilde{G}_R}(\gamma) = \Gamma_{\delta,R,\gamma}\setminus \{\delta\}$ and $N_{\wtilde{G}_R}(\gamma') = \Gamma_{\delta,R,\gamma'}\setminus \{\delta\}$. 
        \item The graph $\wtilde{G}_R$ is a simple graph and it is a subgraph of the dual of $G_R$. Thus, $\wtilde{G}_R$ is a planar graph, hence $|E(\wtilde{G}_R)| < 3b$. (The above follows from the fact that a dual of a planar graph is a planar multi-graph, see for instance, Section 4.6 of~\cite{Diestel}).
    \end{enumerate}
    Using the above, together with the possibility that $\delta\in \Gamma_{\delta,R,\gamma_i}$, for each $i\in [b]$, we can conclude that $\sum_{i\in [b]} |\Gamma_{\delta, R,\gamma_i}| = 2 \cdot |E(\wtilde{G}_R)| + b \leq 7b$. 
\end{proof}

 



As stated previously, to achieve the desired running time, we will not be able to maintain the exact \pds{} for an (unknown) optimal solution $\Ccal$ for $\Pi_{\C{P}}(\alpha,h)$, and instead will be working with their \pdrs. 
This makes the (recursive) computation of $\Pi_{\C{P}}(\alpha,h)$ way more tedious, as for instance, if $\gamma$ is a representative for $\beta\in \C{C}$, we do not know how $\beta$ intersects with other cycles from $\C{C}$, when $\gamma \neq \beta$. 
We need a way by which we can still account for contributions of edges in the selected cycles to $\Pi_{\C{P}}(\alpha,h)$, and we will see how we can still do it using the \pdrs, by carefully going over the following scenarios. 
Let $\{\beta_i \mid i\in [b]\}$ be the set of \pds{} of $\alpha$ in $\Ccal$, where for $i\in [b]$, $\gamma_i$ is the \pdr\ of $\beta_i$ (see \Cref{obs:relation-between-pd-and-pdrs} (\ref{item-2-of-obs:relation-between-pd-and-pdrs})). 
Furthermore, let $R=\{\gamma_i\mid i\in [b]\}$ and for $i\in [b]$, let $\Gamma_i = \Gamma_{\alpha,R, \gamma_i}$. For each $i\in [b]$, we will maintain a state $\theta_i$ which (intuitively discussed below) will help capture the unknown intersections of $\beta_i$ with cycles from $\Gamma_i$, which will be enough for our purpose due to \Cref{lem:intersection-with-ancestor-is-more}.

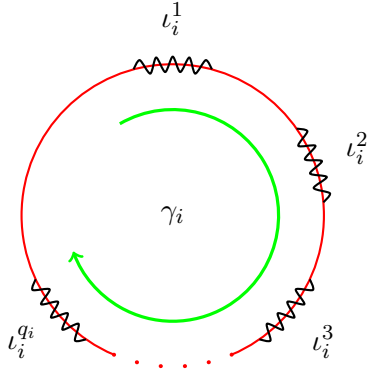
\begin{figure}[t]
\centering
\begin{tikzpicture}[scale=2]

\def\r{1}
\def\delta{30}

\def\aone{90}
\def\atwo{20}
\def\athree{-40}
\def\afour{-140}

\def\gap{12}

\draw[thick, red] (0,0) circle (\r);
\node at (0,0) {$\gamma_i$};

\draw[line width=2pt, white]
(\athree-\delta/2-\gap:\r)
arc[start angle=\athree-\delta/2-\gap, end angle=\afour+\delta/2+\gap, radius=\r];

\draw[
  red,
  line width=1.5pt,
  line cap=round,
  dash pattern=on 0pt off 9pt
]
(\athree-\delta/2-\gap:\r)
arc[start angle=\athree-\delta/2-\gap, end angle=\afour+\delta/2+\gap, radius=\r];

\draw[thick] plot[smooth] coordinates {
    (\aone-\delta/2:\r)
    (\aone-2*\delta/5:\r+0.05)
    (\aone-3*\delta/10:\r-0.05)
    (\aone-\delta/5:\r+0.05)
    (\aone-\delta/10:\r-0.05)
    (\aone:\r+0.05)
    (\aone+\delta/10:\r-0.05)
    (\aone+\delta/5:\r+0.05)
    (\aone+3*\delta/10:\r-0.05)
    (\aone+2*\delta/5:\r+0.05)
    (\aone+\delta/2:\r)
};
\node at (\aone:1.3*\r) {$\iota_i^1$};

\draw[thick] plot[smooth] coordinates {
    (\atwo-\delta/2:\r)
    (\atwo-2*\delta/5:\r+0.05)
    (\atwo-3*\delta/10:\r-0.05)
    (\atwo-\delta/5:\r+0.05)
    (\atwo-\delta/10:\r-0.05)
    (\atwo:\r+0.05)
    (\atwo+\delta/10:\r-0.05)
    (\atwo+\delta/5:\r+0.05)
    (\atwo+3*\delta/10:\r-0.05)
    (\atwo+2*\delta/5:\r+0.05)
    (\atwo+\delta/2:\r)
};
\node at (\atwo:1.3*\r) {$\iota_i^2$};

\draw[thick] plot[smooth] coordinates {
    (\athree-\delta/2:\r)
    (\athree-2*\delta/5:\r+0.05)
    (\athree-3*\delta/10:\r-0.05)
    (\athree-\delta/5:\r+0.05)
    (\athree-\delta/10:\r-0.05)
    (\athree:\r+0.05)
    (\athree+\delta/10:\r-0.05)
    (\athree+\delta/5:\r+0.05)
    (\athree+3*\delta/10:\r-0.05)
    (\athree+2*\delta/5:\r+0.05)
    (\athree+\delta/2:\r)
};
\node at (\athree:1.3*\r) {$\iota_i^3$};

\draw[thick] plot[smooth] coordinates {
    (\afour-\delta/2:\r)
    (\afour-2*\delta/5:\r+0.05)
    (\afour-3*\delta/10:\r-0.05)
    (\afour-\delta/5:\r+0.05)
    (\afour-\delta/10:\r-0.05)
    (\afour:\r+0.05)
    (\afour+\delta/10:\r-0.05)
    (\afour+\delta/5:\r+0.05)
    (\afour+3*\delta/10:\r-0.05)
    (\afour+2*\delta/5:\r+0.05)
    (\afour+\delta/2:\r)
};
\node at (\afour:1.3*\r) {$\iota_i^{q_i}$};

\draw[->, green, very thick]
(120:0.7*\r)
arc[start angle=120, end angle=-160, radius=0.7*\r];

\end{tikzpicture}
\caption{\centering Illustration of $\gamma_i$ and its intersections with cycles in $\Gamma_i$, for $i\in [b]$.}
\label{fig:cycle-intersections}
\end{figure}

\begin{enumerate} 
\item If $\gamma_i=\beta_i$, we will denote this by $\theta_i = \F$, where, intuitively, $\F$ will stand for ``equal''. 
\item If $\gamma_i\neq \beta_i$ and the cycle $\beta_i$ does not share an edge with any cycle in $\Gamma_i$ (or its descendants), we will maintain this state by $\N$, where $\N$ will stand for no intersection. 
\item Now we need to worry about the case when $\gamma_i\neq \beta_i$ and the cycle $\beta_i$ intersects with some cycles in $\Gamma_i$. 
Recall from \Cref{claim:intersection-between-two-cycles} that if two shortest cycles share an edge, then their intersection is a path. 
Moreover, such a path due to intersection of $\beta_i$ with some cycle in $\Gamma_i$, must be a subpath of $\gamma_i$ from \Cref{lem:intersection-with-ancestor-is-more}. Consider the path $P_i = \beta_i \cap \gamma_i$. 
We will now look at the intersections of $\gamma_i$ with cycles in $\Gamma_i$, and let $\iota_i^1, \iota_i^2, \cdots \iota_i^{q_i}$ be the ordering of elements in $\Gamma_i$ based on how the paths $\gamma_i\cap \iota^j_i$ appear in the clockwise order of $\gamma_i$ (see \Cref{fig:cycle-intersections}). 
Note $P_i$ (if at all it intersects cycles from $\Gamma_i$), must start intersecting the stated path sequence at some $\psi_i =\iota_i^s$ and ends the intersection at some $\tau_i = \iota_i^t$ (see \Cref{fig:4c-1}). We will maintain these first and last paths of intersections in the state $\theta_i$. 
We will exploit planarity to ensure the above can be done in the desired amount of time (using \Cref{obs:planar-gamma}). 
\end{enumerate}

To formalize the above we introduce some notations and prove some useful properties. 

\subparagraph*{Useful Notations.} 
Consider $\delta\in V(\C{T})$, and a $\delta$-lineage-free set $R\subseteq \C{M}_\delta$. 
Recall for $\gamma\in R$, the set $\Gamma_{\delta, R,\gamma} = \{\gamma' \mid \gamma' \in (R\cup \{\delta\}) \setminus \{\gamma\}, E(\gamma')\cap E(\gamma) \neq \emptyset\}$ 
is the $(\delta,R,\gamma)$-relevant subset of $R$. Let
$\C{S}_{\delta,R,\gamma} = \{\F, \N\} \cup \big(\Gamma_{\delta,R,\gamma} \times \Gamma_{\delta,R,\gamma} \big)
$, for $\gamma\in R$. We will use $\C{S}_{\delta,R,\gamma}$ to capture the discussed intersections of a descendant of $\gamma$ in the solution with other cycles in $\Gamma_{\delta,R,\gamma}$. 

Consider $\gamma \in R$, and we let $\gamma^\circ$ be the clockwise cyclic ordering of edges in $E(\gamma)$. We use the convention that, for instance, $(a,b,c), (b,c,a), (c,a,b)$, denote the same cyclic ordering. We attach the symbol $\circ$ to denote cyclic ordering, and all such orderings will be the clockwise cyclic ordering obtained using the given planar embedding of $G$. 
We now create the cyclic ordered set of elements in $\Gamma_{\delta,R,\gamma}$, in the order in which the (starting) point of their intersections with $\gamma$ appears in $\gamma^{\circ}$. Note that from \Cref{obs:split-relation-between-pd-and-pdrs} for distinct $\iota, \iota' \in \Gamma_{\delta,R,\gamma}$, $E(\iota) \cap E(\iota') \cap E(\gamma) = \emptyset$. 
Due to the above we obtain the unique (clockwise) cyclic ordering $\Gamma^\circ_{\delta,R,\gamma}$ of $\iota \in \Gamma_{\delta,R,\gamma}$, based on the order in which the path $\iota \cap \gamma$ appears in $\gamma$. For (not necessarily distinct) $\psi,\tau \in \Gamma_{\delta,R,\gamma}$ we let $\Gamma_{\delta,R,\gamma}[\psi,\tau]$ be the ordered set of elements starting from $\psi$ till $\tau$ in $\Gamma^\circ_{\delta,R,\gamma}$ (inclusive of both $\psi,\tau$). We note that $\Gamma_{\delta,R,\gamma}[\psi,\tau]$ is not a cyclic ordering, but a linear ordering. For simplicity, we will also treat $\Gamma_{\delta,R,\gamma}[\psi,\tau]$ as a set (when the ordering is irrelevant). 

Let $\beta$ be a descendant of $\gamma$, i.e., $\beta \leq \gamma$ (see \Cref{prop:lsct-use-property} (\ref{prop:desc-lsct})), or a sub-path of a descendant of $\gamma$, such that $\beta \neq \gamma$ and $E(\beta \cap \gamma) \neq \emptyset$. As $\beta\neq \gamma$, note that $\beta\cap \gamma$ must be a path (not a cycle) in $G$ (see \Cref{claim:intersection-between-two-cycles}). 
Let $\Gamma_{\delta,R,\gamma,\beta}$ be the subset of elements from $\Gamma_{\delta,R,\gamma}$ that share an edge with $\beta$, i.e., $\Gamma_{\delta,R,\gamma,\beta} =\{\iota \in \Gamma_{\delta,R,\gamma} \mid E(\iota) \cap E(\beta) \neq \emptyset\}$. 
We have the following simple observation regarding the set $\Gamma_{\delta,R,\gamma,\beta}$ which follows from \Cref{claim:intersection-between-two-cycles}. 

\begin{observation}\label{obs:consecutive}
If $\Gamma_{\delta,R,\gamma,\beta}\neq \emptyset$, then there are (not necessarily distinct) elements $\psi,\tau \in \Gamma_{\delta,R,\gamma}$ such that $\Gamma_{\delta,R,\gamma,\beta} = \Gamma_{\delta,R,\gamma}[\psi,\tau]$. 
\end{observation}

\subparagraph*{Auxiliary Table.}  
The next definition will allow us to guess the discussed intersection properties and look for solutions that have these properties. 

\begin{definition}\label{def:solution-of-pi-alpha-h-theta}
 
    Consider an $\alpha$-lineage-free set $R = \{\gamma_i\mid i\in [b]\}$, integers $h_1, h_2, \ldots, h_{b} \geq 1$ with $\sum_{i=1}^{b} h_i = h-1$, and for each $i\in [b]$, $\theta_i\in \C{S}_{\alpha,R,\gamma_i}$. We define $\Pi_{\Pcal,\alpha, h}(\gamma_i, h_i, \theta_i)_{i \in [b]}$ as the maximum over $w(\C{C} \cup \C{P})$, where $\C{C}$ is a solution for $\Pi_{\Pcal}(\alpha, h)$ and it additionally satisfies the following: 
    \begin{enumerate}
    \item
    \label{condition-1-of-def:solution-of-pi-alpha-h-theta}
    $R$ is the set of \pdrs{} for the \pds{} $\{\beta_i\mid i\in [b]\}$ of $\alpha$ in $\C{C}$. For each $i\in [b]$, $\gamma_i$ is the \pdr{} for $\beta_i$. 

    \item
    \label{condition-2-of-def:solution-of-pi-alpha-h-theta}
    For each $i\in [b]$, exactly $h_i$ cycles of $\Ccal$ are nodes in the subtree $\Tcal_{\beta_i}$.

    \item 
    \label{condition-3-of-def:solution-of-pi-alpha-h-theta}
    For $i\in [b]$, the following holds (see \Cref{fig:4c-1}):   
        \begin{enumerate}
            \item 
            \label{subcondition-1-of-condition-3-of-def:solution-of-pi-alpha-h-theta}
            If $\theta_i = \F$, then $\beta_i = \gamma_i$.
            \item 
            \label{subcondition-2-of-condition-3-of-def:solution-of-pi-alpha-h-theta}
            If $\theta_i = \N$, then $\beta_i \neq \gamma_i$ and $E(\beta_i) \cap E(\delta) = \emptyset$ for all $\delta \in \Gamma_{\alpha,R,\gamma_i}$.
            \item 
            \label{subcondition-3-of-condition-3-of-def:solution-of-pi-alpha-h-theta} Otherwise, let $\theta_i = (\psi_i,\tau_i)$. We have $\beta_i \neq \gamma_i$ and $\emptyset \neq \Gamma_{\alpha,R,\gamma_i,\beta_i} = \Gamma_{\alpha,R,\gamma_i}[\psi_i,\tau_i]$.
        \end{enumerate}
    \end{enumerate}
    In the above, a $\C{C}$ that satisfies the stated conditions (but may or may not achieve the maximum) is a {\em solution} for $\Pi_{\Pcal,\alpha, h}(\gamma_i, h_i, \theta_i)_{i \in [b]}$. Moreover, if it additionally achieves the maximum, then it is an {\em optimal} solution for $\Pi_{\Pcal,\alpha, h}(\gamma_i, h_i, \theta_i)_{i \in [b]}$. 
\end{definition}

The next lemma (proved in \Cref{sec:conputation-of-pi-alpha-h-using-theta}) helps us identify the correct solution descendants of an optimal solution, if we know how the solution behaves with respect to the definition stated above. For the remainder of this section we show how we can compute $\Pi_{\Pcal}(\alpha, h)$ using this lemma. 

\begin{lemma}[Proved in \Cref{sec:conputation-of-pi-alpha-h-using-theta}]\label{lem:compute-pi-ext}
Consider we are given $(G,w,k,d,\ell)$, $\C{T}$ and a $d$-diverse $\C{P}\subseteq \C{Q}$, a non-leaf node $\alpha$ of $\C{T}$, an integer $h >1$, an $\alpha$-lineage-free set $R = \{\gamma_i\mid i\in [b]\}$, integers $h_1,h_2, \ldots, h_{b} \geq 1$ with $\sum_{i=1}^{b} h_i = h-1$, and for $i\in [b]$, $\theta_i \in \C{S}_{\alpha,R,\gamma_i}$. Furthermore, we are given the values $\Pi_{\C{P}}(\alpha',h')$, for each $\alpha'\in V(\C{T}_{\alpha}) \setminus \{\alpha\}$ and $1\leq h'\leq k$. Then, in polynomial time, we can do the following: 
\begin{enumerate}
\item\label{itm:rep-no} correctly return {\sf NO}, indicating that there is no solution for $\Pi_{\Pcal,\alpha, h}(\gamma_i, h_i, \theta_i)_{i \in [b]}$, or 

\item\label{itm:ret-beta} return a set $\{\beta_i\mid i\in [b]\}$, such that if there is a solution for $\Pi_{\Pcal,\alpha, h}(\gamma_i, h_i, \theta_i)_{i \in [b]}$, then there is an optimal solution $\C{C}$ for $\Pi_{\Pcal,\alpha, h}(\gamma_i, h_i, \theta_i)_{i \in [b]}$ such that for each $i\in [b]$, $\beta_i \in \C{C}$ and $\gamma_i$ is the $(\alpha, \C{C})$-\pdr\ of $\beta_i$. 
\end{enumerate}
\end{lemma} 

We are ready to prove \Cref{lem:compute-pi}.   

\begin{proof}[Proof of \Cref{lem:compute-pi}] 

We have an instance $(G,w,k,d,l)$ of \pmaxddisjointcoverage{}, its LSCT, $\C{T}$, rooted at $C_r$, and a $d$-diverse $\C{P}\subseteq \C{Q}$. We now describe an algorithm to do the stated task in the lemma statement. 

Consider a node $\alpha$ of $\C{T}$ and an integer $h$, where $1\leq h \leq k$ (we will do the stated step for every choice of $\alpha,h$, in bottom-up fashion). Recall our sanity check and discussed base case are correct and they can be applied in polynomial time. 
Therefore, we now assume that $\alpha$ is a non-leaf node of $\Tcal$, $\alpha$ is $d$-diverse from $\Pcal$, and $h > 1$. 
Also, we assume that for each descendant $\alpha'$ of $\alpha$ in $\Tcal \setminus \{\alpha\}$ and any $1\leq h' \leq k$, we have already correctly computed $\Pi_{\Pcal}(\alpha', h')$. 
For each $b \in [h-1]$, let $\C{H}_b = \{(h_1,h_2,\ldots, h_b)\in [h-1]^b \mid \sum_{i\in [b]}h_i = h-1\}$, $\C{R}_{\alpha,b} = \{R\subseteq \C{M}_\alpha \mid R \mbox{ is } \alpha\mbox{-lineage-free}, |R|=b\}$. 
Furthermore, for $R\in \C{R}_{\alpha,b}$, let $\C{S}_{\alpha,R} = \{(\theta_1,\theta_2,\ldots, \theta_b) \mid \theta_i \in \C{S}_{\alpha,R,\gamma_i}, i\in [b]\}$. We initialize $M=\emptyset$. 

We do the following iteration over each $b \in [h-1]$, $(h_i)_{i\in [b]} \in \C{H}_b$, $R=\{\gamma_i\mid {i\in [b]}\} \in \C{R}_{\alpha,b}$,  $(\theta_i)_{i\in [b]} \in \C{S}_{\alpha,R}$. 
Let ${\sf ans}(\gamma_i,h_i, \theta_i)_{i\in [b]}$ be the answer returned by \Cref{lem:compute-pi-ext}, for $(G,w,k,d,\ell)$, $\C{T}, \C{P}$, $\alpha, h$, $R,(h_i)_{i\in [b]}, (\theta_i)_{i\in [b]}$. If ${\sf ans}(\gamma_i,h_i, \theta_i)_{i\in [b]} = {\sf NO}$, then we do nothing (and move to the next iteration). 

Else, let ${\sf ans}(\gamma_i,h_i, \theta_i)_{i\in [b]} = \{\beta_i \mid i\in [b]\}$, and we do the following. 
We construct a number which essentially will help us eliminate the double counting of edges (outside of edges in $\C{P}$) that appear in two distinct cycles from $\{\alpha\} \cup \{\beta_i \mid i\in [b]\}$. (Note that no three cycles from $\{\alpha\} \cup \{\beta_i \mid i\in [b]\}$ can intersect due to \Cref{obs:split-relation-between-pd-and-pdrs}. Let $E'$ be the set of edges from cycles in $\{\beta_i\mid i\in [b]\} \cup \{\alpha\}$, that appear in at least two of these cycles, i.e., $E' = \cup_{\beta, \beta' \in \{\alpha\} \cup \{\beta_i \mid i\in [b]\},\beta \neq \beta'} E(\beta) \cap E(\beta')$. Also, let $\varkappa = w(E'\setminus E(\C{P}))$. We next define a term, which is essentially counting the weight of edges of cycles in optimal solutions for $\Pi_{\C{P}}(\beta_i,h_i)$, $i\in [b]$, the cycle $\alpha$, and the paths in $\C{P}$. 
Note the weight of paths in $\C{P}$ is accounted for, in each $\Pi_{\C{P}}(\beta_i,h_i)$ (and not in $\varkappa$). 
We insert the number $w(E(\alpha) \setminus E(\C{P})) -(b-1)\cdot w(\C{P}) + \sum_{i\in [b]} \big(\Pi_{\C{P}}(\beta_i,h_i)) - \varkappa$ to $M$. 

Finally, we set $\Pi_{\C{P}}(\alpha,h) =\max_{x\in M}x$ (empty maximum is $-\infty$). 

\subparagraph{Running time analysis.}
From \Cref{prop:def:laminar-tree} the number of nodes in $\C{T}$ is bounded by a polynomial in $n$, and the number of distinct choices of $h$ is bounded by $k$. 
Now consider a fixed $\alpha \in V(\C{T})$, $h\in [k]$, and $b \in [h-1]$. 
Since the number of ordered partitions of a positive integer $x$ into non-zero integers is $2^{x-1}$, we have that $|\C{H}_{b}|$ is bounded by $2^{\Ocal(k)}$.
Additionally, by \Cref{obs:size-of-M}, $|\C{R}_{\alpha,b}|$ is bounded by $2^{O(k)}$. 
Using \Cref{prop:num-sum-prod} and \Cref{obs:planar-gamma} we can obtain that $|\C{S}_{\alpha,R}|$ is bounded by $2^{O(k)}$. 
Moreover, the above sets can be enumerated in the amount of time proportional to their sizes (with a multiplicative factor of a polynomial in $\log n$). 
Also, \Cref{lem:compute-pi-ext} is a polynomial time procedure. 
Every other stated step in the algorithm can be executed in polynomial time. 
Thus, the running time is bounded by $2^{O(k)}\cdot n^{O(1)}$.

\subparagraph{Correctness of the algorithm.}
Consider $\alpha \in V(\C{T})$, and a positive integer $h$, where $1\leq h\leq k$. The correctness of the base case (and the sanity check) is immediate from its description. So we now consider the case when $\alpha$ is a non-leaf node of $\C{T}$, $h >1$ and $\{\alpha\}$ is $d$-diverse from $\C{P}$. Our goal will be to show that $\Pi_{\C{P}}(\alpha,h) = \max_{x\in M}x$. 

\noindent{\em Proof of $\Pi_{\C{P}}(\alpha,h) \leq \max_{x\in M}x$}. 
Suppose that $\Pi_{\Pcal}(\alpha,h) \neq -\infty$ (else the inequality is trivially satisfied), and let $\C{C}^*$ be an optimal solution for $\Pi_{\Pcal}(\alpha,h)$ (see \Cref{def:pi-p-alpha-h}). 
Let $B = \{\beta^*_i\mid i\in [b]\}$ be the set of \pds{} of $\alpha$ in $\C{C}^*$ (see \Cref{def:productive-descendant}), where note that $1 \leq b \leq h-1$. Furthermore, for each $i\in [b]$, let $\gamma_i$ be the $(\alpha,\C{C}^*)$-representative for $\beta^*_i$ (see \Cref{def:productive-descendant-representative}), and let $R = \{\gamma_i \mid i\in [b]\}$. From \Cref{obs:relation-between-pd-and-pdrs}, $R \in \C{R}_{\alpha,b}$. For $i\in [b]$, let $h_i = |\C{C}^* \cap V(\C{T}_{\beta^*_i})|$, and note that $(h_i)_{i\in [b]} \in \C{H}_b$. For $i\in [b]$, we construct $\theta_i \in \C{S}_{\alpha,R}$ as follows. If $\beta^*_i = \gamma_i$, set $\theta_i = \F$. If $\beta^*_i \neq \gamma_i$ and $E(\beta^*_i) \cap E(\delta) = \emptyset$ for all $\delta \in \Gamma_{\alpha, R, \gamma_i}$, set $\theta_i = \N$. Otherwise, we set $\theta_i = (\psi_i, \tau_i)$, where $\emptyset \neq \Gamma_{\alpha,R,\gamma_i,\beta^*_i} = \Gamma_{\alpha,R,\gamma_i}[\psi_i,\tau_i]$ (see \Cref{obs:consecutive}).

By construction, $(\theta_i)_{i\in [b]} \in \C{S}_{\alpha, R}$. Furthermore, $\C{C}^*$ is a solution for $\Pi_{\Pcal,\alpha, h}(\gamma_i, h_i, \theta_i)_{i \in [b]}$. 
Thus, at the iteration where the algorithm considers $b$, $(h_i)_{i\in [b]}$, $R$, $(\theta_i)_{i\in [b]}$, the correctness of \Cref{lem:compute-pi-ext} implies that ${\sf ans}(\gamma_i,h_i, \theta_i)_{i\in [b]} = \{\beta_i \mid i\in [b]\}$, and there is an optimal solution $\C{C}$ for $\Pi_{\Pcal,\alpha, h}(\gamma_i, h_i, \theta_i)_{i \in [b]}$ such that for each $i\in [b]$, $\beta_i \in \C{C}$ and $\gamma_i$ is the $(\alpha, \C{C})$-\pdr\ of $\beta_i$. 
Note that $w(\C{C} \cup \C{P}) \geq w(\C{C}^* \cup \C{P})$, and thus $\C{C}$ is an optimal solution for $\Pi_{\Pcal}(\alpha, h)$. For $i\in [b]$, let $\C{C}_i = \C{C} \cap V(\C{T}_{\beta_i})$, and notice that from  \Cref{{obs:partitioning-of-a-solution}} (\ref{item-1-of-obs:partitioning-of-a-solution-opt-sol}), $\Pi_{\Pcal}(\beta_i, h_i) = w(\C{C}_i\cup \C{P})$. 
Let $\varkappa = w(E'\setminus E(\C{P}))$, where $E' = \cup_{\beta, \beta' \in \{\alpha\} \cup \{\beta_i \mid i\in [b]\},\beta\neq \beta'} E(\beta) \cap E(\beta')$. 
From the above, together with \Cref{lem:intersection-with-ancestor-is-more} and \Cref{obs:split-relation-between-pd-and-pdrs} we can obtain that, 
$$
    \Pi_{\Pcal}(\alpha, h) = w(\C{C}\cup \C{P}) = w(E(\alpha) \setminus E(\C{P})) -(b-1)\cdot w(\C{P}) + \big(\sum_{i\in [b]} w(\C{C}_i \cup \C{P})\big) - \varkappa
$$ 
$$ 
    = w(E(\alpha) \setminus E(\C{P})) -(b-1)\cdot w(\C{P}) + \big(\sum_{i\in [b]} \Pi_{\C{P}}(\beta_i,h_i)\big) - \varkappa
$$
Thus, we can conclude that $\Pi_{\C{P}}(\alpha,h) \leq \max_{x\in M}x$.

\noindent{\em Proof of $\Pi_{\C{P}}(\alpha,h) \geq \max_{x\in M}x$.} 
Suppose that $\max_{x\in M}x \neq - \infty$ (else the inequality is trivially satisfied), and let $x^* = \max_{x\in M}x$, where $x^*$ was added to $M$ for the iteration where the algorithm considered $b \in [h-1]$, $(h_i)_{i\in [b]} \in \C{H}_b$, $R =\{\gamma_i\mid i\in [b]\} \in \C{R}_{\alpha,b}$ and $(\theta_i)_{i\in [b]} \in \C{S}_{\alpha,R}$. Furthermore, the answer returned by \Cref{lem:compute-pi-ext}, for $(G,w,k,d,\ell)$, $\C{T}, \C{P}$, $\alpha, h$, $R,(h_i)_{i\in [b]}, (\theta_i)_{i\in [b]}$ is ${\sf ans}(\gamma_i,h_i, \theta_i)_{i\in [b]} = \{\beta_i \mid i\in [b]\}$, where there is an optimal solution $\C{C}$ for $\Pi_{\Pcal,\alpha, h}(\gamma_i, h_i, \theta_i)_{i \in [b]}$ such that for each $i\in [b]$, $\beta_i \in \C{C}$ and $\gamma_i$ is the $(\alpha, \C{C})$-\pdr\ of $\beta_i$. For this iteration we have $\varkappa = w(E'\setminus E(\C{P}))$, where $E' = \cup_{\beta, \beta' \in \{\alpha\} \cup \{\beta_i \mid i\in [b]\},\beta\neq \beta'} E(\beta) \cap E(\beta')$. Note that $\C{C}$ must also be a solution for $\Pi_{\C{P}}(\alpha,h)$ (see \Cref{def:solution-of-pi-alpha-h-theta}), and furthermore, $w(\C{C} \cup \C{P}) = x^*$. Thus, $\Pi_{\C{P}}(\alpha,h) \geq x^*$.

This concludes the proof of \Cref{lem:compute-pi} (subject to to the proof of \Cref{lem:compute-pi-ext}). 
\end{proof} 

    \subsection{Proof of \Cref{lem:compute-pi-ext}: Finding Solution Descendants For Structured Solutions}\label{sec:conputation-of-pi-alpha-h-using-theta}

The goal of this section is to prove \Cref{lem:compute-pi-ext}. We are given an instance $(G,w,k,d,l)$ of \pmaxddisjointcoverage, where recall $G$ is a clean plane graph with $V(G)=[n]$ and the outer face is the shortest cycle $C_r$, and for any cycle $C$ in $G$, we have $C \leq C_r$, $w: E(G) \rightarrow \mathbb{Z}_{>0}$, $k> 1$, $0 \leq d \leq g(G)/2$. Additionally, we are given a laminar shortest cycles tree $\C{T}$ of $G$, rooted at $C_r$, a $d$-diverse $\C{P}\subseteq \C{Q}$, a non-leaf node $\alpha$ of $\C{T}$, an integer $h >1$, an $\alpha$-lineage-free set $R = \{\gamma_i\mid i\in [b]\}$, $h_1, h_2 \ldots, h_{b} \in [h-1]$ with $\sum_{i=1}^{b} h_i = h-1$, and for each $i\in [b]$, $\theta_i \in \C{S}_{\alpha,R,\gamma_i}$. Furthermore, we are given the values $\Pi_{\C{P}}(\alpha',h')$, for each $\alpha'\in V(\C{T}_{\alpha}) \setminus \{\alpha\}$ and $h' \in [k]$. 

Let $R_{\F}, R_{\N}, R_{\sf o}$ be the partition (with some sets possibly empty) of $R$ based on $\theta_i$, for ${i\in [b]}$ as follows. For ${\sf x} \in \{\F,\N\}$, let $R_{\sf x} = \{\gamma_i\mid \theta_i = {\sf x}, i\in [b]\}$, and $R_{\sf o} = R\setminus (R_{\F} \cup R_{\N})$. 
For each $i\in [b]$, we will construct $\beta_i$, satisfying the premise of the lemma. 
Initialize $\beta_i = \phi$, for each $i\in [b]$. 
For each $i\in [b]$, if $\gamma_i \in R_{\F}$, then set $\beta_i = \gamma_i$. The correctness of the above follows from \Cref{def:solution-of-pi-alpha-h-theta} (\ref{subcondition-1-of-condition-3-of-def:solution-of-pi-alpha-h-theta}).

Consider an (unknown) optimal solution $\C{C}^*$ for $\Pi_{\Pcal,\alpha, h}(\gamma_i, h_i, \theta_i)_{i \in [b]}$ (if it exists). 
Furthermore, for $i\in [b]$, let $\beta^*_i \in \C{C}^*$ be such that $\gamma_i$ is the $(\alpha, \C{C}^*)$-\pdr{} of $\beta^*_i$. 
Note that for each $\gamma_i \in R_\F$, it must be the case that $\gamma_i\in \C{C}^{*}$, as $\gamma_i = \beta^*_i$ (see \Cref{def:solution-of-pi-alpha-h-theta} (\ref{subcondition-1-of-condition-3-of-def:solution-of-pi-alpha-h-theta})); so we have correctly set $\beta_i = \gamma_i$. Consider an $i\in [b]$, where $\gamma_i \in R \setminus R_{\F}$. From \Cref{def:solution-of-pi-alpha-h-theta} (\ref{subcondition-2-of-condition-3-of-def:solution-of-pi-alpha-h-theta}, \ref{subcondition-3-of-condition-3-of-def:solution-of-pi-alpha-h-theta}), it must be the case that $\beta^*_i \neq \gamma_i$, and thus, $\beta^*_i$ must be a strict descendant of $\gamma_i$ in $\C{T}_\alpha$. 
Moreover, \Cref{def:productive-descendant-representative} together with the construction of $\C{M}_\alpha$, implies that $\gamma_i$ must be a node of degree exactly $2$ in $\C{T}_\alpha(G)$ (else we can directly report {\sf NO}). 
Let $\gamma^{\sf s}_i$ be the unique child of $\gamma_i$ in $\C{T}_\alpha$ (see \Cref{fig:pd-and-pdr}). In the maximal path starting from $\gamma_i$ in $\C{T}_{\gamma_i}$ where all nodes are of degree exactly $2$ in $\C{T}$, let $\gamma^{\sf t}_i$ be the last node of this path. We have the following simple observation regarding $\gamma^{\sf s}_i$ and $\gamma^{\sf t}_i$.

\begin{observation}\label{obs:simple-prop-beta}
Consider $i\in [b]$ with $\theta_i \neq \F$. If $\gamma^{\sf t}_i=\gamma_i$, then there is no solution for $\Pi_{\Pcal,\alpha, h}(\gamma_i, h_i, \theta_i)_{i \in [b]}$. Moreover, for any solution $\what{\C{C}}$ for $\Pi_{\Pcal,\alpha, h}(\gamma_i, h_i, \theta_i)_{i \in [b]}$, where for $i\in [b]$, $\gamma_i$ is the $(\alpha,\what{\C{C}})$-\pdrs{} of $\what{\beta}_i \in \what{\C{C}}$, $\what{\beta_i}$ is a node in the unique path between $\gamma^{\sf s}_i$ and $\gamma^{\sf t}_i$ (including $\gamma^{\sf s}_i$ and $\gamma^{\sf t}_i$) in $\C{T}_\alpha$.\footnote{Here, we may possibly have $\gamma^{\sf s}_i=\gamma^{\sf t}_i$.}
\end{observation}
\begin{proof}
Follows from \Cref{def:productive-descendant-representative}, the construction of $\C{M}_\alpha$, and \Cref{def:solution-of-pi-alpha-h-theta} (\ref{condition-1-of-def:solution-of-pi-alpha-h-theta}, \ref{subcondition-2-of-condition-3-of-def:solution-of-pi-alpha-h-theta}, \ref{subcondition-3-of-condition-3-of-def:solution-of-pi-alpha-h-theta}). 
\end{proof}  

For $i\in [b]$, let $\Gamma_i = \Gamma_{\alpha,R,\gamma_i}$. Consider $i\in [b]$ where $\gamma_i \in R \setminus R_{\F}$. If $\gamma_i = \gamma^{\sf t}_i$, using the above observation we return {\sf NO}, correctly indicating that there is no solution for $\Pi_{\Pcal,\alpha, h}(\gamma_i, h_i, \theta_i)_{i \in [b]}$ (and stop). Else, let $P_i$ be the unique path between $\gamma^{\sf s}_i$ and $\gamma^{\sf t}_i$ in $\C{T}_\alpha$. 

We will now find $\beta_i$ when $\theta_i = \N$ (or report there is no solution for $\Pi_{\Pcal,\alpha, h}(\gamma_i, h_i, \theta_i)_{i \in [b]}$) as follows. 
For $\beta \in V(P_i)$, we say that $\beta$ is $(\gamma_i,h_i)$-{\em nice} if $\Pi_{\C{P}}(\beta,h_i) \neq -\infty$ and $E(\beta)\cap E(\iota) = \emptyset$, for every $\iota \in \Gamma_i$. 
Let $N_i = \{\iota \in V(P_i) \mid \iota  \text{ is } (\gamma_i,h_i)-\text{nice} \}$. 
Furthermore, if $N_i \neq \emptyset$, let $\iota_i = \argmax_{\beta \in N_i}\Pi_{\C{P}}(\beta,h_i)$.
(If there are multiple such $\beta$s, we arbitrarily pick one of them). 
The next observation (roughly speaking) establishes that $\iota_i$ is a correct choice for $\beta_i$.

\begin{observation}\label{lem:correct-beta-i-n}
    Consider $i\in [b]$, where $\gamma_i \in R_{\N}$ and $\gamma_i \neq \gamma^{\sf t}_i$. Then the following holds:
    \begin{enumerate}
        \item\label{item-1-of-lem:correct-beta-i-n} If $N_i=\emptyset$ then there is no solution for $\Pi_{\Pcal,\alpha, h}(\gamma_{i'}, h_{i'}, \theta_{i'})_{i' \in [b]}$. 
        \item\label{item-2-of-lem:correct-beta-i-new} 
        Suppose there is a solution $\what{\C{C}}$ for $\Pi_{\Pcal,\alpha, h}(\gamma_{i'}, h_{i'}, \theta_{i'})_{i' \in [b]}$, where $\gamma_{i'}$ is $(\alpha, \what{\C{C}})$-\pdr{} for $\hbeta_{i'} \in \hcal{C}$, and $\what{\C{C}}_{i'} = \what{\C{C}}\cap V(\C{T}_{\hbeta_{i'}})$, for $i'\in [b]$. 
        Then, there is a solution $\C{C}$ for $\Pi_{\Pcal,\alpha, h}(\gamma_{i'}, h_{i'}, \theta_{i'})_{i' \in [b]}$, such that $\gamma_i$ is the $(\alpha,\C{C})$-\pdr{} for $\iota_i \in \C{C}$, $\C{C} \setminus V(\C{T}_{\iota_i}) = \what{\C{C}} \setminus \what{\C{C}}_i$, and $w(\C{C} \cup \C{P}) \geq w(\what{\C{C}} \cup \C{P})$. 
    \end{enumerate} 
\end{observation}
\begin{proof}
Suppose that $N_i=\emptyset$. Notice that if there exists a solution $\Ccal$ for $\Pi_{\Pcal,\alpha, h}(\gamma_i, h_i, \theta_i)_{i \in [b]}$, then there exists $\beta \in \C{C}$ such that $\beta$ is a \pd{} of $\alpha$ in $\Ccal$ and $\gamma_i$ is the $(\alpha, \Ccal)$-\pdr{} for $\beta$. As $\gamma_i\in R_{\N}$ (i.e., $\theta_i = \N$), by \Cref{def:productive-descendant-representative}, 
$\beta$ must lie on the unique path between $\gamma_i^s$ and $\gamma_i^t$. As $N_i=\emptyset$, it must be the case that $\Pi_{\Pcal}(\beta, h_i) = -\infty$, or there is $\delta \in \Gamma_i$ such that $E(\beta) \cap E(\delta) \neq \emptyset$. The former case cannot happen due to \Cref{obs:partitioning-of-a-solution} (\ref{item-1-of-obs:partitioning-of-a-solution-opt-sol}) (also see, \Cref{def:solution-of-pi-alpha-h-theta}). The latter case cannot happen due to \Cref{def:solution-of-pi-alpha-h-theta} 
(\ref{subcondition-2-of-condition-3-of-def:solution-of-pi-alpha-h-theta}). Thus, when $N_i=\emptyset$, there is no solution for $\Pi_{\Pcal,\alpha, h}(\gamma_i, h_i, \theta_i)_{i \in [b]}$.

For second item, consider a solution $\what{\C{C}}$ for $\Pi_{\Pcal,\alpha, h}(\gamma_{i'}, h_{i'}, \theta_{i'})_{i' \in [b]}$, where for $i' \in [b]$, $\gamma_{i'}$ is the $(\alpha,\what{\C{C}})$-representative of $\hbeta_{i'}\in \what{\C{C}}$ and $\what{\Ccal}_{i'} = \what{\Ccal}\cap V(\Tcal_{\hbeta_{i'}})$. If $\beta_i = \iota_i$, then \cref{item-2-of-lem:correct-beta-i-new} holds for $\C{C}=\what{\C{C}}$. So we now suppose that $\beta_i \neq \iota_i$. By the choice of $\iota_i$ and the definition of $(\gamma_i, h_i)$-nice, $\Pi_{\Pcal}(\iota_i, h_i) \neq - \infty$; let $\Ccal_i^*$ be an optimal solution for $\Pi_{\Pcal}(\iota_i, h_i)$. Let $\Ccal = (\what{\Ccal}\setminus \what{\Ccal}_i) \cup \Ccal_i^*$. The construction of $\Ccal$ implies that $\Ccal$ is a solution for $\Pi_{\Pcal,\alpha, h}(\gamma_{i'}, h_{i'}, \theta_{i'})_{i' \in [b]}$ that satisfies the requirements of the observation. 
\end{proof}
From the above lemma, for any $i\in [b]$ where $\gamma_i \in R_{\N}$ and $\gamma_i \neq \gamma^{\sf t}_i$, we can set $\beta_i = \iota_i$. 

\subparagraph{Construction of an auxiliary graph.} We will focus on $i\in [b]$, where $\gamma_i\in R_{\sf o}$, and let $\theta_i = (\psi_i, \tau_i)$. Recall that for $i\in [b]$, for the (unknown) $\beta^*_{i}$, $P^*_{i} = \beta^*_{i} \cap \gamma_{i}$ must be a path in $G$. Furthermore, $\psi_{i}\cap \gamma_i$ (resp. $\tau_{i}\cap \gamma_i$) (see \Cref{def:solution-of-pi-alpha-h-theta} (\ref{subcondition-3-of-condition-3-of-def:solution-of-pi-alpha-h-theta})) is supposed to denote the ``first'' (resp. ``last'') intersection of $P^*_{i}$ with the cycles in $\Gamma_i$ (in the clockwise cyclic ordering $\gamma^\circ_i$). We construct an auxiliary graph that will precisely capture these ``end'' intersections, where we can have ``partial'' containment of the edges from the path $\iota \cap \gamma_{i}$ in $\beta^*_{i}$, for $\iota \in \Gamma_{i}$.  

We now state the construction of the auxiliary graph $H$. Set $V(H) = R_{\sf o}$, and for distinct $i,j\in [b]$, where $\gamma_i,\gamma_j \in R_{\sf o}$, we add the edge $\{\gamma_i,\gamma_j\}$ to $E(H)$ if and only if $\gamma_i \in \{\psi_j,\tau_j\}$ and $\gamma_j \in \{\psi_i,\tau_i\}$. We want to highlight that, for instance, if $(\psi_i,\tau_i) = (\gamma_j,\gamma_j)$ and $(\psi_j,\tau_j) = (\gamma_i,\gamma_i)$, then $\gamma_i, \gamma_j$ will be in a connected component with no other vertex and this component will still be a simple graph with exactly one edge, namely, $\{\gamma_i,\gamma_j\}$. 


The next observation easily follows from the fact that for any $i\in [b]$, where $\gamma_i\in R_{\sf o}$, the only possible neighbors of $\gamma_i$ in $H$ are $\psi_i$ and $\tau_i$. 

\begin{observation}\label{obs:h-cycle-path}
Each connected component of the graph $H$ is either i) a path (with one or more vertices), or ii) a cycle (with at least three vertices). 
\end{observation}

We will handle each connected component of $H$ separately. We begin by giving an intuitive description of our dynamic programming routine that will help us construct the $\beta_i$s for the $\gamma_i$s that appear in a fixed connected component of $H$. 


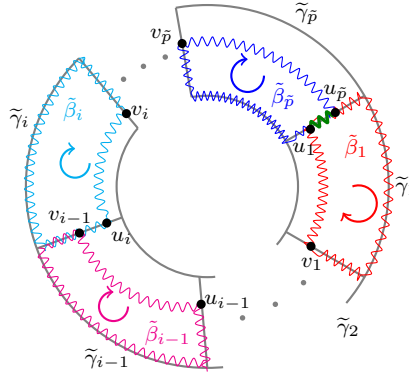
\begin{figure}[t]
\centering
\begin{tikzpicture}[scale=1.2]
\draw[thick, gray] (0,0) circle (2);
\draw[thick, gray] (0,0) circle (1);

\draw[line width=2pt, white]
(100:2)
arc[start angle=100, end angle=135, radius=2];

\draw[line width=2pt, white]
(100:1)
arc[start angle=100, end angle=140, radius=1];

\draw[line width=2pt, white]
(275:2)
arc[start angle=275, end angle=320, radius=2];

\draw[line width=2pt, white]
(275:1)
arc[start angle=275, end angle=320, radius=1];

\draw[thick, gray]
(100:2) -- (100:1);

\draw[thick, gray]
(135:2) -- (140:1);

\draw[thick, gray]
(270:2) -- (265:1);


\draw[thick, gray]
(200:2) -- (200:1);

\draw[thick, gray]
(30:2) -- (30:1);

\draw[thick, gray]
(330:2) -- (330:1);

\fill[gray] (108:1.5) circle (0.03);
\fill[gray] (118:1.5) circle (0.03);
\fill[gray] (128:1.5) circle (0.03);

\fill[gray] (285:1.5) circle (0.03);
\fill[gray] (300:1.5) circle (0.03);
\fill[gray] (315:1.5) circle (0.03);

\def\r{2}
\def\rmid{1.28}

\draw[red]
plot[smooth] coordinates {
(30:\r-0.05)
(28:\r+0.05)
(26:\r-0.05)
(24:\r+0.05)
(22:\r-0.05)
(20:\r+0.05)
(18:\r-0.05)
(16:\r+0.05)
(14:\r-0.05)
(12:\r+0.05)
(10:\r-0.05)
(8:\r+0.05)
(6:\r-0.05)
(4:\r+0.05)
(2:\r-0.05)
(0:\r+0.05)
(-2:\r-0.05)
(-4:\r+0.05)
(-6:\r-0.05)
(-8:\r+0.05)
(-10:\r-0.05)
(-12:\r+0.05)
(-14:\r-0.05)
(-16:\r+0.05)
(-18:\r-0.05)
(-20:\r+0.05)
(-22:\r-0.05)
(-24:\r+0.05)
(-26:\r-0.05)
(-28:\r+0.05)
(-30:\r-0.05)
};

\draw[red]
plot[smooth] coordinates {
(30:\rmid-0.05)
(28:\rmid+0.05)
(26:\rmid-0.05)
(24:\rmid+0.05)
(22:\rmid-0.05)
(20:\rmid+0.05)
(18:\rmid-0.05)
(16:\rmid+0.05)
(14:\rmid-0.05)
(12:\rmid+0.05)
(10:\rmid-0.05)
(8:\rmid+0.05)
(6:\rmid-0.05)
(4:\rmid+0.05)
(2:\rmid-0.05)
(0:\rmid+0.05)
(-2:\rmid-0.05)
(-4:\rmid+0.05)
(-6:\rmid-0.05)
(-8:\rmid+0.05)
(-10:\rmid-0.05)
(-12:\rmid+0.05)
(-14:\rmid-0.05)
(-16:\rmid+0.05)
(-18:\rmid-0.05)
(-20:\rmid+0.05)
(-22:\rmid-0.05)
(-24:\rmid+0.05)
(-26:\rmid-0.05)
(-28:\rmid+0.05)
(-30:\rmid-0.05)
};

\draw[red]
plot[smooth] coordinates {
(-30:\r-0.05)
(-32:1.95)
(-28:1.90)
(-32:1.85)
(-28:1.80)
(-32:1.75)
(-28:1.70)
(-32:1.65)
(-28:1.60)
(-32:1.55)
(-28:1.50)
(-32:1.45)
(-30:\rmid-0.05)
};

\draw[red]
plot[smooth] coordinates {
(30:\r-0.05)
(32:1.95)
(28:1.90)
(32:1.85)
(28:1.80)
(32:1.75)
(28:1.70)
(32:1.65)
(28:1.60)
(32:1.55)
(28:1.50)
(32:1.45)
(30:\rmid-0.05)
};

\def\rA{1.6}
\def\rB{1.0}

\draw[blue]
plot[smooth] coordinates {
(30:\rA-0.05)
(32:\rA+0.05)
(34:\rA-0.05)
(36:\rA+0.05)
(38:\rA-0.05)
(40:\rA+0.05)
(42:\rA-0.05)
(44:\rA+0.05)
(46:\rA-0.05)
(48:\rA+0.05)
(50:\rA-0.05)
(52:\rA+0.05)
(54:\rA-0.05)
(56:\rA+0.05)
(58:\rA-0.05)
(60:\rA+0.05)
(62:\rA-0.05)
(64:\rA+0.05)
(66:\rA-0.05)
(68:\rA+0.05)
(70:\rA-0.05)
(72:\rA+0.05)
(74:\rA-0.05)
(76:\rA+0.05)
(78:\rA-0.05)
(80:\rA+0.05)
(82:\rA-0.05)
(84:\rA+0.05)
(86:\rA-0.05)
(88:\rA+0.05)
(90:\rA-0.05)
(92:\rA+0.05)
(94:\rA-0.05)
(96:\rA+0.05)
(98:\rA-0.05)
(100:\rA+0.05)
};

\draw[blue]
plot[smooth] coordinates {
(30:\rB-0.05)
(32:\rB+0.05)
(34:\rB-0.05)
(36:\rB+0.05)
(38:\rB-0.05)
(40:\rB+0.05)
(42:\rB-0.05)
(44:\rB+0.05)
(46:\rB-0.05)
(48:\rB+0.05)
(50:\rB-0.05)
(52:\rB+0.05)
(54:\rB-0.05)
(56:\rB+0.05)
(58:\rB-0.05)
(60:\rB+0.05)
(62:\rB-0.05)
(64:\rB+0.05)
(66:\rB-0.05)
(68:\rB+0.05)
(70:\rB-0.05)
(72:\rB+0.05)
(74:\rB-0.05)
(76:\rB+0.05)
(78:\rB-0.05)
(80:\rB+0.05)
(82:\rB-0.05)
(84:\rB+0.05)
(86:\rB-0.05)
(88:\rB+0.05)
(90:\rB-0.05)
(92:\rB+0.05)
(94:\rB-0.05)
(96:\rB+0.05)
(98:\rB-0.05)
(100:\rB+0.05)
};

\draw[blue]
plot[smooth] coordinates {
(29:\rA-0.05)
(32:1.55)
(28:1.50)
(32:1.45)
(28:1.40)
(32:1.35)
(28:1.30)
(32:1.25)
(28:1.20)
(32:1.15)
(28:1.10)
(30:\rB-0.05)
};

\draw[blue]
plot[smooth] coordinates {
(100:\rA+0.05)
(102:1.55)
(98:1.50)
(102:1.45)
(98:1.40)
(102:1.35)
(98:1.30)
(102:1.25)
(98:1.20)
(102:1.15)
(98:1.10)
(100:\rB+0.05)
};

\draw[green!50!black, very thick]
plot[smooth] coordinates {
(29:\rA-0.01)
(32:1.55)
(28:1.50)
(32:1.45)
(28:1.40)
(32:1.35)
(28:1.30)
};

\begin{scope}[shift={(0:1.7)}]
    \draw[red, thick, ->]
        (-0.1,0)
        arc[start angle=110, end angle=-150, radius=0.2];
\end{scope}

\begin{scope}[shift={(70:1.3)}]
    \draw[blue, thick, ->]
        (0.18,0)
        arc[start angle=0, end angle=-270, radius=0.18];
\end{scope}

\fill (30:1.63) circle (0.05);
\fill (100:1.6) circle (0.05);

\node[font=\scriptsize] at (36:1.75) {$u_{\tilde{p}}$};
\node[font=\scriptsize] at (107:1.73) {$v_{\tilde{p}}$};

\fill (29:1.3) circle (0.05);
\fill (-30:1.32) circle (0.05);

\node[font=\scriptsize] at (21:1.15) {$u_{1}$};
\node[font=\scriptsize] at (-36.5:1.44) {$v_{1}$};


\node[font=\scriptsize] at (60:2.2) {$\tgamma_{\tilde{p}}$};
\node[font=\scriptsize] at (0:2.2) {$\tgamma_{1}$};
\node[font=\scriptsize] at (315:2.2) {$\tgamma_{2}$};
\node[font=\scriptsize] at (240:2.2) {$\tgamma_{i-1}$};
\node[font=\scriptsize] at (160:2.2) {$\tgamma_{i}$};

\node[font=\scriptsize, red] at (15:1.7) {$\tilde{\beta}_{1}$};
\node[font=\scriptsize, blue] at (50:1.3) {$\tilde{\beta}_{\tilde{p}}$};

\begin{scope}[shift={(170:1.48)}]
    \draw[cyan, thick, ->]
        (0.16,0) arc[start angle=0, end angle=-270, radius=0.16];
\end{scope}

\begin{scope}[shift={(232:1.68)}]
    \draw[magenta, thick, ->]
        (0.16,0) arc[start angle=0, end angle=-270, radius=0.16];
\end{scope}



\coordinate (GiTopLeft)     at (135:2);
\coordinate (GiTopRight)    at (200:1.9);
\coordinate (GiBottomLeft)  at (137:1.23);
\coordinate (GiBottomRight) at (200:1.12);

\tikzset{
  gicurve/.style={
    cyan,
    decorate,
    decoration={
      snake,
      amplitude=0.06cm,
      segment length=0.14cm
    }
  }
}

\draw[gicurve]
(134:1.97)
arc[start angle=134, end angle=204, radius=1.89];

\draw[gicurve]
(GiBottomLeft)
arc[start angle=140, end angle=200, radius=1.20];

\draw[gicurve]
(GiTopLeft) -- (GiBottomLeft);

\draw[gicurve]
(GiTopRight) -- (GiBottomRight);



\coordinate (GiTopLeftn)     at (265:1.3);
\coordinate (GiBottomLeftn)  at (270:2.04);

\coordinate (GiTopRight)    at (200:1.45);
\coordinate (GiBottomRight) at (200:1.9);

\coordinate (GiBottomLeft)  at (200:1.5);

\tikzset{
  gicurve/.style={
    magenta,
    decorate,
    decoration={
      snake,
      amplitude=0.06cm,
      segment length=0.14cm
    }
  }
}

\draw[gicurve]
(200:1.97)
arc[start angle=200, end angle=272, radius=1.89];

\draw[gicurve]
(GiBottomLeft)
arc[start angle=200, end angle=280, radius=1.20];

\draw[gicurve]
(GiTopLeftn) -- (GiBottomLeftn);

\draw[gicurve]
(GiTopRight) -- (GiBottomRight);

\node[font=\scriptsize, cyan] at (150:1.7) {$\tilde{\beta}_{i}$};
\node[font=\scriptsize, magenta] at (255:1.68) {$\tilde{\beta}_{i-1}$};

\fill (138:1.2) circle (0.05);
\fill (200:1.18) circle (0.05);

\node[font=\scriptsize] at (211:1.1) {$u_{i}$};
\node[font=\scriptsize] at (132:1.1) {$v_{i}$};

\fill (267:1.3) circle (0.05);
\fill (200:1.5) circle (0.05);

\node[font=\scriptsize] at (280:1.3) {$u_{i-1}$};
\node[font=\scriptsize] at (193:1.55) {$v_{i-1}$};

\end{tikzpicture}

    \caption{Illustration for the case when $\wtilde{H} = (\tgamma_1, \tgamma_2, \ldots, \tgamma_{\tilde{p}})$ is a cycle. The gray closed curves illustrate the embedding of the cycles $\tgamma_i$s on the plane. The curly closed curves in red and blue colors illustrate $\tbeta_1$ and $\tbeta_{\tilde{p}}$, respectively. Also, $\tbeta_i\cap \tgamma_i = \tgamma_i[u_i,v_i]$, for $i\in [\tilde{p}]$. The thick curly green curve illustrates $\tbeta_1\cap \tbeta_{\tilde{p}}$.}
    \label{fig:H-cycle-select}
\end{figure}

\subparagraph{Intuitive description of our dynamic programming routine.} Consider a connected component $\wtilde{H}= (\tgamma_{1}, \tgamma_{2}, \cdots, \tgamma_{\tilde{p}})$ of $H$, where for $i\in [\tilde{p}]$, $\tgamma_i = \gamma_{i_j}$. 
Recall $\wtilde{H}$ must either be a cycle or a path. For each $i \in [\tilde{p}]$, we initialize $\wtilde{\beta}_i = \beta_{i_j}$ (recall the initialization $\beta_{i_j} = \phi$). 
Intuitively speaking, the states of our table will capture the following:
\begin{itemize}
\item We begin by fixing pairs $\varepsilon_1 = (u_1,v_1)$ and $\varepsilon_{\tilde{p}} = (u_{\tilde{p}},v_{\tilde{p}})$ that are $d$-diverse (i.e., $w(\tgamma_1[u_1,v_1]\cap \tgamma_{\tilde{p}}[u_{\tilde{p}},v_{\tilde{p}}]) \leq d$)---these will respectively denote the endpoints of the subpaths $\tgamma_1[u_1,v_1] = \tbeta_1\cap \tgamma_1$ and $\tgamma_{\tilde{p}}[u_{\tilde{p}},v_{\tilde{p}}] = \tbeta_{\tilde{p}}\cap \tgamma_{\tilde{p}}$ (see \Cref{fig:H-cycle-select}). 
We remark that the above two guesses help us in unifying the cases for paths and cycles. 

\item We will start identifying the ``best'' possible $\tbeta_i$ for each $q\in [\tilde{p}]$, starting from $q=1$. For $q=1$, we would like to choose a cycle that is ``compliant'' with the fixed intersection of $\tbeta_1 \cap \tgamma_1$, which should precisely be $\tgamma_1[u_1,v_1]$. 
In the above, knowing the exact edges in $\tbeta_1 \cap \tgamma_1$ helps us to ``extend'' it meaningfully as our dynamic programming progresses. 
At the $q^{th}$ stage, where $q\in [\tilde{p}] \setminus \{1\}$, we will assume that we have identified the ``best choices'' for each $i \in [q-1]$, where we additionally know the best such choice for $\tbeta_{i-1}\cap \tgamma_{i-1}$. 
Now for a ``guessed'' choice of the subpath $\tgamma_q[u_q, v_q] = \tbeta_q\cap \tgamma_q$, we will find the ``best'' choices of the cycles for each $i\in [q-1]$, by querying the entries for computed up to $q-1$. 
Here, the crux of why this dynamic programming works is that a cycle $\tbeta_i$ can only ``non-trivially'' intersect with cycles $\tbeta_{i-1}$ and $\tbeta_{i+1}$. 
In the above, when we say ``non-trivial'' we mean the case where the common edge set is not exactly $\tgamma_i \cap \tgamma_{i'}$, for $i'\in [b]$. The above is due to the construction of $H$. 
We additionally remark that an edge can belong to at most two of the cycles from $\tbeta_{i-1}, \tbeta_i, \tbeta_{i+1}$ due to \Cref{obs:split-relation-between-pd-and-pdrs}. 

For those $i'\in [b]$, where $\beta_{i'}\cap \tbeta_q$ must be the same as $\gamma_{i'} \cap \tgamma_{q}$, the contribution of this intersection will be a fixed number when we are looking for solutions that are ``$\theta_i$-abiding'', for each $i\in [b]$. 
Thus the above part of the intersections can be handled separately, as $\tgamma_i$s alone determine the above discussed number.
\end{itemize}

\subparagraph{Useful notations.} We introduce some notations that will help us in defining/computing the states of our dynamic programming table over connected components of $H$. 

To capture such pairwise edge intersections (outside $\C{P}$) of cycles we define the following. For two subgraphs $G_1$ and $G_2$ of $G$, let $\zeta_{G_1,G_2}$ denote the weight of the edges common to both graphs that are not in $\Pcal$, i.e., $\zeta_{G_1,G_2} = w\big(E(G_1 \cap G_2) \setminus E(\Pcal)\big)$. 

Consider a connected component $\wtilde{H}$ of $H$, and we let $\wtilde{H} = (\tgamma_{1}, \tgamma_{2}, \cdots ,\tgamma_{\tilde{p}})$, which is a cycle or a path (see \Cref{obs:h-cycle-path}). Note that $V(\wtilde{H}) \subseteq R_{\sf o}$, and for $i \in [\tilde{p}]$, we let $i_j \in [b]$ such that $\gamma_{i_j} = \tgamma_i$. Furthermore, let $\wth_i = h_{i_j}$, $\tpsi_i = \psi_{i_j}$, $\ttau_i = \tau_{i_j}$, $\ttheta_i = \theta_{i_j}$, and $\wtilde{\Gamma}_i = \Gamma_{i_j}$. For $i\in [\tilde{p}]$, let $\wtilde{P}_i$ be the path $P_{i_j}$, which is the unique maximal path of vertices of degree $2$, between $\tgamma^{\sf s}_{i}=\gamma^{\sf s}_{i_j}$ and $\tgamma^{\sf t}_{i}=\gamma^{\sf t}_{i_j}$ in $\C{T}_{\gamma_i}$ (see \Cref{fig:pd-and-pdr}). 

In our dynamic programming, we will start from $q=1$ till $q=\tilde{p}$, and at state for $q$, we will ``guess'' the subpath $\beta^*_q\cap \gamma_q$. To this end we introduce the following notation. For $\gamma \in V(\C{T})$, and distinct $u,v\in V(\gamma)$, $\gamma[u,v]$ is the sub-path of $\gamma$ starting from $u$ till $v$ (inclusive of both $u$ and $v$) in the clockwise ordering $\gamma^\circ$. For $q\in [\tilde{p}]$ and distinct $u,v \in V(\tgamma_q)$, let $\C{D}_{q,u,v} = \{\beta \in V(\wtilde{P}_i) \mid \beta \cap \tgamma_q = \tgamma_q[u,v]\}$. 

Next, we define when a pair of vertices from $\gamma_i$ or a descendant of $\gamma_i$ is ``abiding'' for $\theta_i$, for $\gamma_i\in R_{\sf o}$. Intuitively speaking, it will be the case when \Cref{def:solution-of-pi-alpha-h-theta} (\ref{subcondition-3-of-condition-3-of-def:solution-of-pi-alpha-h-theta}) is satisfied by them. 

\begin{definition}
\label{def:path-comp-theta}
Consider $i \in [b]$, where $\gamma_i \in R_{\sf o}$. For distinct $u,v \in V(\gamma_i)$, $(u,v)$ is $\theta_i$-{\em abiding} if $\emptyset \neq \Gamma_{\delta,R,\gamma,\gamma[u,v]} = \Gamma_{\alpha,R,\gamma_i}[\psi_i,\tau_i]$. Similarly, $\beta \in V(P_i)$ is {\em $\theta_i$-abiding} if $\emptyset \neq \Gamma_{\delta,R,\gamma,\beta} = \Gamma_{\alpha,R,\gamma_i}[\psi_i,\tau_i]$. 
\end{definition}

We next define the notion of {\em $d$-good}, which intuitively speaking will take into account the $d$-diversity of the choice of path intersection $\tgamma_i[u,v]$ (or a descendant of $\gamma$) with edges in $\C{P}$, and with (the unknown) cycles $\psi_i$ and $\tau_i$, if these intersections are ``partial'' but they are not captured by the graph $H$. We remark here that we cannot simply add these missing edges as it will destroy the property that each connected component of $H$ is a path or a cycle---a property that aids us in our upcoming dynamic programming routine.\footnote{Note that not every $2$-degenerate graphs have their connected components as paths and cycles.} 

\begin{definition}\label{def:path-comp-theta-good}
Consider $i \in [b]$, where $\gamma_i \in R_{\sf o}$. For distinct $u,v \in V(\gamma_i)$, $(u,v)$ is $d$-good if the following holds:
\begin{enumerate}
\item For each $P\in \C{P}$, $w(P\cap \gamma_i[u,v]) \leq d$, $w(\alpha \cap \gamma_i[u,v]) \leq d$, and 
\item Consider $i'\in [b]\setminus \{i\}$, where i) $\gamma_{i'} \in \{\psi_i, \tau_i\}$ and ii) $\theta_{i'} ={\sf e}$, or $\gamma_{i'} \in R_{\sf o}$ and $\gamma_i \in \Gamma_{\alpha,R,\gamma_{i'}}[\psi_{i'}, \tau_{i'}] \setminus \{\psi_{i'}, \tau_{i'}\}$. Then, we have $w(\gamma[u,v] \cap \gamma_{i'}) \leq d$.
\end{enumerate} 
Similarly, for $\beta\in V(P_i)$, $\beta$ is {\em $d$-good} if the following holds: 
\begin{enumerate}
\item For each $P\in \C{P}$, $w(P\cap \beta) \leq d$, $w(\alpha \cap \beta) \leq d$, and 
\item Consider $i'\in [b]\setminus \{i\}$, where i) $\gamma_{i'} \in \{\psi_i, \tau_i\}$ and ii) $\theta_{i'} ={\sf e}$, or $\gamma_{i'} \in R_{\sf o}$ and $\gamma_i \in \Gamma_{\alpha,R,\gamma_{i'}}[\psi_{i'}, \tau_{i'}] \setminus \{\psi_{i'}, \tau_{i'}\}$. Then, we have $w(\beta \cap \gamma_{i'}) \leq d$.
\end{enumerate} 
\end{definition}

For $i \in [b]$, where $\gamma_i \in R_{\sf o}$, let $\C{A}_i$ be the set of ordered pairs of distinct vertices from $V(\gamma_i)$ that are $\theta_i$-abiding and $d$-good. For $i \in [\tilde{p}]$, let $\tcal{A}_i = \C{A}_{i_j}$ (recall, $\tgamma_i = \gamma_{i_j}$). 

Consider a connected component $\wtilde{H} = (\tgamma_{1}, \tgamma_{2}, \cdots ,\tgamma_{\tilde{p}})$. Note for each $i\in [\tilde{p}]$, we need to choose $\wtilde{\beta}_i$ ($= \beta_{i_j}$), that satisfies the requirements of \Cref{lem:compute-pi-ext} (see \Cref{fig:H-cycle-select}). To do this, we will be fixing the intersections $\wtilde{\beta}_1\cap \tgamma_1$ and $\wtilde{\beta}_{\wtilde{p}} \cap \tgamma_{\wtilde{p}}$ (these two guesses helps us unify the cases for paths and cycles). To this end, we introduce the following. For $\varepsilon_1 = (u_1,v_1) \in \C{\wtilde{A}}_1$ and $\varepsilon_{\tilde{p}} = (u_{\tilde{p}},v_{\tilde{p}}) \in \C{\wtilde{A}}_{\tilde{p}}$, $(\varepsilon_1,\varepsilon_{\tilde{p}})$ is {\em $d$-diverse} if $w(\tgamma_1[u_1,v_1]\cap \tgamma_{\tilde{p}}[u_{\tilde{p}},v_{\tilde{p}}]) \leq d$.

\subparagraph{Dynamic programming states for computing $\beta_i$s.} We now formalize the intuitive description of the states that we gave previously, in the next definition. Consider a connected component $\wtilde{H} = (\tgamma_1,\tgamma_2,\ldots, \tgamma_{\tilde{p}})$. We remark that, for simplicity, our dynamic programming will actually be solving an optimization version of the problem, however, by simple memoizations, we can actual compute the $\wtilde{\beta}_i$s.

\begin{definition}\label{def:H-dp}
Consider $\varepsilon_1 = (u_1,v_1) \in \C{\wtilde{A}}_1$ and $\varepsilon_{\tilde{p}} = (u_{\tilde{p}},v_{\tilde{p}}) \in \C{\wtilde{A}}_{\tilde{p}}$ (see \Cref{def:path-comp-theta}), where $(\varepsilon_1,\varepsilon_{\tilde{p}})$ is $d$-diverse if $\tilde{p}\geq 2$. For $q \in [\tilde{p}]$, $(u,v) \in \C{\wtilde{A}}_q$, we define $\Lambda_{\varepsilon_1,\varepsilon_{\tilde{p}}}(q,u,v)$ as the maximum over $\sum_{i\in [q]}\Pi_{\C{P}}(\tbeta_i, \wth_i)  - \sum_{i \in [q-1]}{\zeta_{\tbeta_i, \tbeta_{i+1}}}$, 
where $\wtilde{\C{C}} = \{\tbeta_i\mid i\in [q]\}$ is a set of $q$ tree shortest cycles in $G$ such that: 

\begin{enumerate}
\item\label{item-1-of-def:H-dp} for each $i\in [q]$, we have $\tbeta_i \in V(\wtilde{P}_i)$, $\tbeta_i$ is $\ttheta_i$-abiding and $d$-good.  

\item\label{item-2-of-def:H-dp} $\tbeta_q \cap \tgamma_q = \tgamma_q[u,v]$, and for $i' \in [q] \cap \{1, \tilde{p}\}$,  $\tbeta_{i'} \cap \tgamma_{i'} = \tgamma_{i'}[u_{i'}, v_{i'}]$, and

\item \label{item-3-of-def:H-dp}
for each $i \in [q-1]$, $\wtilde{\beta}_i$ and $\wtilde{\beta}_{i+1}$ are $d$-diverse, i.e., $w(E(\beta_i) \cap E(\beta_{i+1})) \leq d$.
\end{enumerate}
In the above, if such a $\tcal{C}$ does not exist, then $\Lambda_{\varepsilon_1,\varepsilon_{\tilde{p}}}(q,u,v) = -\infty$. 
\end{definition}

Fix $\varepsilon_{1} = (u_1,v_1) \in \tcal{A}_1$ and $\varepsilon_{\tilde{p}}=(u_{\tilde{p}},v_{\tilde{p}}) \in \tcal{A}_{\tilde{p}}$, where if $\tilde{p}\geq 2$, then $(\varepsilon_{1}, \varepsilon_{\tilde{p}})$ are $d$-diverse. We will next describe the (recursive) computation of $\Lambda_{\varepsilon_1,\varepsilon_{\tilde{p}}}(q,u,v)$, for each $q\in [\tilde{p}]$ and $(u,v)\in \tcal{A}_q$, starting with the base case.  

\vspace{1mm}\noindent{\em Base Case ($q=1$).} Consider $u,v\in \tcal{A}_1$. If $(u,v)\neq (u_1,v_1)$, set $\Lambda_{\varepsilon_1,\varepsilon_{\tilde{p}}}(q,u,v) =-\infty$ (see \Cref{def:H-dp} (\ref{item-2-of-def:H-dp})). Otherwise, we set $\Lambda_{\varepsilon_1,\varepsilon_{\tilde{p}}}(q,u,v) = \max_{\beta\in \C{D}_{1,u,v}} \Pi_{\C{P}}(\beta,\wth_1)$. The correctness and the applicability of the base case in polynomial time is immediate from its description (and the definition of $\tcal{A}_1$).

\vspace{1mm}\noindent{\em Recursive formula.} Consider $q\in [\tilde{p}]\setminus \{1\}$, and $(u,v) \in \tcal{A}_q$. We assume that the entries $\Lambda_{\varepsilon_1,\varepsilon_{\tilde{p}}}(q',u',v')$ have been computed correctly for each $q' \in [q-1]$ and $(u',v') \in \tcal{A}_{q'}$. Notice that if $\tilde{p} = q$ and $(u, v) \neq (u_p, v_p)$, then we can correctly set $\Lambda_{\varepsilon_1,\varepsilon_{\tilde{p}}}(q,u,v) =-\infty$ (see \Cref{def:H-dp} (\ref{item-2-of-def:H-dp})). 

Intuitively, for our recursive formula we query the entry for $q-1$, over every choice of $\beta\in \C{D}_{q,u,v}$ and the choice of intersection, $\tgamma_i[u',v']$, between the cycle $\tbeta_{q-1}$ and $\tgamma_{q-1}$. We remark that while doing the above we need to ensure that the cycle selected for $q-1$ and the cycle selected from $q$, should be $d$-diverse from each other. This leads us to the following recursive formula. 
\begin{equation}\label{eq:H-rec}
\Lambda_{\varepsilon_1,\varepsilon_{\tilde{p}}}(q,u,v) = \max_{\substack{(\beta, (u', v')) \in \C{D}_{q,u,v} \times \tcal{A}_{q-1}: \\ w(E(\tgamma_{q-1}[u',v'] \cap \tgamma_q[u,v])) \leq d}} \big(\Pi_{\C{P}}(\beta, \wth_q) +  \Lambda_{\varepsilon_1,\varepsilon_{\tilde{p}}}(q-1,u',v') - \zeta_{\tgamma_{q-1}[u',v'], \tgamma_q[u, v]}
\big)
\end{equation}

In the above equation we note that the minus term is to reduce the weight of the edges (outside $\C{P}$) from $\beta$ that belong to the cycle selected for $q-1$, where note that the intersection must be precisely $\tgamma_{q-1}[u',v'] \cap \tgamma_q[u, v]$. In the next lemma we prove the correctness of the recursive formula. 

\begin{lemma}\label{lem:H-rec-correct}
In polynomial time, for each $\varepsilon_1 = (u_1,v_1) \in \C{\wtilde{A}}_1$, $\varepsilon_{\tilde{p}} = (u_{\tilde{p}},v_{\tilde{p}}) \in \C{\wtilde{A}}_{\tilde{p}}$, where $(\varepsilon_1,\varepsilon_{\tilde{p}})$ is $d$-diverse if $\tilde{p} \geq 2$, $q \in [\tilde{p}]$, and $(u,v) \in \C{\wtilde{A}}_q$, we can compute $\Lambda_{\varepsilon_1,\varepsilon_{\tilde{p}}}(q,u,v)$. Moreover, in polynomial time, using memoization, we can also find $\wtilde{\C{C}} = \{\tilde{\iota}_i \mid i \in [q]\}$ that satisfies all the conditions of \Cref{def:H-dp} for $\varepsilon_1,\varepsilon_{\tilde{p}}$, $q$, and $(u,v)$, that achieves the stated maximum. 
\end{lemma}

\begin{proof}
From the description of the base case and \Cref{eq:H-rec}, to obtain a proof of the lemma, it is enough to prove that recurrence relation given by \Cref{eq:H-rec} is correct. We first show that:
\begin{equation*}
\Lambda_{\varepsilon_1,\varepsilon_{\tilde{p}}}(q,u,v) \leq 
\max_{\substack{(\beta,(u',v')) \in \C{D}_{q,u,v} \times \tcal{A}_{q-1}:\\
w(E(\tgamma_{q-1}[u',v'] \cap \tgamma_q[u,v])) \leq d}}
\Big(
\Pi_{\C{P}}(\beta,\wth_q) + \Lambda_{\varepsilon_1,\varepsilon_{\tilde{p}}}(q-1,u',v')
- \zeta_{\tgamma_{q-1}[u',v'],\tgamma_q[u,v]}
\Big)
\end{equation*}
We assume that $\Lambda_{\varepsilon_1, \varepsilon_{\tilde{p}}}(q,u,v) \neq -\infty$, as otherwise, the inequality is trivially satisfied. Consider $\wtilde{\C{C}} = \{\tbeta_i \mid i\in [q]\}$ that satisfies the conditions stated in \Cref{def:H-dp}, such that $\Lambda_{\varepsilon_1,\varepsilon_{\tilde{p}}}(q,u,v) = \sum_{\tbeta_i \in \wtilde{\C{C}}}\Pi_{\C{P}}(\tbeta_i, \wth_i)  - \sum_{i \in [q-1]}{\zeta_{\tbeta_i, \tbeta_{i+1}}}$. 
Note that from \Cref{def:H-dp} (\ref{item-1-of-def:H-dp}) (also see \Cref{def:path-comp-theta}), $\theta_{q-1}$-abiding and from the construction of $H$, there must exist distinct vertices $\what{u},\what{v}\in V(\tgamma_{q-1})$, such that $\tbeta_{q-1} \cap  \tgamma_{q-1} = \tgamma_{q-1}[\what{u},\what{v}]$. Let $\what{\C{C}} = \wtilde{\C{C}}\setminus \{\tbeta_q\}$.  
Notice that by the constructions of $\C{D}_{q,u,u}$, $\tcal{A}_{q-1}$, and \Cref{def:H-dp} (\ref{item-3-of-def:H-dp}), we have $(\tbeta_q, (\what{u}, \what{v}))\in \C{D}_{q,u,u} \times \tcal{A}_{q-1}$ and $w(E(\tgamma_{q-1}[\what{u},\what{v}] \cap \tgamma_q[u,v])) \leq d$. Also, 
\begin{align*}
\Lambda_{\varepsilon_1,\varepsilon_{\tilde{p}}}(q,u,v) & = \sum_{i\in [q]}\Pi_{\C{P}}(\tbeta_i, \wth_i) - \sum_{i \in [q-1]}{\zeta_{\tbeta_i, \tbeta_{i+1}}}\\ 
  & = \Pi_{\C{P}}(\tbeta_q, \wth_q) + \big(\sum_{i\in [q-1]} \Pi_{\C{P}}(\tbeta_i, \wth_i) - \sum_{i \in [q-2]}{\zeta_{\tbeta_i, \tbeta_{i+1}}} 
  \big) 
  - \zeta_{\tbeta_{q-1}, \tbeta_q} 
\end{align*}
From the construction of $\what{u},\what{v}$ and \Cref{def:H-dp} (\ref{item-2-of-def:H-dp}), we have $\tbeta_{q-1} \cap \tbeta{q} = \tgamma[\what{u}, \what{v}] \cap \tgamma[u, v]$. Thus, $\Lambda_{u_1,v_1}(q,u,v) \leq  \Pi_{\C{P}}(\tbeta_q, \wth_q) + \Lambda_{\varepsilon_1, \varepsilon_{\tilde{p}}}(q-1, \what{u}, \what{v}) - \zeta_{\tgamma_{q-1}[\what{u}, \what{v}], \tgamma_q[u, v]}$, 
and the required inequality holds. 

We next show the other direction of inequality, i.e.,
\begin{equation*}
\Lambda_{\varepsilon_1,\varepsilon_{\tilde{p}}}(q,u,v) \geq \max_{\substack{(\beta, (u', v')) \in \C{D}_{q,u,v} \times \tilde{\C{A}}_{q-1}: \\ w(E(\tgamma_{q-1}[u',v'] \cap \tgamma_q[u,v])) \leq d}} \big(\Pi_{\C{P}}(\beta, \wth_q) +  \Lambda_{\varepsilon_1,\varepsilon_{\tilde{p}}}(q-1,u',v') - \zeta_{\tgamma_{q-1}[u',v'], \tgamma_q[u, v]} 
\big).
\end{equation*}


We assume that the right-hand side of the above inequality is not $-\infty$, as otherwise, trivially, the inequality is satisfied. Consider $\tbeta_q \in \C{D}_{q, u, v}$ and $(\what{u}, \what{v}) \in \tcal{A}_{q-1}$ with $w(E(\tgamma_{q-1}[\what{u}, \what{v}] \cap \tgamma_q[u,v])) \leq d$ that achieves the maximum on the right hand side of the inequality. To prove the inequality, it is enough to prove that $\Lambda_{\varepsilon_1,\varepsilon_{\tilde{p}}}(q,u,v) \geq \Pi_{\C{P}}(\tbeta_q, \wth_q) +  \Lambda_{\varepsilon_1,\varepsilon_{\tilde{p}}}(q-1,\what{u}, \what{v}) - \zeta_{\tgamma_{q-1}[\what{u}, \what{v}], \tgamma_q[u, v]}$.

Consider $\hcal{C} = \{\tbeta_i \mid i \in [q-1]\}$ that satisfies the conditions stated in \Cref{def:H-dp} for $q-1$ and $(\what{u}, \what{v})$ such that $\Lambda_{\varepsilon_1,\varepsilon_{\tilde{p}}}(q-1,\what{u}, \what{v}) = \sum_{i\in [q-1]}\Pi_{\C{P}}(\tbeta_i, \wth_i)  - \sum_{i \in [q-2]}{\zeta_{\tbeta_i, \tbeta_{i+1}}} $. 
Let $\tcal{C} = \hcal{C} \cup \{\tbeta_q\}$. By the construction of $\C{D}_{q,u,v}$ and the fact that $(u,v)$ if $\ttheta_q$-abiding implies that $\tbeta_i$ is $\ttheta_q$-abiding and $\tbeta_q \cap \tgamma_q = \tgamma_q[u,v]$. Also, $\tbeta_q \cap \tbeta_{q-1} = \tgamma_q[u,v]\cap \tgamma_{q-1}[\what{u},\what{v}]$, and thus, $w(\tbeta_q \cap \tbeta_{q-1}) \leq d$. The above together with the assumption that $\hcal{C}$ satisfies all the conditions in \Cref{def:H-dp} for $q-1$ and $(\what{u}, \what{v})$, implies that $\tcal{C}$ satisfies all the conditions in the definition for $q$ and $(u,v)$. Next, notice that:
\begin{align*}
\Lambda_{\varepsilon_1,\varepsilon_{\tilde{p}}}(q, u, v)
    & \geq  \sum_{i\in [q]} \Pi_{\C{P}}(\tbeta_i, \wth_i)  - \sum_{i \in [q-1]}{\zeta_{\tbeta_i, \tbeta_{i+1}}} \\
     & =\Pi_{\C{P}}(\tbeta_q, \wtilde{h}_q) + \big(\sum_{i\in [q-1]} \Pi_{\C{P}}(\tbeta_i, \wth_i)  - \sum_{i \in [q-2]}{\zeta_{\tbeta_i, \tbeta_{i+1}}} 
    \big) 
     - \zeta_{\tbeta_{q-1}, \tbeta_q}\\
     & = \Pi_{\C{P}}(\tbeta_q, \wtilde{h}_q) + \Lambda_{\varepsilon_1,\varepsilon_{\tilde{p}}}(q-1,\what{u}, \what{v}) - w(E(\beta_{q-1} \cap \beta_q) \setminus E(\Pcal)) \\
     & = \Pi_{\C{P}}(\tbeta_q, \wtilde{h}_q) + \Lambda_{\varepsilon_1,\varepsilon_{\tilde{p}}}(q-1,\what{u}, \what{v}) - w(E(\tgamma_{q-1}[u',v'] \cap \tgamma_q[u, v]) \setminus E(\Pcal))\\ 
    & = \Pi_{\C{P}}(\tbeta_q, \wtilde{h}_q) + \Lambda_{\varepsilon_1,\varepsilon_{\tilde{p}}}(q-1,\what{u}, \what{v}) - \zeta_{\tgamma_{q-1}[u',v'], \tgamma_q[u, v]}.
\end{align*}
This concludes the proof. 
\end{proof}

\subparagraph{Finding the best pair from $\tcal{A}_1\times\tcal{A}_{\tilde{p}}$.} We will now choose $(\varepsilon^*_1, \varepsilon^*_{\tilde{p}})$ as a pair from $\tcal{A}_1 \times \tcal{A}_{\tilde{p}}$, where $\varepsilon^*_1 = (u_1,v_1)$ and $\varepsilon^*_1 = (u_{\tilde{p}},v_{\tilde{p}})$ are $d$-diverse when $\tilde{p}\geq 2$, which is the ``best choice'' in the following sense. Roughly speaking it maximizes the total weight of the cycles we can get as per \Cref{def:H-dp}, but additionally taking into account the losses due to double counting of edges which are already accounted for by solution cycles that do not belong to the component $\wtilde{H}$. 
We remark that a part of these intersections are solely determined by the given $\{\tgamma_i \mid i\in [\tilde{p}]\}$, and a part of it is determined by the choice of $\tbeta_i$s by the algorithm. Our next definition captures different such intersection terms in a solution for $\Pi_{\Pcal,\alpha, h}(\gamma_i, h_i,\theta_i)_{i \in [b]}$ (which are not accounted for in our dynamic programming for $\wtilde{H}$), which is split into three types: i) determined by the chosen $\tbeta_i$s, ii) determined completely by $\tgamma_i$s and $\beta_{i'}$, where $\gamma_{i'} \notin V(\wtilde{H})$ (and is independent of the choice of the selected $\tbeta_i$s), and iii) solely determined by $\tgamma_i$s.

\begin{definition}\label{def:overflow-wt}
Consider a solution $\C{D}$ for $\Pi_{\Pcal,\alpha, h}(\gamma_i, h_i, \theta_i)_{i \in [b]}$, and for $i\in [b]$, let $\delta_i \in \C{D}$ be such that $\gamma_i$ is the $(\alpha,\C{D})$-representative for $\delta_i$. For $q\in [\tilde{p}]$, let $\tdelta_q = \delta_{q_j}$. We define the following numbers (see \Cref{fig:overflow-def}). 

    \begin{enumerate}
    \item\label{one-def:overflow-wt} For $q\in [\tilde{p}]$, let $\Gamma_{\alpha,R,\tgamma_q}^1$ be the set containing each $\gamma_i$, where: 
    i) $\gamma_i \in \{\tpsi_q, \ttau_q\}$, and 
    ii) $\theta_i = (\psi_i,\tau_i)$ and $\tgamma_q \in \Gamma_{\alpha,R,\gamma_i}[\psi_i, \tau_i]\setminus \{\psi_i, \tau_i\}$, or 
    $\theta_i ={\sf e}$. 
    Also, if $\alpha\in \{\tpsi_q,\ttau_q\}$, then add $\alpha$ to $\Gamma_{\alpha,R,\tgamma_q}^1$. Notice that if $\Gamma_{\alpha,R,\tgamma_q}^1 \neq \emptyset$, then its size can be at most $2$, and furthermore, $\tgamma_q$ must be an endpoint of the connected component $\wtilde{H}$ (in which case $\wtilde{H}$ must be a path), where $\Gamma_{\alpha,R,\tgamma_q}^1 \cap V(\wtilde{H}) =\emptyset$. Let $X = \cup_{q \in \{1,\tilde{p}\}, \gamma \in \Gamma_{\alpha,R,\tgamma_q}^1} E(\gamma)$, and $w^1(\C{D}) = w(X\setminus E(\C{P}))$. 
    
    \item\label{two-def:overflow-wt} For $q\in [\tilde{p}]$, let $\Gamma_{\alpha,R,\tgamma_q}^2$ be the set containing each $\gamma_i$, where: i) $\gamma_i \in \Gamma_{\alpha,R,\tgamma_q}[\tpsi_q, \ttau_q] \setminus \{\tpsi_q, \ttau_q\}$, and ii) $\theta_i = (\psi_i,\tau_i)$ and $\tgamma_q \in \{\psi_i, \tau_i\}$. Also, let $w^2(\C{D}) = \sum_{q\in [\tilde{p}]}\sum_{\gamma_i \in \Gamma_{\alpha,R,\tgamma_q}^2} w(E(\tgamma_q \cap \delta_i)\setminus E(\C{P}))$.
    
    \item\label{three-def:overflow-wt} For $q\in [\tilde{p}]$, let $\Gamma_{\alpha,R,\tgamma_q}^3$ be the set containing each $\gamma_i$, where: i) $\gamma_i \in \Gamma_{\alpha,R,\tgamma_q}[\tpsi_q, \ttau_q] \setminus \{\tpsi_q, \ttau_q\}$, and ii) $\theta_i = {\sf e}$, or $\tgamma_q\in \Gamma_{\alpha,R,\gamma_i}[\psi_i, \tau_i] \setminus \{\psi_i, \tau_i\}$. Also, if $\alpha\in \Gamma_{\alpha,R,\gamma_i}[\psi_i, \tau_i] \setminus \{\psi_i, \tau_i\}$, add $\alpha$ to $\Gamma_{\alpha,R,\tgamma_q}^3$. Let $w^3(\C{D}) = \sum_{q\in [\tilde{p}]}\sum_{\gamma_i\in \Gamma_{\alpha,R,\tgamma_q}^3} w(E(\tgamma_q\cap \gamma_i)\setminus E(\C{P}))$.
    \end{enumerate} 
\end{definition}

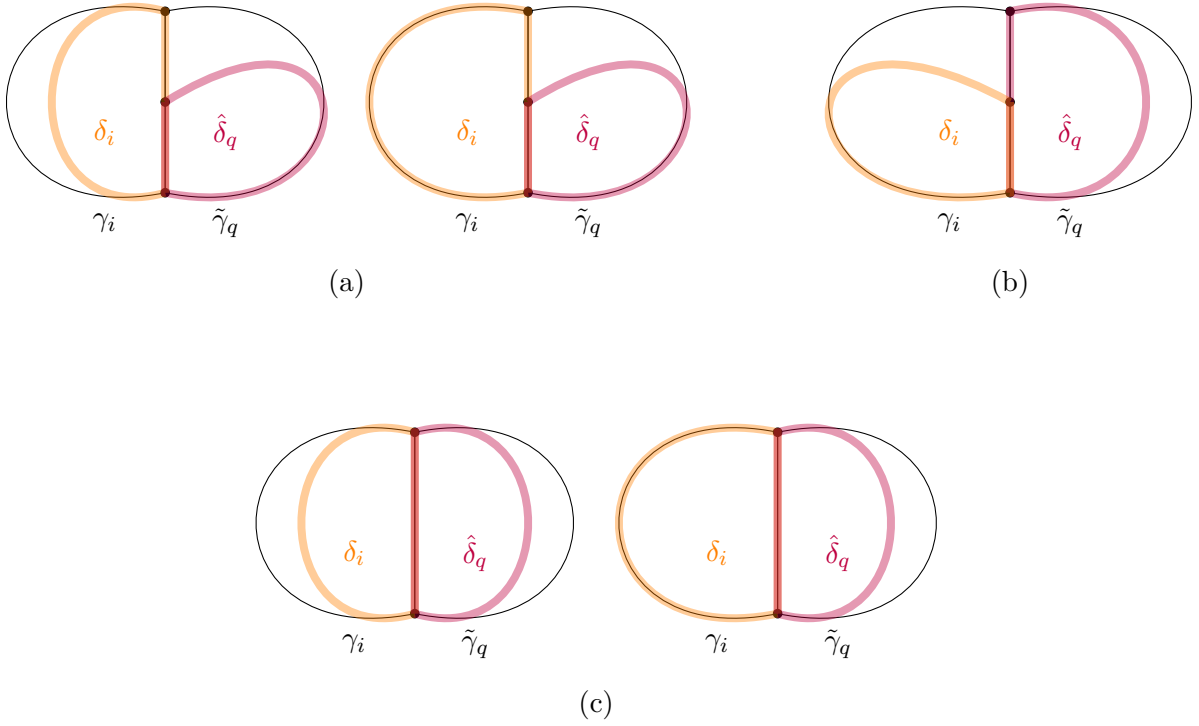
\begin{figure}[t]
\centering

\begin{subfigure}{0.60\textwidth}
\centering
\begin{tikzpicture}[scale=0.8]

    \coordinate (A) at (-3,-1.5);
    \coordinate (B) at (-3, 1.5);
    \coordinate (M1) at (-3,0); 
    \coordinate (M2) at (3,0); 

    \draw[fill] (A) circle[radius=2pt];
    \draw[fill] (B) circle[radius=2pt];
    \draw[fill] (M1) circle[radius=2pt];
    \draw[fill] (M2) circle[radius=2pt];
    \draw (A)
      .. controls (-6.5,-2.2) and (-6.5,2.2) ..
      (B) -- (A);

    \draw (A)
      .. controls (0.5,-2.2) and (0.5,2.2) ..
      (B) -- (A);

    \draw[orange, line width=3pt, opacity=0.4]
      (A)
      .. controls (-5.5,-2.2) and (-5.5,2.2) ..
      (B) -- (A);

    \draw[purple, line width=3pt, opacity=0.4]
        (M1) -- (A)
        .. controls (0.5,-2.2) and (0.5,2.2) ..
        (M1)
        -- (M1);

    \node[orange] at (-4,-0.5) {$\delta_i$};
    \node[purple] at (-2,-0.5) {$\hat{\delta}_q$};
    
    \node at (-4,-2) {$\gamma_i$};
    \node at ( -2,-2) {$\tilde{\gamma}_q$};

    \coordinate (C) at (3,-1.5);
    \coordinate (D) at (3, 1.5);
    
    \draw[fill] (C) circle[radius=2pt];
    \draw[fill] (D) circle[radius=2pt];
    
    \draw (C)
      .. controls (-0.5,-2.2) and (-0.5,2.2) ..
      (D) -- (C);

    \draw (C)
      .. controls (6.5,-2.2) and (6.5,2.2) ..
      (D) -- (C);

    \draw[orange, line width=3pt, opacity=0.4]
      (C)
      .. controls (-0.5,-2.2) and (-0.5,2.2) ..
      (D) -- (C);

    \draw[purple, line width=3pt, opacity=0.4]
        (M2) -- (C)
        .. controls (6.5,-2.2) and (6.5,2.2) ..
        (M2)
        -- (M2);

    \node[orange] at (2,-0.5) {$\delta_i$};
    \node[purple] at (4,-0.5) {$\hat{\delta}_q$};
    
    \node at (4,-2) {$\tilde{\gamma}_q$};
    \node at (2,-2) {$\gamma_i$};

    \node at (0,-3) {(a)};

\end{tikzpicture}
\end{subfigure}
\hfill
\begin{subfigure}{0.35\textwidth}
\centering
\begin{tikzpicture}[scale=0.8]

    \coordinate (A) at (0,-1.5);
    \coordinate (B) at (0,1.5);
    \coordinate (M) at (0,0);

    \draw[fill] (A) circle[radius=2pt];
    \draw[fill] (B) circle[radius=2pt];
    \draw[fill] (M) circle[radius=2pt];

    \draw (A)
      .. controls (-4,-2.2) and (-4,2.2) ..
      (B) -- (A);

    \draw (A)
      .. controls (4,-2.2) and (4,2.2) ..
      (B) -- (A);

    \draw[purple, line width=3pt, opacity=0.4]
      (A)
      .. controls (3,-2.2) and (3,2.2) ..
      (B) -- (A);

    \draw[orange, line width=3pt, opacity=0.4]
        (M) -- (A)
        .. controls (-4,-2.2) and (-4,2.2) ..
        (M)
        -- (M);
    \node[orange] at (-1,-0.5) {$\delta_i$};
    \node[purple] at (1,-0.5) {$\hat{\delta}_q$};

    \node at (-1,-2) {$\gamma_i$};
    \node at (1,-2) {$\tilde{\gamma}_q$};

    \node at (0,-3) {(b)};

\end{tikzpicture}
\end{subfigure}

\vspace{1cm}

\begin{subfigure}{0.60\textwidth}
\centering
\begin{tikzpicture}[scale=0.8]

    \coordinate (A) at (-3,-1.5);
    \coordinate (B) at (-3, 1.5);
    \draw[fill] (A) circle[radius=2pt];
    \draw[fill] (B) circle[radius=2pt];

    \draw (A)
      .. controls (-6.5,-2.2) and (-6.5,2.2) ..
      (B) -- (A);

    \draw (A)
      .. controls (0.5,-2.2) and (0.5,2.2) ..
      (B) -- (A);

    \draw[orange, line width=3pt, opacity=0.4]
      (A)
      .. controls (-5.5,-2.2) and (-5.5,2.2) ..
      (B) -- (A);
      
    \draw[purple, line width=3pt, opacity=0.4]
      (A)
      .. controls (-0.5,-2.2) and (-0.5,2.2) ..
      (B) -- (A);

    \node[orange] at (-4,-0.5) {$\delta_i$};
    \node[purple] at (-2,-0.5) {$\hat{\delta}_q$};
    
    \node at (-4,-2) {$\gamma_i$};
    \node at (-2,-2) {$\tilde{\gamma}_q$};

    \coordinate (C) at (3,-1.5);
    \coordinate (D) at (3, 1.5);
    \draw[fill] (C) circle[radius=2pt];
    \draw[fill] (D) circle[radius=2pt];

    \draw (C)
      .. controls (-0.5,-2.2) and (-0.5,2.2) ..
      (D) -- (C);

    \draw (C)
      .. controls (6.5,-2.2) and (6.5,2.2) ..
      (D) -- (C);

    \draw[orange, line width=3pt, opacity=0.4]
      (C)
      .. controls (-0.5,-2.2) and (-0.5,2.2) ..
      (D) -- (C);

    \draw[purple, line width=3pt, opacity=0.4]
      (C)
      .. controls (5.5,-2.2) and (5.5,2.2) ..
      (D) -- (C);

    \node[orange] at (2,-0.5) {$\delta_i$};
    \node[purple] at (4,-0.5) {$\hat{\delta}_q$};
    
    \node at (4,-2) {$\tilde{\gamma}_q$};
    \node at (2,-2) {$\gamma_i$};
    
    \node at (0,-3) {(c)};

\end{tikzpicture}
\end{subfigure}

\caption{The three types of overflow interactions used in the \Cref{def:overflow-wt} of $w^1(\mathcal{D})$, $w^2(\mathcal{D})$, and $w^3(\mathcal{D})$. 
The orange cycle is $\delta_i$ with representative $\gamma_i$, and the purple cycle is $\hat{\delta}_q$ with representative $\tilde{\gamma}_q$. 
(a) Type-1 interactions which are $\Gamma_{\alpha,R,\tilde{\gamma}_q}^1$, where $\gamma_i\in\{\tilde{\psi}_q,\tilde{\tau}_q\}$ and $\tilde{\gamma}_q \in \Gamma_{\alpha,R,\gamma_i}[\psi_i,\tau_i] \setminus \{\psi_i,\tau_i\}$ ($\tilde{\gamma}_q$  lies in the interior of the guessed intersection of $E(\gamma_i) \cap E(\delta_i) $) or $\theta_i={\sf e}$.
(b) Type-2 interactions which are $\Gamma_{\alpha,R,\tilde{\gamma}_q}^2$, where $\tilde{\gamma}_q\in\{\psi_i,\tau_i\}$ and $\gamma_i \in \Gamma_{\alpha,R,\tgamma_q}[\tpsi_q, \ttau_q] \setminus \{\tpsi_q, \ttau_q\}$.
(c) Type-3 interactions which are $\Gamma_{\alpha,R,\tilde{\gamma}_q}^3$, where both $\gamma_i \in \Gamma_{\alpha,R,\tgamma_q}[\tpsi_q, \ttau_q] \setminus \{\tpsi_q, \ttau_q\}$ and $\tilde{\gamma}_q \in \Gamma_{\alpha,R,\gamma_i}[\psi_i,\tau_i] \setminus \{\psi_i,\tau_i\}$ or $\theta_i={\sf e}$.}
\label{fig:overflow-def}
\end{figure}


Due to the above definition, the ``best'' choice of $(\varepsilon_1,\varepsilon_{\tilde{p}}) \in \tcal{A}_1 \times \tcal{A}_{\tilde{p}}$ is independent on its second and third item. Note the first term of the definition is applicable only when $\wtilde{H}$ is a path. Another term that our dynamic programming has not accounted for, is when $\wtilde{H}$ is a cycle, the edges (outside of $\C{P}$) common to $\tbeta_1$ and $\tbeta_{\tilde{p}}$ (see \Cref{def:H-dp}). The above motivates our next definition.  

For $\varepsilon_1 = (u_1,v_1) \in \C{\wtilde{A}}_1$ and $\varepsilon_{\tilde{p}} = (u_{\tilde{p}},v_{\tilde{p}}) \in \C{\wtilde{A}}_{\tilde{p}}$, we define the {\em overcounted weight}, $w^{\sf ovr}(\varepsilon_1, \varepsilon_{\tilde{p}})$ as follows. First let us consider the case when $\wtilde{H}$ is a cycle (which must have at least $3$ vertices from \Cref{obs:h-cycle-path}). In this case, we set $w^{\sf ovr}(\varepsilon_1, \varepsilon_{\tilde{p}})$ as the weight of the common edges (outside of $\C{P}$) in these two cycles. That is, if $\wtilde{H}$ is a cycle then set $w^{\sf ovr}(\varepsilon_1, \varepsilon_{\tilde{p}}) = w(E(\tgamma_1[u_1,v_1]\cap \tgamma_{\tilde{p}}[u_{\tilde{p}},v_{\tilde{p}}])\setminus E(\C{P}))$ (see \Cref{lem:intersection-with-ancestor-is-more}). 

We now consider the case when $\wtilde{H}$ is a path. Intuitively speaking, we construct $w^{\sf ovr}(\varepsilon_1, \varepsilon_{\tilde{p}})$ in accordance with \Cref{def:overflow-wt} (\ref{one-def:overflow-wt}) as follows. For $q\in [\tilde{p}]$, let $\Gamma_{\alpha,R,\tgamma_q}^1$ be the set containing each $\gamma_i$, where: i) $\gamma_i \in \{\tpsi_q, \ttau_q\}$, and ii) $\theta_i ={\sf e}$, or $\theta_i = (\psi_i,\tau_i)$ and $\tgamma_q \in \Gamma_{\alpha,R,\gamma_i}[\psi_i, \tau_i]\setminus \{\psi_i, \tau_i\}$. Also, if $\alpha\in \{\tpsi_q,\ttau_q\}$, then add $\alpha$ to $\Gamma_{\alpha,R,\tgamma_q}^1$. Let $X = \cup_{q \in \{1,\tilde{p}\}, \gamma \in \Gamma_{\alpha,R,\tgamma_q}^1} E(\gamma)$, and $w^{\sf ovr}(\varepsilon_1, \varepsilon_{\tilde{p}}) = w(X\setminus E(\C{P}))$. 


We are ready to pick a pair from $\tcal{A}_1\times \tcal{A}_{\tilde{p}}$, which we will prove to be the ``best'' choice.

\begin{align*}
(\varepsilon^*_1,\varepsilon^*_{\tilde p})
=
\argmax_{\substack{
    \varepsilon_1 \in \tcal{A}_1,\;
    \varepsilon_{\tilde p}=(u_{\tilde p},v_{\tilde p}) \in \tcal{A}_{\tilde p}: \\
    \tilde p =1\;\text{or}\;
(\varepsilon_1,\varepsilon_{\tilde p}) \text{ is } d\text{-diverse}
}}
\Big(
    \Lambda_{\varepsilon_1,\varepsilon_{\tilde p}}
    (\tilde p,u_{\tilde p},v_{\tilde p})
    - w^{\sf ovr}(\varepsilon_1,\varepsilon_{\tilde p})
\Big).
\end{align*}

In the above, if the maximum is $-\infty$ (due to empty maximum, or for some other reason), then we take $(\varepsilon^*_1,\varepsilon^*_{\tilde{p}}) = (\phi,\phi)$. The next observation handles the case when $(\varepsilon^*_1,\varepsilon^*_{\tilde{p}}) = (\phi,\phi)$.

\begin{observation}\label{obs:e1ep-phi-phi}
If $(\varepsilon_1^*, \varepsilon_{\tilde{p}}^*) =  (\phi,\phi)$, then there is no solution for $\Pi_{\Pcal,\alpha, h}(\gamma_i, h_i, \theta_i)_{i \in [b]}$.
\end{observation}
\begin{proof}
We will prove the contrapositive of the required statement. Suppose that there is a solution $\cal{C}$ for $\Pi_{\Pcal,\alpha, h}(\gamma_i, h_i, \theta_i)_{i \in [b]}$, where for $i \in [b]$, $\tgamma_i$ is the $(\alpha, \cal{C})$-\pdr{} of $\hbeta_i \in \cal{C}$. To prove the statement it will be enough for us to construct $\varepsilon_1 = (u_1,v_1) \in \tcal{A}_1$ and $\varepsilon_{\tilde{p}} = (u_{\tilde{p}},v_{\tilde{p}}) \in \tcal{A}_{\tilde{p}}$ such that $(\varepsilon_1,\varepsilon_{\tilde{p}})$ is $d$-diverse when $\tilde{p} \geq 2$, and furthermore, $\Lambda_{\varepsilon_1,\varepsilon_{\tilde{p}}}(\tilde{p},u_{\tilde{p}},v_{\tilde{p}}) \neq -\infty$. To this end, for $i\in \{1, \tilde{p}\}$, let $u_i, v_i \in V(\tgamma_i)$ be distinct vertices such that $\hbeta_i\cap \tgamma_i = \tgamma_i[u_i,v_i]$. Notice that such vertices exist as $\tgamma_1,\tgamma_{\tilde{p}} \in R_{\sf o}$; see \Cref{def:solution-of-pi-alpha-h-theta} (\ref{subcondition-3-of-condition-3-of-def:solution-of-pi-alpha-h-theta}). As $\hcal{C}$ is a solution for $\Pi_{\Pcal,\alpha, h}(\gamma_i, h_i, \theta_i)_{i \in [b]}$, if $\tilde{p} \geq 2$, it must hold that $w(\tgamma_1[u_1,v_1]\cap \tgamma_{\tilde{p}}[u_{\tilde{p}}, v_{\tilde{p}}]) \leq w(\hbeta_1\cap \hbeta_{\tilde{p}}) \leq d$, and thus $(\varepsilon_1, \varepsilon_{\tilde{p}})$ must be $d$-diverse. Furthermore, for each $i\in \{1,\tilde{p}\}$, as $\hbeta_i$ must be $\ttheta_i$-abiding and $d$-good (see \Cref{def:path-comp-theta,def:path-comp-theta-good}), $\varepsilon_i = (u_i,v_i)\in \tcal{A}_i$. Also, by the construction of $\varepsilon_1, \varepsilon_{\tilde{p}}$ and as $\hcal{C}$ is a solution for $\Pi_{\Pcal,\alpha, h}(\gamma_i, h_i, \theta_i)_{i \in [b]}$, it follows that $\{\hbeta_i\mid i\in [\tilde{b}]\}$ satisfies all the conditions of \Cref{def:H-dp} for $\varepsilon_1, \varepsilon_{\tilde{p}}$, $\tilde{p}$ and $(u_{\tilde{p}},v_{\tilde{p}}) \in \tcal{A}_{\tilde{p}}$. Furthermore, for each $i\in [\tilde{p}]$, $\Pi_{\C{P}}(\hbeta_i, \wth_i) \neq -\infty$. Thus we can obtain that $\Lambda_{\varepsilon_1,\varepsilon_{\tilde{p}}}(\tilde{p},u_{\tilde{p}},v_{\tilde{p}}) \geq \sum_{i\in [\tilde{p}]}\Pi_{\C{P}}(\hbeta_i, \wth_i)  - \sum_{i \in [q-1]}{\zeta_{\hbeta_i, \hbeta_{i+1}}} \neq -\infty$. This concludes the proof. 
\end{proof}

Due to the above observation, we now assume that $(\varepsilon^*_1,\varepsilon^*_{\tilde{p}}) \neq (\phi,\phi)$, and suppose that $\varepsilon^*_1 = (u_1,v_1)$ and $\varepsilon^*_{\tilde{p}} = (u_{\tilde{p}},v_{\tilde{p}})$ such that $\Lambda_{\varepsilon^*_1,\varepsilon^*_{\tilde{p}}}(\tilde{p},u_{\tilde{p}},v_{\tilde{p}}) \neq -\infty$. Let $\wtilde{\C{C}} = \{\tilde{\iota}_i \mid i \in [\tilde{p}]\}$ be the set returned by \Cref{lem:H-rec-correct}, for $\varepsilon^*_1,\varepsilon^*_{\tilde{p}}$, $\tilde{p}$, and $(u_{\tilde{p}},v_{\tilde{p}})$, where $\tilde{\iota}_i$ is a strict descendant of $\tgamma_i$, for each $i\in [\tilde{p}]$. For each $i\in [\tilde{p}]$, we set $\tbeta_i = \tilde{\iota}_i$. 
Also, for $i\in [\tilde{p}]$, let $\tcal{C}_i$ be an optimal solution for $\Pi_{\C{P}}(\tbeta_i, \wth_i)$, which must exists as $\Lambda_{\varepsilon^*_1,\varepsilon^*_{\tilde{p}}}(\tilde{p},u_{\tilde{p}},v_{\tilde{p}}) \neq -\infty$. To prove \Cref{lem:compute-pi-ext}, we will establish some smaller results, for which we introduce the notations that we will follow for the next two lemmas.
Suppose that $\hcal{C}$ is an optimal solution for $\Pi_{\Pcal,\alpha, h}(\gamma_i, h_i, \theta_i)_{i \in [b]}$, where for $i\in [\tilde{p}]$, $\tgamma_i$ is the $(\alpha,\hcal{C})$-representative for $\hbeta_{i} \in \hcal{C}$, and $\hcal{C}_i = \hcal{C} \cap V(\C{T}_{\hbeta_{i}})$. We let $\C{C} = (\hcal{C} \setminus (\cup_{i\in [\tilde{p}]} \hcal{C}_i) ) \cup (\cup_{i\in [\tilde{p}]} \tcal{C}_i)$. 
For each $i\in [b]$, let $\bbeta_{i} \in \C{C}$ be such that $\gamma_i$ is the $(\alpha,\C{C})$-representative for $\bbeta_i$, and $\C{C}_i = \C{C} \cap V(\C{T}_{\bbeta_i})$. 

\begin{lemma}\label{obs-H-path-d-div}
$\C{C}$ is a solution for $\Pi_{\Pcal,\alpha, h}(\gamma_i, h_i, \theta_i)_{i \in [b]}$.
\end{lemma}
\begin{proof}
As $\hcal{C}$ is a solution for $\Pi_{\Pcal,\alpha, h}(\gamma_i, h_i, \theta_i)_{i \in [b]}$, and for each $i\in [b]$, $\tcal{C}_i$ is a solution for $\Pi_{\C{P}}(\tbeta_i,\wtilde{h}_i)$, we can obtain that for each $C\in \C{C}$ and $P\in \C{P}$, $w(C\cap P) \leq d$. Furthermore, if we prove that cycles in $\C{C}$ are $d$-diverse, then this together with the construction of $\C{C}$ and \Cref{def:H-dp} (\ref{item-1-of-def:H-dp}) will imply that $\C{C}$ is a solution for $\Pi_{\Pcal,\alpha, h}(\gamma_i, h_i, \theta_i)_{i \in [b]}$. Thus, we will now focus on proving that cycles in $\C{C}$ are $d$-diverse.

For each $i\in [\tilde{p}]$, as $\tcal{C}_i$ is a solution for $\Pi_{\C{P}}(\tbeta_i, \wth_i)$, it must be the case that cycles in $\tcal{C}_i$ are $d$-diverse (see \Cref{def:pi-p-alpha-h}). Also, as $\hcal{C}$ is a solution for $\Pi_{\Pcal,\alpha, h}(\gamma_i, h_i, \theta_i)_{i \in [b]}$, for distinct $i,i' \in [b] \setminus \{i_j \mid i\in [\tilde{p}]\}$, and any $C\in \C{C}_i , C'\in \C{C}_{i'}$, it must be the case that $w(C\cap C') \leq d$. We will now consider $i \in [\tilde{p}]$, $i' \in [b]\setminus \{i_j\}$, and cycles $C\in \tcal{C}_i$ and $C' \in \C{C}_{i'}$, and our goal will be to argue that $w(C\cap C') \leq d$. Note once we establish the above, it will imply that $\C{C}$ is $d$-diverse. For the remainder of the proof we will be relying on \Cref{lem:intersection-with-ancestor-is-more} without mentioning it every time, in particular we will use that, for any distinct $j_1,j_2\in [b]$, $C_1\in \C{C}_{\bbeta_{j_1}}$ and $C_2 \in \C{C}_{\bbeta_{j_2}}$, $E(C_1\cap C_2) \subseteq E(\bbeta_{j_1}\cap \bbeta_{j_2})$. Due to the above, to prove the required statement it will be enough for us to argue that $w(\tbeta_i \cap \bbeta_{i'}) \leq d$. If $E(\tgamma_i\cap \gamma_{i'}) =\emptyset$, then it directly follows that $w(\tbeta_i \cap \beta_{i'}) \leq d$. We will now assume that $E(\tgamma_i\cap \gamma_{i'}) \neq\emptyset$, and let $u,v$ be distinct vertices such that $\tgamma_i[u,v] =\gamma_{i'}[u,v]=\tgamma_i\cap \gamma_{i'}$. Consider the following cases: 

\begin{itemize}
\item Consider the case when $\gamma_{i'} \notin \Gamma_{\alpha,R,\tgamma_i}[\tpsi_i,\ttau_i]$. By the construction of $\tbeta_i$, we have that $\tbeta_i$ is $\ttheta_i$-abiding (see \Cref{def:path-comp-theta}), and thus, $E(\tbeta_i\cap \tgamma_i[u,v]) =\emptyset$. However, $E(\bbeta_{i'}\cap \tgamma_i) \subseteq E(\tgamma_i[u,v])$, and thus we can conclude that $E(\tbeta_i\cap \bbeta_{i'})=\emptyset$. Hence, $w(\tbeta_i\cap \bbeta_{i'})\leq d$. 

\item Now consider the case when $\theta_{i'} = {\sf n}$, or $\gamma_{i'}\in R_{\sf o}$ and $\tgamma_{i} \notin \Gamma_{\alpha,R,\gamma_{i'}}[\psi_{i'},\tau_{i'}]$. As $\bbeta_{i'}$ must be $\theta_{i'}$-abiding, $E(\bbeta_{i'}\cap \gamma_{i'}[u,v]) =\emptyset$. However, $E(\bbeta_{i'}\cap \tgamma_i) \subseteq E(\tgamma_i[u,v])$, and thus we can conclude that $E(\tbeta_i\cap \bbeta_{i'})=\emptyset$. Hence, $w(\tbeta_i\cap \bbeta_{i'})\leq d$. 

\item Consider the case when $\gamma_{i'} \in \Gamma_{\alpha,R,\tgamma_i}[\tpsi_i,\ttau_i]\setminus \{\tpsi_i,\ttau_i\}$, and one of the following holds: i) $\theta_{i'}={\sf e}$, or ii) $\gamma_{i'}\in R_{\sf o}$ and $\tgamma_i \in \Gamma_{\alpha,R,\gamma_{i'}}[\psi_{i'},\tau_{i'}]\setminus \{\psi_{i'},\tau_{i'}\}$. Then, as $\hbeta_i$ and $\bbeta_{i'}$ must be $\ttheta_i$- and $\theta_{i'}$-abiding, respectively, we can obtain that $E(\hbeta_i \cap \bbeta_{i'}) = E(\tgamma_i\cap \gamma_{i'})$. Thus we can conclude that $w(\tgamma_i\cap \gamma_{i'})\leq d$, which implies that $w(\tbeta_i\cap \bbeta_{i'})\leq d$.
\end{itemize}
If none of the above cases are applicable, it must be the case that i) $\gamma_{i'} \in \Gamma_{\alpha,R,\tgamma_i}[\tpsi_i,\ttau_i]$, and ii) $\theta_{i'} ={\sf e}$, or $\theta_{i'} = (\psi_{i'}, \tau_{i'})$, and $\tgamma_{i} \in \Gamma_{\alpha,R,\gamma_{i'}}[\psi_{i'},\tau_{i'}]$. Furthermore, one of $\gamma_{i'} \in \{\tpsi_i,\ttau_i\}$ or $\tgamma_{i} \in \{\psi_{i'},\tau_{i'}\}$ holds. 

\begin{itemize}
\item Suppose that $\gamma_{i'} \in \{\ttau_i, \tpsi_i\}$ and $\tgamma_{i} \in \{\tau_{i'}, \tpsi_{i'}\}$. In this case, $\{\gamma_{i'}, \tgamma_i\} \in E(H)$, and thus, $\gamma_{i'}, \tgamma_{i}$ are consecutive vertices of the connected component of $\wtilde{H}$. Then, by \Cref{def:H-dp} (\ref{item-3-of-def:H-dp}) it follows that $w(\tbeta_i \cap \bbeta_{i'}) \leq d$. 

\item Now consider the case when $\gamma_{i'} \in \{\tpsi_i,\ttau_i\}$ and one of the following holds: i) $\theta_{i'} = {\sf e}$, or $\tgamma_{i} \in \Gamma_{\alpha,R,\gamma_{i'}}[\psi_{i'},\tau_{i'}] \setminus \{\psi_{i'},\tau_{i'}\}$. As $\tbeta_i$ must be $d$-good (see \Cref{def:path-comp-theta-good,def:H-dp} (\ref{item-1-of-def:H-dp})), it follows that $w(\tbeta_i\cap \bbeta_{i'})\leq d$. 

\item Finally consider the case when $\gamma_{i'} \in \Gamma_{\alpha,R,\tgamma_i}[\tpsi_i,\ttau_i] \setminus \{\tpsi_i,\ttau_i\}$ and $\tgamma_{i} \in \{\psi_{i'},\tau_{i'}\}$. Then, as $\hcal{C}$ is a solution we have that $w(\hbeta_i\cap \bbeta_{i'}) \leq d$. Also, since $\hbeta_i$ and $\tbeta_i$ are $\ttheta_i$-abiding, we must have $E(\hbeta_i \cap \bbeta_{i'}) \subseteq E(\tgamma_i \cap \gamma_{i'}) = E(\hbeta_i \cap \gamma_{i'}) = E(\tbeta_i \cap \gamma_{i'})$. Thus, we can obtain that $E(\hbeta_i \cap \bbeta_{i'}) = E(\tbeta_i \cap \bbeta_{i'})$, and hence, $w(\tbeta_i\cap \bbeta_{i'})\leq d$.
\end{itemize}
This concludes the proof.  
\end{proof}

We will now prove that $\C{C}$ is at least as good a solution as $\hcal{C}$, for $\Pi_{\Pcal,\alpha, h}(\gamma_i, h_i, \theta_i)_{i \in [b]}$. 
 
\begin{lemma}\label{lem:dp-H-wt-better}
We have $w(\C{C}\cup \C{P})\geq w(\hcal{C}\cup \C{P})$.
\end{lemma}
\begin{proof} 
We recall the following property which we will be using in our proof. Consider distinct $i_1,i_2,i_3 \in [b]$ and $\C{D} \in \{\C{C}, \hcal{C}\}$. Also consider $C_1 \in \C{D}\cap V(\C{T}_{\gamma_{i_1}})$, $C_2 \in \C{D}\cap V(\C{T}_{\gamma_{i_2}})$, and $C_3 \in \C{D}\cap V(\C{T}_{\gamma_{i_3}})$. From \Cref{lem:intersection-with-ancestor-is-more} and \Cref{obs:split-relation-between-pd-and-pdrs} we have that $E(C_1)\cap E(C_2) \cap E(C_3) = \emptyset$ and $E(C_1\cap C_2) \subseteq E(\delta_{i_1}\cap \gamma_{i_1})$, where $\gamma_{i_1}$ is the $(\alpha,\C{D})$-representative for $\delta_{i_1}\in \C{D}$. In the above, the first part will be used in the arguments that an edge cannot be part of three distinct cycles from $\{\tbeta_i\mid i\in [\tilde{p}]\}$ (resp. $\{\hbeta_i\mid i\in [\tilde{p}]\}$). 

Let $I = [b] \setminus \{i_j \mid i\in [\tilde{p}]\}$, and $\C{C}_{\sf cm} = \{\alpha\} \cup (\cup_{q \in I} \C{C}_i)$. Note that $\C{C}_{\sf cm}=\{\alpha\} \cup (\cup_{q \in I} \hcal{C}_i)$. Let $\tcal{C}_{\sf ex} = \cup_{i\in [\tilde{p}]} \tcal{C}_i$ and $\hcal{C}_{\sf ex} = \cup_{[i\in \tilde{p}]} \hcal{C}_i$. Notice that $\C{C} = \C{C}_{\sf cm} \uplus \tcal{C}_{\sf ex}$ and $\hcal{C} = \C{C}_{\sf cm} \uplus \hcal{C}_{\sf ex}$. Let $w_{\sf cm} = w(E(\C{C}_{\sf cm}) \setminus E(\C{P}))$. Also,
\begin{align}
w(\C{C}\cup \C{P}) & = w(\tcal{C}_{\sf ex} \cup \C{P}) -  w\big((E(\tcal{C}_{\sf ex}) \cap E(\C{C}_{\sf cm})) \setminus E(\C{P}) \big) + w_{\sf cm} \notag\\
w(\hcal{C}\cup \C{P}) & = w(\hcal{C}_{\sf ex} \cup \C{P}) -  w\big((E(\hcal{C}_{\sf ex}) \cap E(\C{C}_{\sf cm})) \setminus E(\C{P}) \big) + w_{\sf cm}\label{eq:wt-diff-one}
\end{align}

As $\hcal{C}$ is a solution for $\Pi_{\Pcal,\alpha, h}(\gamma_i, h_i, \theta_i)_{i \in [b]}$, using the property we stated towards the beginning of the proof of this lemma, and the construction of $\tcal{C}_i$, for $i\in [\tilde{p}]$, we can obtain: 
\begin{align}
w(\tcal{C}_{\sf ex} \cup \C{P}) & = \sum_{i\in [\tilde{p}]} \Pi_{\C{P}}(\tbeta_i, \wtilde{h}_i) - (\tilde{p}-1)w(\C{P}) - \sum_{i,i' \in [\tilde{p}], i< i'} w(E(\tbeta_i\cap \tbeta_{i'}) \setminus E(\C{P})) \notag\\
w(\hcal{C}_{\sf ex} \cup \C{P}) & \leq \sum_{i\in [\tilde{p}]} \Pi_{\C{P}}(\hbeta_i, \wtilde{h}_i) - (\tilde{p}-1)w(\C{P}) - \sum_{i,i' \in [\tilde{p}], i< i'} w(E(\hbeta_i\cap \hbeta_{i'})\setminus E(\C{P}))\label{eq:wt-diff-two}
\end{align}

From \Cref{def:overflow-wt} we can obtain the following:

\begin{align}
w\big((E(\tcal{C}_{\sf ex}) \cap E(\C{C}_{\sf cm})) \setminus E(\C{P}) \big) + \sum_{i,i'\in [\tilde{p}], i< i'} w(E(\tbeta_i\cap \tbeta_{i'}) \setminus E(\C{P})) & \notag
\end{align}
\begin{align}
\hspace{10mm}& = w^{\sf ovr}(\varepsilon^*_1, \varepsilon^*_{\tilde{p}}) + w^2(\C{C}) +w^3(\C{C}) + \sum_{i\in [\tilde{p}-1]} w(E(\tbeta_i\cap \tbeta_{i+1}) \setminus E(\C{P})) \notag
\end{align}
\begin{align}
w\big((E(\hcal{C}_{\sf ex}) \cap E(\C{C}_{\sf cm})) \setminus E(\C{P}) \big) + \sum_{i,i'\in [\tilde{p}], i< i'} w(E(\hbeta_i\cap \hbeta_{i'}) \setminus E(\C{P})) & \notag
\end{align}
\begin{align}
\hspace{10mm}& = w^{\sf ovr}(\varepsilon_1, \varepsilon_{\tilde{p}}) + w^2(\hcal{C}) +w^3(\hcal{C}) + \sum_{i\in [\tilde{p}-1]} w(E(\hbeta_i\cap \hbeta_{i+1}) \setminus E(\C{P})) 
\label{eq:wt-diff-three}
\end{align}

Also, from \Cref{def:overflow-wt} (\ref{two-def:overflow-wt} and \ref{three-def:overflow-wt}) we can obtain that $w^2(\C{C}) =w^2(\hcal{C})$ and $w^3(\C{C}) =w^3(\hcal{C})$. Using the above combined with \Cref{eq:wt-diff-one}--\Cref{eq:wt-diff-three}, we can obtain that:  \\

$w(\C{C}\cup \C{P}) - w(\hcal{C}\cup \C{P})$
$$\geq \big(\sum_{i\in [\tilde{p}]} \Pi_{\C{P}}(\tbeta_i, \wtilde{h}_i) -\sum_{i\in [\tilde{p}-1]} \zeta_{\tbeta_i,\tbeta_{i+1}} - w^{\sf ovr}(\varepsilon^*_1, \varepsilon^*_{\tilde{p}}) \big) - \big( \sum_{i\in [\tilde{p}]} \Pi_{\C{P}}(\hbeta_i, \wtilde{h}_i) -\sum_{i\in [\tilde{p}-1]} \zeta_{\hbeta_i,\hbeta_{i+1}} - w^{\sf ovr}(\varepsilon_1, \varepsilon_{\tilde{p}}) \big)$$

Moreover, as $\{\hbeta_i \mid i\in [\tilde{p}]\}$ must satisfy all the requirements of \Cref{def:H-dp} (but may or may not achieve the maximum), by the choice of $(\varepsilon^*_1,\varepsilon^*_{\tilde{p}})$ we can further obtain that:
$$\big(\sum_{i\in [\tilde{p}]} \Pi_{\C{P}}(\tbeta_i, \wtilde{h}_i) -\sum_{i\in [\tilde{p}-1]} \zeta_{\tbeta_i,\tbeta_{i+1}} - w^{\sf ovr}(\varepsilon^*_1, \varepsilon^*_{\tilde{p}}) \big) - \big( \sum_{i\in [\tilde{p}]} \Pi_{\C{P}}(\hbeta_i, \wtilde{h}_i) -\sum_{i\in [\tilde{p}-1]} \zeta_{\hbeta_i,\hbeta_{i+1}} - w^{\sf ovr}(\varepsilon_1, \varepsilon_{\tilde{p}}) \big) \geq 0$$

Thus, $w(\C{C}\cup \C{P}) - w(\hcal{C}\cup \C{P}) \geq 0$. 
\end{proof}

We summarize the results obtained using the previous three results in the following corollary.

\begin{corollary}\label{lem:correct-beta-i-n-when-p-is-1}
The following holds:
\begin{enumerate}
\item 
If $(\varepsilon^*_1,\varepsilon^*_{\tilde{p}}) = (\phi,\phi)$, then there is no solution for $\Pi_{\Pcal,\alpha, h}(\gamma_i, h_i, \theta_i)_{i \in [b]}$.  

\item
Suppose there is an (optimal) solution $\hcal{C}$ for $\Pi_{\C{P}}(\alpha,h)$, where for $i\in [\tilde{p}]$, we have $\hcal{C}_i = \hcal{C} \cap V(\C{T}_{\tgamma_i})$. Then there is an (optimal) solution $\C{C}$ for $\Pi_{\Pcal,\alpha, h}(\gamma_i, h_i, \theta_i)_{i \in [b]}$, where for $i\in [\tilde{p}]$, $\tcal{C}_i = \C{C} \cap V(\C{T}_{\tgamma_i})$, and the following holds:  i) $\C{C} \setminus (\cup_{i\in [\tilde{p}]} \tcal{C}_i) = \hcal{C} \setminus (\cup_{i\in [\tilde{p}]} \hcal{C}_i)$, ii) for $i\in [\tilde{p}]$, $\tgamma_i$ is the $(\alpha, \C{C})$ representative for $\tbeta_i \in \C{C}$, and iii) $w(\C{C} \cup\C{P}) \geq w(\hcal{C}\cup\C{P})$. 
\end{enumerate} 
\end{corollary}
We are ready to prove \Cref{lem:compute-pi-ext}.
\begin{proof}[Proof of \Cref{lem:compute-pi-ext}]
This follows from \Cref{lem:correct-beta-i-n} and \Cref{lem:correct-beta-i-n-when-p-is-1}.
\end{proof}

We have proved \Cref{lem:compute-pi-ext}, which completes the proof of \Cref{lem:compute-pi}, which in turn finishes the proof of \Cref{thm:maxddisjointcoverage}.

\section{Conclusion and Open Questions}\label{conclusion}
In this paper we resolved two open questions from~\cite{psc,DBLP:journals/talg/BentertFGKLPRSS26}  on packing shortest cycles. 
First, we showed that \dsce is \WOH parameterized by $k$ already on unweighted subcubic graphs, ruling out \classFPT algorithms for the problem on general graphs. 
Second, we  obtained a single-exponential algorithm for \maxddisjointcoveragefull on planar graphs. Several questions remain open.

\medskip\noindent\textbf{Kernelization.}
First, does \maxddisjointcoveragefull admit a polynomial-in-$k$ 
kernel on planar graphs? We are able to establish kernelization for some special 
cases, but the general case remains open. More precisely, for planar graphs, 
 ~\cite{psc,DBLP:journals/talg/BentertFGKLPRSS26} show that \dsce 
admits a kernel with $\Ocal(k^2)$ vertices. By making use of a very similar 
approach, we are able to show that the special case of \maxddisjointcoveragefull 
without constraints on covering $\ell$ --- that is, the task of finding $k$ 
$d$-diverse shortest cycles --- admits a kernel with $\Ocal(k^2)$ vertices on 
planar graphs. However, we could not obtain a similar result for the coverage 
variant, even without constraints on diversity; that is, the problem of finding 
$k$ shortest cycles covering the edge set of maximum total weight. Obtaining a polynomial kernel for coverage problem on planar graphs is the first natural step towards the kernel for  \maxddisjointcoveragefull. 

\medskip\noindent\textbf{Subexponential algorithms.}
Many cycle- and packing-related problems on planar graphs admit subexponential-time algorithms, yet we do not know whether this holds for problems around packing of shortest cycles. 
Proving a lower bound that rules out $2^{o(k)} \cdot n^{O(1)}$ algorithms under ETH, or obtaining a subexponential algorithm for \maxddisjointcoveragefull and \dsce is another interesting open question.

\medskip\noindent\textbf{Sparse classes of graphs.}
We establish that \dsce is \WOH even on subcubic graphs, so the problem remains 
intractable on sparse graphs in general. However, it is plausible that the problem 
could be \classFPT on sparse graphs of specific structure, such as $H$-minor-free 
graphs or graphs of bounded expansion. Our methods are tailored to planar graphs, 
and different techniques would most likely be needed here.

\bibliography{book_pc}
\end{document}